\documentclass[aps,prx,reprint]{revtex4-2}

\usepackage{graphicx}
\usepackage{amsfonts,amssymb,amsmath}
\usepackage{amsthm}
\usepackage{hyperref}
\hypersetup{
  pdftitle={Protecting Astronomical Interferometry through Quantum-Memory Scrambling},
  pdfauthor={Jianqi Sheng},
  bookmarksnumbered=true,
  bookmarksopen=true,
  bookmarksdepth=3
}

\theoremstyle{plain}
\newtheorem{theorem}{Theorem}
\newtheorem{lemma}{Lemma}
\newtheorem{corollary}{Corollary}
\newtheorem{remark}{Remark}

\begin{document}

\title{Protecting Astronomical Interferometry through Quantum-Memory Scrambling}

\author{Jianqi Sheng}
\affiliation{Department of Physics, City University of Hong Kong, Hong Kong}

\begin{abstract}
Preserving the complex visibility of coherently captured starlight is essential
for quantum-assisted long-baseline interferometry, because this nonlocal
coherence carries the spatial information needed to form astronomical images.
Yet finite memory lifetime and imperfect retrieval inevitably produce storage
failures; when a failed cell is heralded, its erased subsystem can leak
which-node information to the environment and dephase the stored coherence.
Strict finite-dimensional (\(U(1)\)) covariance further forbids uniform exact
correction of all single-memory erasures. We address this combined physical
and symmetry-imposed limit using local number-conserving quantum scrambling,
parameter-independent pattern-conditioned recovery, and a fixed
Gottesman--Jennewein--Croke receiver. We prove a
channel-to-Fisher-information stability theorem that converts approximate
logical recovery into an operational guarantee on receiver-accessible
information over compact visibility regions. All-pattern finite-size
simulations show that scrambling redistributes erasure risk and suppresses
high-leakage events, while a separate equal-budget comparison identifies
a shallow design that outperforms the tested deeper and charge-sector Haar-random benchmarks
at the prespecified operating point. A separate compilation resolves every
nontrivial recovery branch into abstract nearest-neighbor number-conserving
gates. Although the present low-rate design does not yet improve
fixed-total-memory throughput, it establishes a blueprint for converting spare
memory capacity into protection of astronomical coherence, opening a path
toward higher-rate, erasure-resilient quantum telescope architectures.
\end{abstract}

\maketitle

\section{Introduction}

Long-baseline interferometry converts the separation between optical
receivers into angular resolution.  For a spatially incoherent source, the
van Cittert--Zernike theorem identifies the complex mutual coherence measured
on a baseline with a normalized Fourier component of the source intensity.
Extending the baseline therefore accesses finer spatial structure without
enlarging any individual aperture.  At optical frequencies, however, the
information-bearing field is typically weak, with far below one photon per
spatiotemporal mode.  Direct beam combination preserves the corresponding
single-photon coherence, but requires the collected fields to be transported
through a low-loss, phase-stable optical channel.  Independent local
measurements avoid that transport but, in the weak-source regime, cannot in
general attain the leading-order information available to a nonlocal
measurement~\cite{Tsang2011,Gottesman2012}.  The same property that makes a
long baseline valuable thus makes an optical array difficult to scale.

Quantum networks offer a different way to realize the required nonlocal
measurement.  In the Gottesman--Jennewein--Croke (GJC) protocol, a
pre-distributed single-photon-entangled state acts as a nonlocal phase
reference, allowing the astronomical modes to be interfered only with local
ancillas rather than transported across the baseline~\cite{Gottesman2012}.
Subsequent proposals reduced the entanglement cost by coherently recording and
compressing photon-arrival information in distributed memories~\cite{Khabiboulline2019PRL,Khabiboulline2019PRA}, protected captured
starlight with conventional quantum error-correcting codes~\cite{Huang2022},
or avoided long-lived signal memories through larger linear-optical resource
states~\cite{Marchese2023}.  Continuous-variable teleportation provides a
complementary route beyond the vacuum--single-photon approximation~\cite{Wang2025CV}, while spatial-mode sorting and entanglement-assisted mode
mixing extend the objective from phase scanning to quantum-limited
multimode imaging~\cite{Padilla2026PRA,Padilla2026PRL}.  These architectures
also make clear that performance must be assessed under locality,
the photon-number superselection rule (SSR), which prevents local access
to coherence between different photon-number sectors unless a phase reference
is supplied, and an explicit measurement model; unconstrained quantum Fisher
information (QFI) alone is not an operational
comparison~\cite{Zhang2025}.

The underlying ingredients have now entered experiments.  A path-entangled
single photon has been used as a shared reference to reconstruct a
pseudothermal source on a laboratory baseline~\cite{Brown2023}.  Networked
silicon-vacancy (SiV) centers have subsequently combined event-ready remote
entanglement, signal-photon storage, mode erasure, and nonlocal heralding over
a $1.55\,\mathrm{km}$ fiber separation~\cite{Stas2026}.  In a distinct
implementation, cold-atom memories stored the ancillary Fock-state
entanglement that was later retrieved for GJC interference with a thermal
field, with entanglement distributed over a $20\,\mathrm{km}$ fiber link and
memory-assisted compensation of an equivalent $1.5\,\mathrm{km}$ geometric
delay~\cite{Wang2026}.  These advances expose two different roles of a quantum
memory: it may bank a replaceable ancillary reference, or it may carry the
unknown astronomical coherence after the signal has interacted with the
network.  Only in the latter case do subsequent memory errors act directly on
the parameter-bearing state.  For this signal-storage role, designed quantum
error correction (QEC) has already shown that coherently captured starlight can
be protected against specified storage and processing noise, including
known-location erasures, and can exhibit large-code threshold behavior~\cite{Huang2022}.  This established signal-QEC result motivates a
complementary finite-block question: what protection can local scrambling
provide when the encoder and recovery are finite, local,
\(U(1)\)-covariant, and number conserving, and when performance is judged
only after an explicit GJC receiver?  Here covariance means commutation with
local photon-number phase rotations.  A flagged erasure is a known-location
loss event: the classical flag identifies the erased memory cells, while the
lost subsystem remains inaccessible.  Recovery may depend on the heralded
erasure pattern but not on the unknown visibility.  We therefore neither treat
flagged erasure as a new noise model nor seek to reproduce the large-code
threshold analysis of Ref.~\cite{Huang2022}.  Instead, we isolate the effect
of local scrambling under locality, photon-number symmetry, and fixed readout,
while distinguishing conditional retention of one logical mode from a
conservative memory-only \(1/n\) throughput benchmark.

Motivated by scrambling-based metrological protection~\cite{Wysocki2026}, we
ask whether local number-conserving scrambling---the unitary spreading of
logical occupation across memory cells---can reduce the vulnerability of
stored complex visibility to location-dependent loss, and whether a
parameter-independent decoder can return the surviving information to a fixed
GJC receiver.  Scrambling diagnostics or entanglement growth do not answer
this question by themselves.  Nor does a large post-erasure QFI guarantee
either a parameter-independent decoder or information accessible through
local optical measurements.  The required chain has three logically distinct
links: retention of the two-parameter state family, physical recovery of that
family, and extraction of its information by the prescribed measurement.
Establishing all three is especially nontrivial when every flagged erasure
pattern must be included in the performance average.

The covariance constraint is not a blank-slate coding problem.  Exact
finite-dimensional quantum error correction with a nontrivial continuous
transversal symmetry is obstructed, and approximate covariant correction
obeys quantitative accuracy--size tradeoffs~\cite{Eastin2009,Faist2020}.  Moreover, random $U(1)$-covariant codes
generated by charge-conserving unitaries are already known: in the asymptotic
regimes analyzed in Ref.~\cite{Kong2022}, their average- and worst-case
purified-distance errors against erasure scale as $O(n^{-1})$.  That result
sets an asymptotic benchmark but does not supply the local finite-depth
encoder, explicit decoder, and fixed-receiver assessment required here.
Our contribution is the finite-size operational chain from a local
number-conserving encoder, through a parameter-independent recovery, to the
classical Fisher information (CFI) of a fixed GJC receiver, together with an all-pattern calculation and
explicit gate-set compilation.  We do not claim the introduction of random
covariant coding, an exact covariant erasure code, or an asymptotic capacity
theorem.

We establish the channel-to-CFI link analytically and evaluate the complete
operational chain directly for finite instances of the architecture in
Fig.~\ref{fig:architecture}.  First, we derive
the stored two-parameter signal family and the CFI of its fixed GJC readout.
We then place the exact-QEC benchmark beside its relevant covariance
obstruction: exact product recovery restores the complete state family and
all subsequent GJC probabilities, but a finite-dimensional strictly
$U(1)$-covariant encoder of the logical occupancy cannot correct every
single-memory erasure exactly.  The main analytical result therefore concerns
the approximate-recovery regime.  We prove that the half-diamond distance between the effective logical
channel and the identity---a worst-case channel-error measure that includes an
arbitrary purifying reference---uniformly controls the operator-norm deviation
of the fixed-receiver CFI on
$\lvert g\rvert\leq g_{\max}<1$, including general bias and
quadrature-mixing errors that are not described by one visibility-retention
factor.  To delimit that statement, we also derive an unrestricted
fixed-pattern Haar decoupling estimate, show that one ideal shared reference
photon retains exactly one half of the single-photon QFI after local
twirling, and prove a causal-depth necessary condition for uniform
local-erasure protection.  These results do not prove finite-depth covariant
achievability; the pattern-averaged and high-probability worst-pattern
decoupling targets are stated separately as open questions.

The finite-size calculations then determine what local scrambling, explicit
recovery, and a fixed receiver contribute separately.  Number-conserving
scrambling suppresses the high-leakage tail over erasure locations and
improves lower-tail recoverability, although mean retained information is not
monotone in depth.  Task-weighted decoders and fixed phase designs are tested
on held-out instances, while the final all-pattern calculation retains every
flagged event.  An equal-statistical-budget screen identifies the
depth-\(2\) nearest-neighbor brickwork encoder, built from alternating
layers of disjoint two-memory gates, as the low-cost candidate.  A frozen
follow-up then places the encoders in one abstract nearest-neighbor \(U(1)\)
gate alphabet with a depth-\(7\) cap.  It compiles all 60 frozen
single-node targets drawn from the invariant Haar distribution within
each fixed-charge sector (hereafter charge-Haar) and one symmetric
fixed-excitation (Dicke) isometry to
the stated numerical tolerances, then recomputes the follow-up control
recoveries and flagged events from the actual circuit codewords, and at \(p=0.30\) supports the
depth-\(2\) candidate against the sampled depth-\(7\) brickwork and
charge-Haar controls and the fixed Dicke comparator.  The reported
simultaneous 95\% intervals are conditional on that fixed comparator.  This
is an encoder-only comparison of evaluated circuits; depth \(7\) is a
feasible synthesis upper bound, and neither
recovery nor calibrated-platform resources are equalized.  For a separate
frozen batch of thirty separately generated pairs of \(n=5\), depth-\(2\)
encoders, all nontrivial local recovery branches are compiled into
nearest-neighbor \(U(1)\)-conserving circuits and verified.  The frozen simulator rerun is a deterministic computational reproduction on those same 30 complete
encoder-pair/compilation realizations, followed by a data-informed, post hoc
reanalysis of the inherited synthetic output-coherence attenuation and
partial-timing model.  A two-sided realization-cluster bootstrap uses studentized statistics
(normalized by cluster standard errors) and the maximum absolute statistic
over all 360 cells (the max-\(t\) correction) to identify
191\ cells with positive simultaneous 95\% lower
endpoints.  These nominal finite-sample approximate intervals apply only to
the evaluated discrete model grid.  In a separate 360-cell family, the
memory-only \(1/n\) throughput calculation has
0\ cells with positive simultaneous lower endpoints
and does not outperform the tested parallel bare-memory benchmark.  We make no asymptotic threshold,
nonzero-rate coding, calibrated-platform, or fixed-total-resource advantage
claim.

\paragraph*{Notation and schematic conventions.}
The node index is \(j\in\{A,B\}\), and \(b\) denotes the telescope baseline.
The collected astronomical modes are \(s_j\), with pre-storage state
\(\rho_s(g)\), where \(g=g_R+i g_I\) is the scalar complex visibility and
\(\boldsymbol g=(g_R,g_I)^{\mathsf T}\) is its real two-parameter
representation.  At node \(j\), \(M_j\) is the interface memory carrying the
logical occupancy qubit \(L_j\), \(P_j\) is the \(n_j\)-cell physical memory
block, and \(V_j:\mathcal H_{L_j}\to\mathcal H_{P_j}\) is the local encoding
isometry.  The two-node stored state is \(\rho_M(g)\).  A flagged erasure set
\(E_j\subseteq P_j\) identifies the lost cells and
\(E=(E_A,E_B)\) denotes the complete pattern.  The encoded-erasure channel is
\(\mathcal N_{E_A,E_B}\), the pattern-conditioned local recovery is
\(\mathcal R_j^{E_j}\), and \(M'_j\) is its logical output.  For fixed \(E\),
the recovered two-node state is
\(\widetilde\rho_{M,E}(g)=(\mathcal R_A^{E_A}\otimes
\mathcal R_B^{E_B})\mathcal N_{E_A,E_B}[\rho_M(g)]\); the schematic suppresses
the subscript \(E\) in \(\widetilde\rho_M(g)\).  The corresponding local
logical channel is \(\Lambda_{j,E_j}\); when the node label is suppressed,
\(\mathcal R_E\) and \(\Lambda_E\) denote the recovery and logical channel for
one fixed branch.  We use \(H\) for quantum Fisher information and \(F\) for
the classical Fisher information of a specified receiver.  The programmed
receiver phase is \(\delta\), whereas \(\delta'\) denotes the effective phase
after a known pattern-dependent correction.

In Fig.~\ref{fig:architecture}, teal cylinders denote interface memories,
teal circles denote individual cells of \(P_j\), and a red flag marks a cell
in \(E_j\).  Orange arrows denote signal capture, blue arrows denote local
encoding or recovery, black lines denote receiver optical paths, and dashed
inter-node lines denote distributed joint states rather than communication
channels.  SPS and BS denote the single-photon source and balanced beam
splitters, respectively, and \(D_{j1},D_{j2}\) are the two detectors at node
\(j\).

Figure~\ref{fig:architecture} fixes the physical and logical order used
throughout the paper: the astronomical coherence is stored before encoding,
whereas the ancillary single-photon reference is introduced only after
recovery for the GJC readout.  The protected object is therefore the
parameter-bearing signal memory, not the replaceable ancillary resource.

\begin{figure}[t]
\centering
\includegraphics[width=\columnwidth]{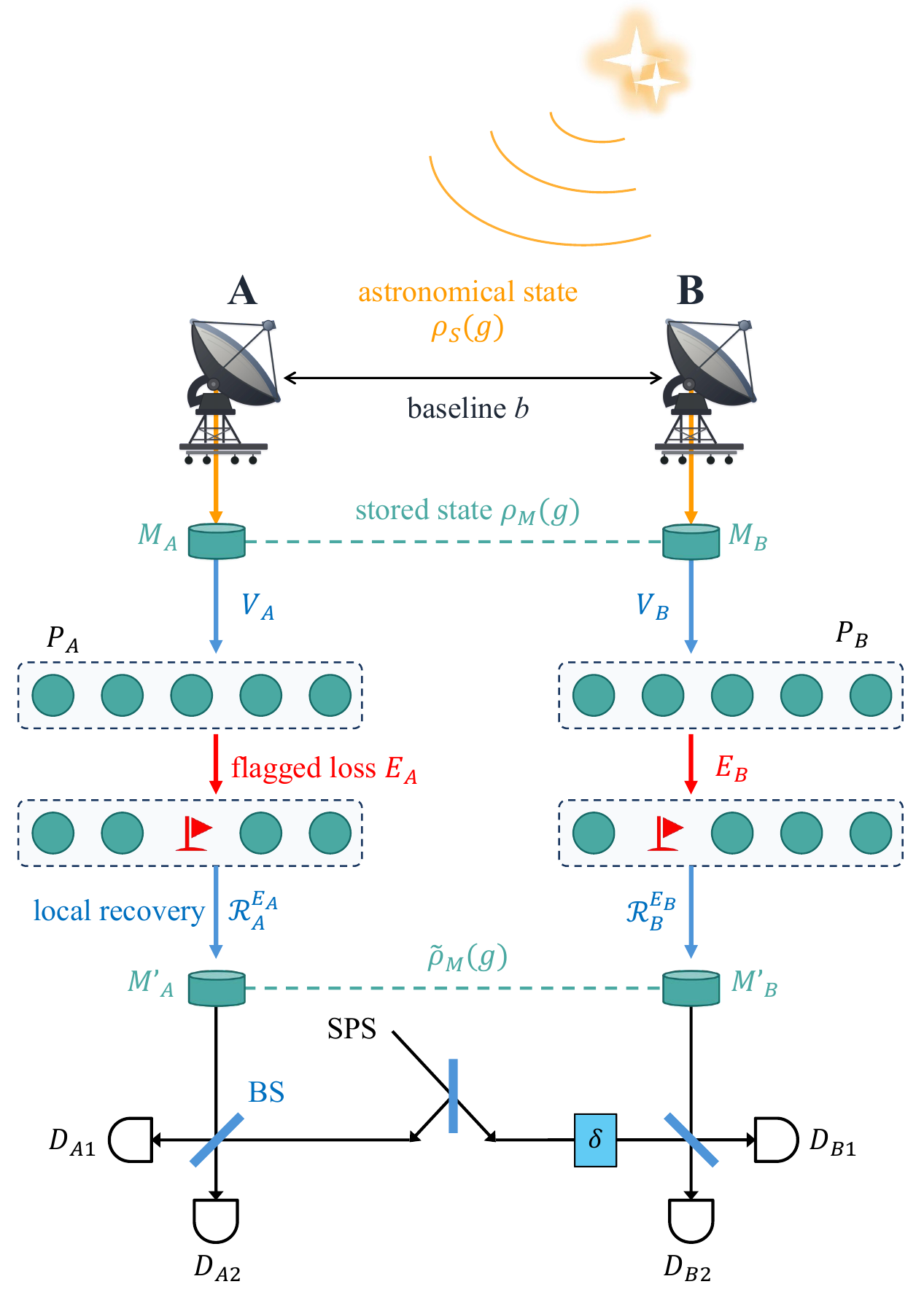}
\caption{Protocol architecture for finite-size flagged-erasure protection.
Signal capture, local encoding, flagged loss, pattern-conditioned recovery,
and fixed GJC readout appear in their logical order.  The shared
single-photon reference is introduced only after recovery.  Symbols and
graphical conventions are defined in the notation paragraph; reference
generation and distribution, flag acquisition, classical control,
feed-forward, reset, and detector occupancy are not shown.}
\label{fig:architecture}
\end{figure}

\paragraph*{Resource boundary and denominators.}
We distinguish three quantities throughout: information per incident
temporal mode, a dimensionless retention conditioned on successful capture
and reference availability, and information per wall-clock time under a
fixed total resource vector \(\mathcal C\).  Only the last can support a
platform-level system-advantage claim.  For a strategy \(\Pi\), the relevant
platform endpoint is
\begin{equation}
\mathcal J_{\Pi}(g;\mathcal C)
=
\lim_{T\rightarrow\infty}
\frac{F_g[Y_{\Pi}(0,T)]}{T},
\label{eq:platform-information-rate}
\end{equation}
where \(Y_{\Pi}(0,T)\) is the complete classical record, including
capture and reference failures, erasure and flag outcomes, rejected events,
no-clicks, dark counts, and timeouts.  The comparator is the best allowed
bare strategy under the same \(\mathcal C\), not a single favored bare
implementation.  The reported \(R_C\) and \(G\) below are conditional
diagnostics and do not evaluate Eq.~\eqref{eq:platform-information-rate}.
The platform-total-resource requirements and the present scope of any
hardware comparison are given in Appendix~\ref{sec:platform-contract}.

\section{Signal model and operational readout}
\label{sec:source}

\subsection{Two-aperture source model}

Let $s_A$ and $s_B$ denote the astronomical signal modes collected in one spatiotemporal mode at telescopes $A$ and $B$, respectively.  We reserve $a_A$ and $a_B$ for the two local modes of the ancillary single-photon reference used in the GJC readout.  The numerical study uses a balanced known-parameter model convention with equal mean signal intensities and mean total photon number per temporal mode $\epsilon\ll1$; the unequal-intensity and nuisance-flux extension is given below.  This is not a calibration to a selected platform.  The corresponding weak-thermal-light expansion is standard in quantum optical interferometry~\cite{Tsang2011,Gottesman2012}.  To first order in $\epsilon$, the balanced source state is
\begin{equation}
\rho_s(g)
=
(1-\epsilon)\left\lvert 00\right\rangle{}\!\left\langle 00\right\rvert{}
+\epsilon\rho_s^{(1)}(g)
+O(\epsilon^2),
\label{eq:source-state}
\end{equation}
where
\begin{align}
\rho_s^{(1)}(g)
&=
\frac{1}{2}
\left(
\left\lvert 10\right\rangle{}\!\left\langle 10\right\rvert{}
+\left\lvert 01\right\rangle{}\!\left\langle 01\right\rvert{}
\right.
\nonumber\\[-0.5ex]
&\qquad\left.
+g\left\lvert 10\right\rangle{}\!\left\langle 01\right\rvert{}
+g^*\left\lvert 01\right\rangle{}\!\left\langle 10\right\rvert{}
\right)
\nonumber\\
&=
\frac12
\begin{pmatrix}
1 & g\\
g^* & 1
\end{pmatrix}_{\{\left\lvert 10\right\rangle{},\left\lvert 01\right\rangle{}\}}.
\label{eq:single-photon-source}
\end{align}
The normalized complex visibility is
\begin{equation}
g=g_R+i g_I=\left\lvert g\right\rvert{}e^{i\phi},
\qquad \left\lvert g\right\rvert{}\leq1.
\label{eq:g-definition}
\end{equation}
In the far-field paraxial regime, $g$ is a normalized Fourier component of the source intensity distribution through the van Cittert--Zernike theorem, which provides the imaging interpretation of the measured visibility~\cite{Gottesman2012,Khabiboulline2019PRA}.  The reported visibility-estimation problem is therefore two-parameter:
\begin{equation}
\boldsymbol{\theta}=(g_R,g_I)^{\mathsf T}.
\end{equation}

\subsection{QFI matrix of the astronomical state}

Within the single-photon sector,
\begin{equation}
\rho_s^{(1)}(g)
=\frac12\left(\mathbb{I}+g_R\sigma_x-g_I\sigma_y\right),
\label{eq:bloch-state}
\end{equation}
with Bloch vector
\begin{equation}
\boldsymbol{r}=(g_R,-g_I,0)^{\mathsf T}.
\end{equation}
For a qubit state $\rho=(\mathbb{I}+\boldsymbol{r}\cdot\boldsymbol{\sigma})/2$, the symmetric-logarithmic-derivative (SLD) QFI matrix is~\cite{Braunstein1994}
\begin{equation}
[H(\boldsymbol{\theta})]_{\mu\nu}
=
\partial_\mu\boldsymbol{r}\cdot\partial_\nu\boldsymbol{r}
+
\frac{(\boldsymbol{r}\cdot\partial_\mu\boldsymbol{r})
(\boldsymbol{r}\cdot\partial_\nu\boldsymbol{r})}{1-\left\lvert \boldsymbol{r}\right\rvert{}^2},
\label{eq:qubit-qfi-general}
\end{equation}
for $\left\lvert \boldsymbol{r}\right\rvert{}<1$.  Direct substitution gives
\begin{equation}
\begin{aligned}
H^{(1)}(g)
&=
\frac{1}{1-\left\lvert g\right\rvert{}^2}
\\[-0.5ex]
&\quad\times
\begin{pmatrix}
1-g_I^2 & g_Rg_I\\
g_Rg_I & 1-g_R^2
\end{pmatrix}.
\end{aligned}
\label{eq:qfi-single}
\end{equation}
Because the vacuum and single-photon weights in Eq.~\eqref{eq:source-state} are independent of $g$, orthogonality of the photon-number sectors gives
\begin{equation}
H_s(g)=\epsilon H^{(1)}(g)+O(\epsilon^2).
\label{eq:qfi-weak-source}
\end{equation}

\paragraph*{Intensity imbalance and nuisance flux.}
Let \(\kappa\in(0,1)\) be the fraction of the one-photon intensity arriving
at node \(A\).  The normalized one-photon state becomes
\begin{equation}
\rho_{s,\kappa}^{(1)}(g)
=
\begin{pmatrix}
\kappa & \sqrt{\kappa(1-\kappa)}\,g\\
\sqrt{\kappa(1-\kappa)}\,g^* & 1-\kappa
\end{pmatrix},
\label{eq:unbalanced-source}
\end{equation}
and Eq.~\eqref{eq:single-photon-source} is the case \(\kappa=1/2\).
Within the vacuum--one-photon model, take
\(\boldsymbol{\vartheta}=(\epsilon,\kappa,g_R,g_I)\).  Direct use of
Eq.~\eqref{eq:qubit-qfi-general} gives
\begin{align}
H_{\epsilon\epsilon}
&=\frac{1}{\epsilon(1-\epsilon)},
&
H_{\epsilon\alpha}&=0,
\label{eq:flux-qfi-block}\\
H_{\kappa\kappa}^{(1)}
&=\frac{1}{\kappa(1-\kappa)},
&
H_{\kappa g_\mu}^{(1)}
&=0\quad(\mu=R,I),
\nonumber\\
H_{gg}^{(1)}
&=4\kappa(1-\kappa)H^{(1)}(g),
\label{eq:unbalanced-qfi-block}
\end{align}
where \(\alpha\in\{\kappa,g_R,g_I\}\), and the
\((\kappa,g_R,g_I)\) block of the full weak-source QFI carries the common
factor \(\epsilon\).  Consequently, at this order the QFI Schur complement
for \(g\), after treating both flux and intensity imbalance as nuisance
parameters, is
\(\epsilon\,4\kappa(1-\kappa)H^{(1)}(g)\).  The balanced numerical results
below set \(\kappa=1/2\); they do not claim robustness to an uncalibrated
imbalance in the fixed receiver.

The boundary $\left\lvert g\right\rvert{}=1$ is understood by a limiting procedure because the rank of the single-photon block changes there.  A detailed derivation is given in Appendix~\ref{app:qfi}.

\subsection{Coherent signal storage}
\label{sec:storage}

Let $M_A$ and $M_B$ be interface memories whose logical basis $\{\left\lvert 0\right\rangle{},\left\lvert 1\right\rangle{}\}$ records the absence or presence of one stored optical excitation.  Coherent storage of the astronomical excitation before nonlocal readout is the signal-memory route considered in quantum-network interferometry~\cite{Khabiboulline2019PRL,Khabiboulline2019PRA}.  An ideal local storage isometry satisfies
\begin{equation}
\begin{aligned}
\left\lvert 0\right\rangle{}_{s_j}\left\lvert 0\right\rangle{}_{M_j}
&\mapsto \left\lvert 0\right\rangle{}_{s_j}\left\lvert 0\right\rangle{}_{M_j},
\\
\left\lvert 1\right\rangle{}_{s_j}\left\lvert 0\right\rangle{}_{M_j}
&\mapsto \left\lvert 0\right\rangle{}_{s_j}\left\lvert 1\right\rangle{}_{M_j}.
\end{aligned}
\end{equation}
After tracing out the reset optical modes, the stored state has the same one-excitation coherence as Eq.~\eqref{eq:single-photon-source}.

To include imperfect but herald-independent capture, let $\eta_c$ be the symmetric capture efficiency and let $\lambda_c e^{i\varphi_c}$ be the coherence-transfer amplitude.  A convenient phenomenological model is
\begin{align}
\rho_M(g)
&=
(1-\epsilon\eta_c)\left\lvert 00\right\rangle{}\!\left\langle 00\right\rvert{}
+\epsilon\eta_c\rho_M^{(1)}(g_c),
\label{eq:stored-state}\\
g_c&=\lambda_c e^{i\varphi_c}g,
\qquad 0\leq\lambda_c\leq1.
\label{eq:storage-coherence}
\end{align}
The known phase $\varphi_c$ may be calibrated, whereas an unknown fluctuating
phase reduces $\lambda_c$ after averaging.  For an ideal coherent storage
map, $\lambda_c=1$.  In real coordinates,
\begin{equation}
\boldsymbol{g}_c=J_c\boldsymbol{g},
\qquad
J_c=\lambda_c
\begin{pmatrix}
\cos\varphi_c&-\sin\varphi_c\\
\sin\varphi_c&\cos\varphi_c
\end{pmatrix}.
\label{eq:storage-jacobian}
\end{equation}
The stored QFI for the original parameters \((g_R,g_I)\) is therefore
\begin{equation}
H_M^{(g)}(g)
=
\epsilon\eta_c J_c^{\mathsf T}
H^{(1)}(g_c)J_c.
\label{eq:stored-qfi-general}
\end{equation}
For unequal signal intensities, the right-hand side acquires the additional
factor \(4\kappa(1-\kappa)\).  In the ideal known-transfer convention,
\(J_c=\mathbb I\), and
\begin{equation}
H_M(g)=\epsilon\eta_c H^{(1)}(g).
\label{eq:stored-qfi}
\end{equation}

Loss before successful storage only reduces $\eta_c$ and is not correctable by any subsequent encoding.  The proposed protection concerns noise occurring after Eq.~\eqref{eq:stored-state} has been established.

\subsection{GJC readout and its operational Fisher information}
\label{sec:gjc}

\subsubsection{General single-photon ancillary resource}

Following the entanglement-assisted GJC construction~\cite{Gottesman2012,Zhang2025}, the ancillary state delivered to the two nodes is modeled as
\begin{equation}
\rho_a
=
(1-\eta_a)\left\lvert 00\right\rangle{}\!\left\langle 00\right\rvert{}
+\eta_a\rho_a^{(1)},
\label{eq:ancilla-total}
\end{equation}
where \(\eta_a\) is a model-level reference-availability factor that may lump
source success, distribution, and retrieval loss, and
\begin{equation}
\rho_a^{(1)}
=
\begin{pmatrix}
q_A & c\\
c^* & q_B
\end{pmatrix}_{\{\left\lvert 10\right\rangle{}_a,\left\lvert 01\right\rangle{}_a\}},
\qquad q_A+q_B=1.
\label{eq:ancilla-one}
\end{equation}
Define
\begin{equation}
\mu=2c=\nu_a e^{i\delta},
\qquad 0\leq\nu_a\leq2\sqrt{q_Aq_B}\leq1.
\label{eq:mu-resource}
\end{equation}
The ideal balanced resource has $q_A=q_B=1/2$ and $\nu_a=1$.  The source realization sketched in Fig.~\ref{fig:architecture} injects one photon from the single-photon source (SPS) into a central balanced beam splitter and applies a controllable phase shift $\delta$ to one output arm.  Absorbing the fixed beam-splitter phase into $\delta$, the resulting reference is
\begin{equation}
\left\lvert\psi_\delta\right\rangle{}_a
=
\frac{
\left\lvert1\right\rangle{}_{a_A}\left\lvert0\right\rangle{}_{a_B}
+e^{i\delta}
\left\lvert0\right\rangle{}_{a_A}\left\lvert1\right\rangle{}_{a_B}
}{\sqrt2},
\label{eq:ideal-reference-state}
\end{equation}
so that $\rho_a^{(1)}=\left\lvert\psi_\delta\right\rangle{}\!\left\langle\psi_\delta\right\rvert{}$ in the ideal limit.  Operationally, \(\delta\) is a known programmed relative optical phase that could be set by a stabilized path-length delay or an electro-optic phase shifter; no platform-specific precision, bandwidth, or switching latency is assumed.  The calculations below retain the general mixed resource in Eq.~\eqref{eq:ancilla-one}; SPS inefficiency, distribution loss, and imperfect reference coherence are represented phenomenologically by $\eta_a$, $q_A$, $q_B$, and $\nu_a$.  These variables do not account for failed preparation attempts, source throughput, synchronization, phase-lock bandwidth, reference memory occupancy, or a calibrated detector response; those are separate entries in Appendix~\ref{sec:platform-contract}.

\subsubsection{Local interference measurement}

The receiver is described by a positive-operator-valued measure (POVM).  At node $j\in\{A,B\}$, the recovered memory output $M'_j$ is retrieved into
a signal slot and interferes with the local ancillary mode $a_j$ on a
balanced beam splitter.  In Fig.~\ref{fig:architecture}, clicks at
$D_{A1}$ and $D_{A2}$ are assigned $r=+1$ and $r=-1$, while clicks at
$D_{B1}$ and $D_{B2}$ are assigned $s=+1$ and $s=-1$, respectively.
Writing the retrieved logical excitation as \(M\), an output sign
$x\in\{\pm1\}$ within either local one-photon sector corresponds to
\begin{equation}
\left\lvert x\right\rangle{}_j
=
\frac{\left\lvert 1_M0_a\right\rangle{}_j+x\left\lvert 0_M1_a\right\rangle{}_j}{\sqrt2}.
\label{eq:local-click-state}
\end{equation}
For the signal-memory architecture we henceforth use the ideal known-transfer
storage convention $g_c=g$ while retaining the capture probability
$\eta_c$.  Thus $g$ denotes the logical visibility after capture, and all
per-input GJC probabilities carry the modeled common factor
$\epsilon\eta_c$.  This convention is not normalized by platform throughput.
Projecting the retrieved stored block
$\rho_M^{(1)}(g)\otimes\rho_a^{(1)}$ onto
$\left\lvert r\right\rangle{}_A\left\lvert s\right\rangle{}_B$ gives
\begin{equation}
p_{rs}(g)
=
\frac{\epsilon\eta_c\eta_a}{8}
\left[1+rs\operatorname{Re}(g\mu^*)\right]
+O(\epsilon^2).
\label{eq:gjc-probability}
\end{equation}
The derivation is given explicitly in Appendix~\ref{app:gjc}.  Summing over all four outcomes,
\begin{equation}
p_{\mathrm{succ}}=\frac{\epsilon\eta_c\eta_a}{2}+O(\epsilon^2),
\label{eq:gjc-success}
\end{equation}
which exhibits the factor-$1/2$ success probability of the linear-optical GJC measurement~\cite{Gottesman2012,Zhang2025}.

Let $z=rs$ be the outcome parity and define
\begin{equation}
\boldsymbol{g}=\begin{pmatrix}g_R\\g_I\end{pmatrix},
\qquad
\boldsymbol{u}_\delta=\begin{pmatrix}\cos\delta\\\sin\delta\end{pmatrix}.
\end{equation}
Conditioned on success,
\begin{equation}
P(z\mid\mathrm{succ},g)
=
\frac12\left[1+z\nu_a\boldsymbol{u}_\delta^{\mathsf T}\boldsymbol{g}\right].
\label{eq:parity-probability}
\end{equation}
The CFI per input temporal mode is therefore
\begin{equation}
F_\delta^{\mathrm{GJC}}(g)
=
\frac{\epsilon\eta_c\eta_a}{2}
\frac{
\nu_a^2\boldsymbol{u}_\delta\boldsymbol{u}_\delta^{\mathsf T}
}{
1-\nu_a^2(\boldsymbol{u}_\delta^{\mathsf T}\boldsymbol{g})^2
}
+O(\epsilon^2).
\label{eq:gjc-fisher}
\end{equation}
Thus $\delta=0$ and $\delta=\pi/2$ probe the real and imaginary quadratures, respectively.  Equation~\eqref{eq:gjc-fisher}, rather than the unconstrained QFI alone, is the operational benchmark for the proposed architecture.

\paragraph*{Unequal intensities, storage transfer, and classical nuisance
parameters.}
For the source in Eq.~\eqref{eq:unbalanced-source}, define
\begin{equation}
A_\kappa=\kappa q_B+(1-\kappa)q_A,
\qquad
B_{\kappa,\delta}
=
\sqrt{\kappa(1-\kappa)}\,\nu_a
\boldsymbol{u}_\delta^{\mathsf T}\boldsymbol{g}_c.
\label{eq:unbalanced-gjc-ab}
\end{equation}
The four successful detector probabilities and their sum are
\begin{align}
p_{rs}
&=
\frac{\epsilon\eta_c\eta_a}{4}
\left(A_\kappa+rsB_{\kappa,\delta}\right)
+O(\epsilon^2),
\label{eq:unbalanced-gjc-probability}\\
p_{\rm succ}
&=
\epsilon\eta_c\eta_a A_\kappa+O(\epsilon^2).
\label{eq:unbalanced-gjc-success}
\end{align}
At fixed known \(\kappa\), the CFI for the original visibility
\(\boldsymbol g\) is
\begin{equation}
F_{\delta,\kappa}^{\mathrm{GJC},(g)}
=
\epsilon\eta_c\eta_a
J_c^{\mathsf T}
\frac{\kappa(1-\kappa)\nu_a^2 A_\kappa}
{A_\kappa^2-B_{\kappa,\delta}^2}
\boldsymbol{u}_\delta\boldsymbol{u}_\delta^{\mathsf T}
J_c.
\label{eq:unbalanced-gjc-fisher}
\end{equation}
Equations~\eqref{eq:gjc-probability}--\eqref{eq:gjc-fisher} follow for
\(\kappa=q_A=q_B=1/2\) and \(J_c=\mathbb I\).
If \(\boldsymbol{\xi}=(\epsilon,\kappa)\) is estimated jointly rather than
calibrated, the relevant classical information is the Schur complement
\begin{equation}
F_{g\mid\boldsymbol{\xi}}
=
F_{gg}
-F_{g\boldsymbol{\xi}}
F_{\boldsymbol{\xi}\boldsymbol{\xi}}^{+}
F_{\boldsymbol{\xi}g},
\label{eq:gjc-nuisance-schur}
\end{equation}
formed from the complete outcome distribution, including failure, where
\({}^{+}\) denotes the Moore--Penrose inverse.  The total flux is orthogonal
to \(g\) at this order \((F_{g\epsilon}=0)\), whereas an unknown imbalance
can correlate with \(g\) away from the balanced point.  For the balanced
ancillary resource at \(\kappa=1/2\), both \(A_\kappa\) and
\(\sqrt{\kappa(1-\kappa)}\) have zero first derivative with respect to
\(\kappa\), so \(F_{g\kappa}=0\) locally.  This local orthogonality, rather
than an assumption that nuisance parameters never matter, is the convention
used in the reported balanced simulations.

\subsection{Why scrambling an ancillary memory is not metrological locking}

In the experimental memory-assisted GJC architecture reported by
Wang \textit{et al.}~\cite{Wang2026}, the pre-interference state is
\begin{equation}
\rho_{\mathrm{pre}}(g)=\rho_s(g)\otimes\rho_{M,a},
\label{eq:ancillary-memory-product}
\end{equation}
where $\rho_{M,a}$ is a memory state used to generate $\rho_a$ upon readout.  A $g$-independent scrambling unitary $U_M$ acting only on the ancillary memories gives
\begin{equation}
\rho_{\mathrm{pre}}'(g)
=
\rho_s(g)\otimes U_M\rho_{M,a}U_M^\dagger.
\end{equation}
Tensoring with a parameter-independent state and applying a parameter-independent unitary do not change the QFI with respect to $g$:
\begin{equation}
H_g[\rho_{\mathrm{pre}}'(g)]=H_g[\rho_s(g)].
\end{equation}
Scrambling may protect the resource coherence $\nu_a$ or efficiency $\eta_a$, thereby improving Eq.~\eqref{eq:gjc-fisher}, but the astronomical parameter is not stored in that memory.  This route is correctly described as random erasure coding of the ancillary entangled resource, not as QFI locking of $g$.

\section{Recovery benchmarks, operational stability, and the covariance obstruction}
\label{sec:encoding}

\subsection{Local encoding isometries}

At each node $j\in\{A,B\}$, a logical memory qubit $L_j$ is encoded into $n_j$ physical memory qubits:
\begin{equation}
\begin{aligned}
V_j:\mathcal{H}_{L_j}
&\longrightarrow \mathcal{H}_{P_j}
=(\mathbb C^2)^{\otimes n_j},
\\
V_j\left\lvert i\right\rangle{}&=\left\lvert i_L\right\rangle{}_j,
\qquad i=0,1.
\end{aligned}
\label{eq:local-isometry}
\end{equation}
The global encoder is deliberately restricted to
\begin{equation}
V=V_A\otimes V_B,
\label{eq:local-product-encoder}
\end{equation}
so that no nonlocal gate is required during encoding.  A global Haar unitary across geographically separated nodes would largely defeat the architectural motivation of GJC\@.

For erasure sets $E_A\subseteq P_A$ and $E_B\subseteq P_B$, define
\begin{equation}
\mathcal{N}_{E_A,E_B}(\rho)
=
\operatorname{Tr}_{E_AE_B}\left[V\rho V^\dagger\right].
\label{eq:erasure-channel}
\end{equation}
The erasure locations are assumed to be flagged.  This is stronger than unheralded amplitude damping and is revisited in Appendix~\ref{sec:limitations}.

\subsection{Exact local erasure-correction condition}

\paragraph*{Definition (Local erasure correction).}
The local code $V_j$ exactly corrects erasure of $E_j$ if there exists a parameter-independent channel $\mathcal{R}_j^{E_j}$ such that
\begin{equation}
\mathcal{R}_j^{E_j}\circ\operatorname{Tr}_{E_j}
\left[V_j X V_j^\dagger\right]=X
\label{eq:exact-local-recovery}
\end{equation}
for every operator $X$ on $\mathcal{H}_{L_j}$.

For erasure errors, the Knill--Laflamme condition~\cite{Knill2000} is equivalent to
\begin{equation}
\operatorname{Tr}_{\bar E_j}\left[\left\lvert i_L\right\rangle{}\!\left\langle k_L\right\rvert{}\right]
=\delta_{ik}\sigma_{E_j},
\qquad i,k\in\{0,1\},
\label{eq:erasure-KL}
\end{equation}
for a state $\sigma_{E_j}$ independent of the logical value.  Equation~\eqref{eq:erasure-KL} states that the erased subsystem contains neither population nor coherence information about the logical input.

\subsection{Covariance obstruction to uniform exact correction}
\label{sec:covariant-no-go}

Let the physical number generator at node $j$ be additive,
\begin{equation}
\hat N_j=\sum_{\ell=1}^{n_j}\hat n_{j,\ell},
\end{equation}
and let $\hat n_{L_j}=\left\lvert1\right\rangle{}\!\left\langle1\right\rvert{}$
be the nontrivial logical occupancy generator.  Strict covariance, including a
fixed background charge $Q_{0,j}$, means
\begin{equation}
e^{-i\alpha\hat N_j}V_j
=
e^{-i\alpha Q_{0,j}}V_j e^{-i\alpha\hat n_{L_j}}
\quad\text{for every }\alpha.
\label{eq:strict-covariance-main}
\end{equation}
For the neighboring-charge construction used below, the two codewords have
physical charges $Q_{0,j}$ and $Q_{0,j}+1$.

\begin{lemma}[Strict covariance excludes uniform exact single-memory correction]
\label{lem:covariant-no-go}
For a finite-dimensional encoder satisfying
Eq.~\eqref{eq:strict-covariance-main}, erasure of every individual physical
memory cannot be exactly correctable when $\hat n_{L_j}$ is nontrivial.
\end{lemma}

\begin{proof}
Suppose, to the contrary, that erasure of every physical memory $\ell$ is
exactly correctable.  The Knill--Laflamme condition for erasure of memory
$\ell$ then implies
\begin{equation}
V_j^\dagger O_{j,\ell}V_j
=c_{j,\ell}(O_{j,\ell})\mathbb I_{L_j}
\end{equation}
for every operator $O_{j,\ell}$ supported on that memory.  Taking
$O_{j,\ell}=\hat n_{j,\ell}$ and summing over $\ell$ gives
\begin{equation}
V_j^\dagger\hat N_jV_j=c_j\mathbb I_{L_j}
\end{equation}
for a scalar $c_j$.  On the other hand, differentiating
Eq.~\eqref{eq:strict-covariance-main} at $\alpha=0$ yields
\begin{equation}
\hat N_jV_j
=V_j\left(Q_{0,j}\mathbb I_{L_j}+\hat n_{L_j}\right),
\end{equation}
and hence
\begin{equation}
V_j^\dagger\hat N_jV_j
=Q_{0,j}\mathbb I_{L_j}+\hat n_{L_j},
\end{equation}
which is not a scalar.  This is a contradiction.
\end{proof}

Lemma~\ref{lem:covariant-no-go} is the single-erasure form of the standard
continuous-symmetry obstruction to exact finite-dimensional covariant quantum
error correction~\cite{Eastin2009,Faist2020}.  It does not exclude
approximate covariant correction, exact correction of a restricted set of
locations, a charge-neutral logical subsystem, or a reference-assisted
implementation whose asymmetry resource is included explicitly.  It does
exclude interpreting the following exact-QEC statements as an achievable
uniform exact code for the strictly covariant logical-occupancy architecture.

\subsection{Exact-recovery benchmark for the complex visibility}
\label{sec:exact-preservation}

\begin{lemma}[Tensor-product exact-recovery benchmark]
\label{lem:exact-state}
Suppose that $V_A$ and $V_B$ exactly correct erasures $E_A$ and $E_B$, respectively.  Then, for every two-node logical operator $X_{AB}$,
\begin{equation}
(\mathcal{R}_A^{E_A}\otimes\mathcal{R}_B^{E_B})
\circ\mathcal{N}_{E_A,E_B}(X_{AB})
=X_{AB}.
\label{eq:two-node-recovery}
\end{equation}
In particular, the complete stored astronomical family is restored:
\begin{equation}
(\mathcal{R}_A^{E_A}\otimes\mathcal{R}_B^{E_B})
\circ\mathcal{N}_{E_A,E_B}[\rho_M(g)]
=\rho_M(g)
\label{eq:state-family-restored}
\end{equation}
for every physically allowed $g$ within the vacuum--single-photon storage
model.
\end{lemma}

\begin{proof}
Expand an arbitrary operator as
\begin{equation}
X_{AB}=\sum_{i,k=0}^1\sum_{j,l=0}^1
x_{ij,kl}\left\lvert i\right\rangle{}\!\left\langle k\right\rvert{}_A\otimes\left\lvert j\right\rangle{}\!\left\langle l\right\rvert{}_B.
\end{equation}
The product encoder and product erasure channel act independently on each tensor factor.  Applying Eq.~\eqref{eq:exact-local-recovery} to each basis operator gives
\begin{align}
&(\mathcal{R}_A^{E_A}\otimes\mathcal{R}_B^{E_B})
\circ\mathcal{N}_{E_A,E_B}
\left(\left\lvert i\right\rangle{}\!\left\langle k\right\rvert{}_A\otimes\left\lvert j\right\rangle{}\!\left\langle l\right\rvert{}_B\right)
\nonumber\\
&\hspace{4cm}=
\left\lvert i\right\rangle{}\!\left\langle k\right\rvert{}_A\otimes\left\lvert j\right\rangle{}\!\left\langle l\right\rvert{}_B.
\end{align}
Linearity proves Eq.~\eqref{eq:two-node-recovery}.  Taking $X_{AB}=\rho_M(g)$ proves Eq.~\eqref{eq:state-family-restored}.
\end{proof}

\begin{corollary}[Recovery of the cross-aperture coherence]
Under the hypotheses of Lemma~\ref{lem:exact-state},
\begin{equation}
g\left\lvert 1_L0_L\right\rangle{}\!\left\langle 0_L1_L\right\rvert{}
\longmapsto
g\left\lvert 10\right\rangle{}\!\left\langle 01\right\rvert{}
\label{eq:coherence-restored}
\end{equation}
after erasure and recovery.  Both $g_R$ and $g_I$ are restored; the result is not restricted to a one-parameter phase family.
\end{corollary}

\subsection{QFI benchmark under exact recovery}
\label{sec:qfi-preservation}

\begin{lemma}[QFI matrix under an exactly reversible channel]
\label{lem:qfi-preservation}
Let $\rho(\boldsymbol{\theta})$ be any differentiable state family, and let $\mathcal{N}$ and $\mathcal{R}$ be parameter-independent quantum channels satisfying
\begin{equation}
\mathcal{R}\circ\mathcal{N}[\rho(\boldsymbol{\theta})]=\rho(\boldsymbol{\theta})
\label{eq:sufficiency-state-family}
\end{equation}
for every $\boldsymbol{\theta}$ in an open parameter region.  Then the SLD QFI matrices obey
\begin{equation}
H[\mathcal{N}(\rho(\boldsymbol{\theta}))]=H[\rho(\boldsymbol{\theta})].
\label{eq:qfi-equality}
\end{equation}
\end{lemma}

\begin{proof}
For any real vector $\boldsymbol{v}$, $\boldsymbol{v}^{\mathsf T}H\boldsymbol{v}$ is the scalar SLD QFI along the one-parameter curve $\boldsymbol{\theta}(t)=\boldsymbol{\theta}_0+t\boldsymbol{v}$.  Scalar QFI is monotone under parameter-independent channels~\cite{Braunstein1994}.  Hence
\begin{equation}
H[\rho]\succeq H[\mathcal{N}(\rho)]
\succeq H[\mathcal{R}\circ\mathcal{N}(\rho)]=H[\rho],
\end{equation}
where $\succeq$ denotes positive-semidefinite matrix order.  Both inequalities are therefore equalities.
\end{proof}

Combining Lemmas~\ref{lem:exact-state} and~\ref{lem:qfi-preservation} yields
\begin{equation}
H\!\left[\mathcal{N}_{E_A,E_B}(\rho_M(g))\right]
=H[\rho_M(g)]
=\epsilon\eta_c H^{(1)}(g)
\label{eq:stored-qfi-preserved}
\end{equation}
for every exactly correctable erasure pattern and ideal coherent storage.

\begin{remark}
Equation~\eqref{eq:stored-qfi-preserved} is stronger than a statement that one scalar QFI happens to be numerically unchanged at one value of $g$.  Exact recoverability restores the complete two-parameter state family and therefore every statistical decision problem defined on that family.
\end{remark}

\subsection{GJC-CFI benchmark under exact recovery}
\label{sec:cfi-preservation}

\begin{corollary}[Operational GJC statistics under exact recovery]
\label{cor:cfi}
Under the hypotheses of Lemma~\ref{lem:exact-state}, perform the same GJC POVM after recovery as would have been performed on the unencoded state.  Then every outcome probability is unchanged and
\begin{equation}
F_{\delta,E_A,E_B}^{\mathrm{GJC}}(g)
=F_\delta^{\mathrm{GJC}}(g),
\label{eq:cfi-exact}
\end{equation}
where the right-hand side is given by Eq.~\eqref{eq:gjc-fisher} with only the
modeled factors appearing there.
\end{corollary}

\begin{proof}
By Lemma~\ref{lem:exact-state}, the recovered density operator equals the unencoded logical density operator for every $g$.  The Born probabilities of any fixed POVM, including the GJC POVM, are therefore identical functions of $g$.  Equality of the probability functions and their derivatives implies equality of their CFI matrices.
\end{proof}

These exact-recovery statements are standard consequences of the defining
identity in Eq.~\eqref{eq:exact-local-recovery}; they are included to specify
the benchmark used later and are not claimed as new coding theorems.  A
decoder followed by the original GJC measurement is an explicit measurement
strategy attaining the unencoded GJC statistics on every exactly correctable
pattern.

\subsection{Operational stability under approximate recovery}
\label{sec:main-cfi-stability}

Exact correction is unavailable uniformly for the covariant logical
occupancy, but a channel-level approximation still gives a quantitative
operational guarantee.  The following result separates this general
conversion from the architecture-specific problem of constructing and
certifying a finite-depth recovery channel.

\begin{theorem}[Fisher-information stability under a fixed receiver]
\label{thm:fi-continuity}
Let $\rho(\boldsymbol{\theta})$ be a continuously differentiable
(\(C^1\)) state family on a compact parameter set contained in an open
differentiability domain, let $\Lambda$ and $\Gamma$ be
parameter-independent completely positive trace-preserving (CPTP) channels, and let
$\{M_y\}_{y=1}^{K}$ be a fixed finite POVM\@.  Define
\begin{align}
p_y(\boldsymbol{\theta})
&=
\operatorname{Tr}\!\left[M_y\Lambda(\rho(\boldsymbol{\theta}))\right],
\\
q_y(\boldsymbol{\theta})
&=
\operatorname{Tr}\!\left[M_y\Gamma(\rho(\boldsymbol{\theta}))\right],
\label{eq:fi-output-probabilities}\\
\varepsilon_\diamond
&=
\frac12\left\lVert\Lambda-\Gamma\right\rVert{}_\diamond,
\\
S_\rho^2
&=
\sum_\mu
\sup_{\boldsymbol{\theta}}
\left\lVert\partial_\mu\rho(\boldsymbol{\theta})\right\rVert{}_1^2.
\label{eq:fi-channel-error}
\end{align}
Let $F_\Lambda$ and $F_\Gamma$ denote the classical Fisher matrices of
$\{p_y\}$ and $\{q_y\}$, respectively, with the pointwise convention
\begin{equation}
F[p]
=
\sum_{y:p_y>0}
\frac{
\boldsymbol{\nabla}p_y
\boldsymbol{\nabla}p_y^{\mathsf T}
}{p_y}.
\label{eq:pointwise-cfi-convention}
\end{equation}
At an interior point where $p_y=0$, nonnegativity and differentiability imply
$\boldsymbol{\nabla}p_y=0$, and that outcome contributes zero in
Eq.~\eqref{eq:pointwise-cfi-convention}; no universal continuous extension
through a probability zero is assumed.
Suppose the ideal probabilities obey
$q_y(\boldsymbol{\theta})\geq q_0>0$ uniformly, and suppose
$\lVert H[\rho(\boldsymbol{\theta})]\rVert_{\mathrm{op}}\leq H_{\max}$.
Then
\begin{equation}
\left\lVert F_\Lambda-F_\Gamma\right\rVert{}_{\mathrm{op}}
\leq
C_{\mathrm{FI}}\varepsilon_\diamond,
\label{eq:fi-general-bound}
\end{equation}
where
\begin{equation}
C_{\mathrm{FI}}
=
\max\!\left\{
K S_\rho^2\left(\frac{2}{q_0}+\frac{1}{q_0^2}\right),
\frac{2H_{\max}}{q_0}
\right\}.
\label{eq:fi-general-constant}
\end{equation}
\end{theorem}

The proof, including the probability, tangent, and large-error estimates, is
given in Appendix~\ref{sec:approximate}.  The mathematical statement does not
require phase covariance.  Covariance remains a physical requirement for the
SSR-consistent implementation and for the scalar visibility-transfer
interpretation below.

\begin{corollary}[Uniform GJC-CFI stability]
\label{cor:gjc-fi-continuity}
Apply Theorem~\ref{thm:fi-continuity} to the normalized single-photon family
in Eq.~\eqref{eq:bloch-state} and the two-outcome GJC parity POVM in
Eq.~\eqref{eq:parity-probability}.  Assume that recovery and acceptance induce
a parameter-independent CPTP map on the normalized accepted single-photon
subensemble and that the acceptance probability is independent of $g$.
For $E=(E_A,E_B)$, write
\begin{equation}
\Lambda_{E_A,E_B}
=
\Lambda_{A,E_A}\otimes\Lambda_{B,E_B},
\label{eq:gjc-joint-logical-channel}
\end{equation}
and let
\begin{equation}
\varepsilon_{E_A,E_B}
=
\frac12
\left\lVert
\Lambda_{E_A,E_B}
-
\operatorname{id}_{L_AL_B}
\right\rVert{}_\diamond
\label{eq:gjc-joint-channel-error}
\end{equation}
be the error of the joint recovered logical channel relative to the identity
channel on the normalized two-logical-mode space, after any known
pattern-dependent phase correction.  Uniformly on
$\lvert g\rvert\leq g_{\max}<1$,
\begin{equation}
\left\lVert
F_{E,\mathrm{cond}}^{\mathrm{GJC}}
-
F_{0,\mathrm{cond}}^{\mathrm{GJC}}
\right\rVert{}_{\mathrm{op}}
\leq
C_F(g_{\max})\varepsilon_{E_A,E_B},
\label{eq:gjc-uniform-bound}
\end{equation}
with the explicit, generally nonoptimal constant
\begin{equation}
C_F(g_{\max})
=
\frac{16}{1-g_{\max}}
+
\frac{16}{(1-g_{\max})^2}.
\label{eq:gjc-uniform-constant}
\end{equation}
If the accepted GJC success weight is the same
$g$-independent quantity
\begin{equation}
s_0=\frac{\epsilon\eta_c\eta_a}{2}
\label{eq:gjc-success-weight}
\end{equation}
for the recovered and ideal experiments, then, to the retained order of the
weak-source model,
\begin{equation}
\left\lVert
F_E^{\mathrm{GJC}}-F_0^{\mathrm{GJC}}
\right\rVert{}_{\mathrm{op}}
\leq
s_0 C_F(g_{\max})\varepsilon_{E_A,E_B}+O(\epsilon^2).
\label{eq:gjc-mode-bound}
\end{equation}
\end{corollary}

For this specialization, the parity probabilities satisfy
\(q_z\geq(1-\nu_a g_{\max})/2\), so the conservative choice
\(q_0=(1-g_{\max})/2\) is valid; moreover,
\(H_{\max}=1/(1-g_{\max}^2)\) and \(S_\rho^2=2\).
The constant is therefore finite on every fixed compact region
\(g_{\max}<1\), but diverges as the pure-state boundary is approached.

For the sharper small-error form one may replace $g_{\max}$ by
$\nu_a g_{\max}$ in Eq.~\eqref{eq:gjc-uniform-constant}, provided
$\varepsilon_{E_A,E_B}\leq(1-\nu_a g_{\max})/4$.  Fixed mixtures of GJC phases obey
the same bound because their Fisher matrices add with
parameter-independent allocation weights.

If the two local effective channels have errors $\varepsilon_{A,E_A}$ and
$\varepsilon_{B,E_B}$, the joint product channel satisfies
\begin{equation}
\varepsilon_{E_A,E_B}
\leq
\varepsilon_{A,E_A}+\varepsilon_{B,E_B}.
\label{eq:local-to-joint-diamond}
\end{equation}
Furthermore, for $g$-independent flagged-pattern weights
$w_E\geq0$ with $\sum_Ew_E=1$, suppose every pattern satisfies the same
conditional-channel assumptions and has the common success weight $s_0$.
Here $E$ abbreviates $(E_A,E_B)$ and
$\varepsilon_E=\varepsilon_{E_A,E_B}$.
Then
\begin{equation}
\left\lVert\overline F^{\mathrm{GJC}}-\overline F_0^{\mathrm{GJC}}
\right\rVert{}_{\mathrm{op}}
\leq
s_0 C_F(g_{\max})
\sum_E w_E\varepsilon_E+O(\epsilon^2).
\label{eq:flagged-average-stability}
\end{equation}
Here $\overline F_0^{\mathrm{GJC}}$ denotes the same ideal comparator averaged
with the weights $w_E$.  If a pattern instead has a different but
$g$-independent success weight $s_E$, then
\begin{align}
\Delta_E^{(s)}
&\equiv
s_EF_{E,\mathrm{cond}}^{\mathrm{GJC}}
-
s_0F_{0,\mathrm{cond}}^{\mathrm{GJC}},
\nonumber\\
\left\lVert\Delta_E^{(s)}\right\rVert{}_{\mathrm{op}}
&\leq
s_EC_F(g_{\max})\varepsilon_E
+
\left\lvert s_E-s_0\right\rvert H_{\max}.
\label{eq:gjc-success-mismatch}
\end{align}
If $s_E$ depends on $g$, the normalized accepted map is nonlinear and the
fixed-receiver theorem must instead be applied to the complete instrument
containing success, failure, rejection, and erasure-flag outcomes.

\subsection{Independent erasure and average information per incident mode}
\label{sec:average-fi}

For an unrestricted exact code, suppose each physical memory in a local block
of size $n$ is independently erased with probability $p$, and let the code
correct every pattern containing at most $t$ erasures.  This is a benchmark
model, not an achievable strictly covariant code when $t\geq1$ and the logical
charge is nontrivial, by Lemma~\ref{lem:covariant-no-go}.  The probability that
one block is correctable is
\begin{equation}
P_{\mathrm{block}}(n,t,p)
=
\sum_{k=0}^{t}\binom{n}{k}p^k(1-p)^{n-k}.
\label{eq:block-success}
\end{equation}
If the two blocks suffer independent erasures and both must be recovered, then
\begin{equation}
P_{AB}=P_{\mathrm{block},A}P_{\mathrm{block},B}.
\label{eq:two-block-success}
\end{equation}

\paragraph*{Assumption (flagged rejection of uncorrectable patterns).}
The erasure locations are known.  Trials outside the correctable set are rejected by a rule independent of $g$.  Correctable trials are decoded exactly.  The aggregate logical processing efficiency $\eta_{\mathrm{log}}$ is also independent of $g$.

Under this assumption, the average CFI per incident input temporal mode is
\begin{equation}
\overline F_\delta^{\mathrm{enc}}
=
P_{AB}\eta_{\mathrm{log}}F_\delta^{(0)},
\label{eq:encoded-average-fi}
\end{equation}
where
\begin{equation}
F_\delta^{(0)}
=
\frac{\epsilon\eta_c\eta_a}{2}
\frac{\nu_a^2\boldsymbol{u}_\delta\boldsymbol{u}_\delta^{\mathsf T}}
{1-\nu_a^2(\boldsymbol{u}_\delta^{\mathsf T}\boldsymbol{g})^2}
\label{eq:baseline-F0}
\end{equation}
is the ideal logical GJC CFI per incident input mode, including the modeled
capture and ancillary-availability factors.  It does not include reference
generation latency or failed-attempt consumption, detector capacity, reset,
or any other platform-total resource in
Appendix~\ref{sec:platform-contract}.

For an unencoded two-memory state, erasure of either local logical subsystem destroys the cross-aperture coherence.  Under the same flagged-erasure model,
\begin{equation}
\overline F_\delta^{\mathrm{unc}}
=(1-p)^2F_\delta^{(0)}.
\label{eq:uncoded-average-fi}
\end{equation}
Therefore, under a common attempt clock and the same modeled factors, the
dimensionless encoded retention exceeds the bare retention precisely when
\begin{equation}
P_{AB}\eta_{\mathrm{log}}>(1-p)^2.
\label{eq:advantage-condition}
\end{equation}
If the encoded and unencoded protocols operate at rates $R_{\mathrm{enc}}$
and $R_0$, respectively, the corresponding algebraic per-unit-time condition is
\begin{equation}
R_{\mathrm{enc}}P_{AB}\eta_{\mathrm{log}}
>
R_0(1-p)^2.
\label{eq:rate-advantage}
\end{equation}
The rates in Eq.~\eqref{eq:rate-advantage} cannot be inferred from circuit
depth alone: they must be produced by the complete schedule and common
capacity vector in Eq.~\eqref{eq:platform-information-rate}.
For approximate recovery, coherence loss and bias must be included explicitly
rather than absorbed into a success probability.  The channel-to-CFI
conversion is given by Theorem~\ref{thm:fi-continuity} and
Corollary~\ref{cor:gjc-fi-continuity}, with the proof and the special
visibility-transfer form collected in Appendix~\ref{sec:approximate}.

More generally, when every flagged pattern has its own parameter-independent approximate recovery and the flag is retained in the classical record, the exact aggregate Fisher matrix is
\begin{equation}
\begin{aligned}
\overline F_C(p)
={}&
\sum_{E_A,E_B}
p^{|E_A|+|E_B|}
\\[-2pt]
&\times
(1-p)^{2n-|E_A|-|E_B|}
F_C(E_A,E_B).
\end{aligned}
\label{eq:all-pattern-cfi}
\end{equation}
No binomial coefficient appears in Eq.~\eqref{eq:all-pattern-cfi} because the
sum runs over individual erasure sets.  Uninformative patterns contribute the
zero matrix rather than being removed by postselection.  This is the
expression used in the all-pattern numerical calculation below under the
idealized instrument of independent and identically distributed (i.i.d.) erasures with perfectly known flags.  A physical
instrument with false flags must retain the reported flag as an additional
outcome and replace these Bernoulli weights by calibrated joint branch
probabilities.

\section{Finite-size mechanism and operational realization}
\label{sec:finite-summary}

We now restrict the claims to finite memory blocks and explicitly constructed
decoders.  The asymptotic decoupling questions, complete numerical protocols,
and diagnostic tables are collected in the appendices.  The main text retains
only the results needed for the operational chain
\begin{equation}
\text{scramble}\;\to\;\text{recover}\;\to\;F^{\mathrm{GJC}}.
\label{eq:finite-operational-chain}
\end{equation}

\subsection{Leakage redistribution and decoder choice}

We use \(n=5\) as a controlled finite-size test bed.  It is the smallest odd
block with equal adjacent central charge-sector dimensions,
\(\binom52=\binom53=10\), while still permitting a nontrivial \(k=2\)
erasure test.  At this size we can exhaustively evaluate all
\(\binom52^2=100\) fixed-\(k\) two-node position pairs, all
\(32^2=1024\) two-node flagged-erasure events, and every nonempty local
compilation branch in the frozen implementation batch.  This makes the
entire finite-size chain auditable; it is not a scaling or universality claim.

We begin with the architecture-specific number-conserving result.  Fix
$k_A=k_B=2$.  For 30 complete
number-conserving brickwork trajectories and all 100 two-node erasure-position
pairs, Fig.~\ref{fig:finite-size-mechanism}(c) shows that increasing the
encoder depth suppresses the position-dependent high tail of the
task-directed tangent leakage \(\epsilon_{\rm tan}\), the root-mean-square
norm of the erased-environment visibility derivatives defined in
Eq.~\eqref{eq:tangent-leakage}.  Between depths $L=0$ and
$L=15$, its position 95th percentile (q95) falls from $1$ to $0.3234$.  Correspondingly,
Fig.~\ref{fig:finite-size-mechanism}(d) shows that the worst-position recovered
SSR-QFI retention rises from zero to approximately $0.15$.  The lower-tail
and worst-position values improve relative to the unscrambled and shallow
cases and then saturate within the sampled depths; from $L=10$ to $L=15$ both
change slightly downward.  The mean recovered SSR-QFI is also nonmonotone.
Thus the finite-size role of scrambling is risk redistribution over erasure
locations, rather than a universal increase of every information metric.

For context only, Figs.~\ref{fig:finite-size-mechanism}(a) and~\ref{fig:finite-size-mechanism}(b) compare position-averaged and position 5th percentile (q05)
QFI retentions for unrestricted local brickwork encoders at $L=3n$ with
local-Haar benchmarks, obtained from independent Haar-random local encoding
isometries.  At the sampled points $n=3,5,7$, the brickwork values
are numerically close to the corresponding Haar values and display a
finite-size drop near $k/n=1/2$.  Charge conservation is absent from these
two panels, the vertical line is only a guide, and the data establish neither
an asymptotic threshold nor correctability of the physical covariant encoder.
These unrestricted panels use 30 brickwork encoders and 60 independent
local-Haar encoders for each sampled size.  Erasure-position pairs are
enumerated when \(\binom{n}{k}^{2}\leq120\); for
\(n=7\), \(k=2,3,4,5\), 120 pairs per encoder are sampled without
replacement.  By contrast, the covariant mechanism study in panels (c,d)
uses the same 30 complete depth trajectories and all 100 two-node
erasure-position pairs at every plotted depth.

\begin{figure*}[t]
\centering
\includegraphics[width=\textwidth]{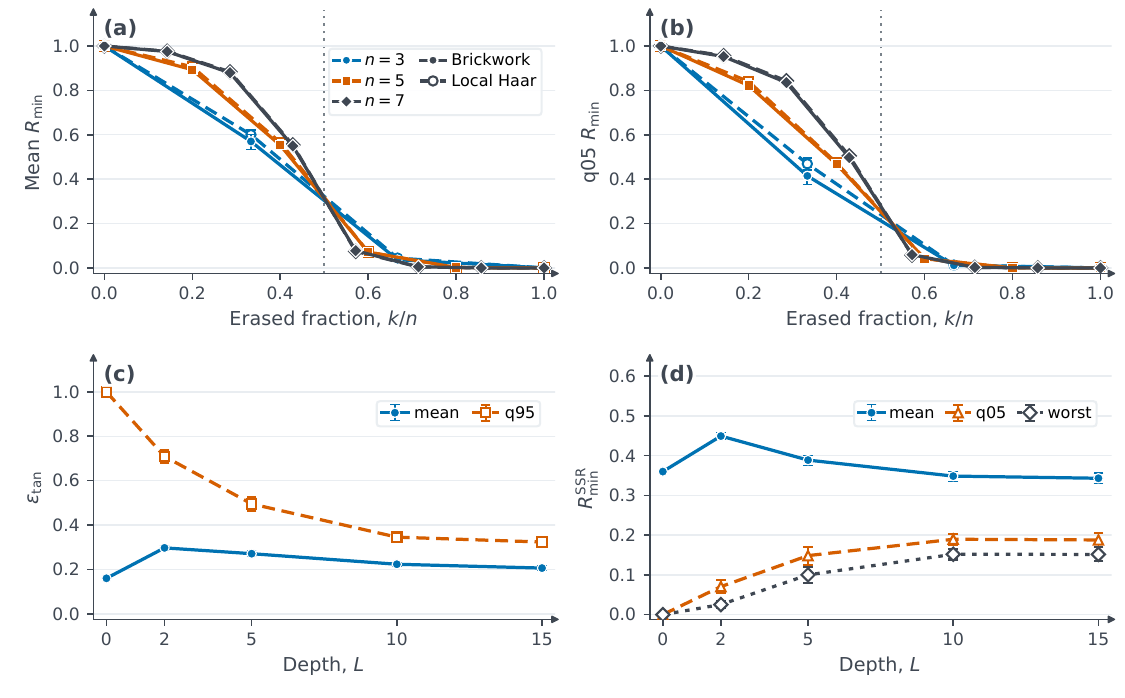}
\caption{Finite-size scrambling and recovery diagnostics under flagged
erasures.  (a) Position-mean generalized worst-direction QFI retention
\(R_{\min}\) versus the erased fraction \(k/n\).
(b) Position-q05 retention, computed within each encoder before ensemble
averaging.  Panels (a,b) compare unrestricted local brickwork encoders
(solid, filled) at \(L=3n\) with local-Haar benchmarks (dashed, open) for
\(n=3,5,7\); charge conservation is not imposed.  Each size uses 30
brickwork and 60 Haar encoders.  Position pairs are enumerated when
\(\binom nk^2\leq120\), otherwise 120 are sampled without replacement.
The line at \(k/n=1/2\) is a finite-size guide, not an erasure threshold.
(c) Position mean and q95 of the task-directed tangent leakage
\(\epsilon_{\mathrm{tan}}\) versus number-conserving brickwork depth.
(d) Position mean, q05, and worst-position recovered SSR-QFI retention
\(R_{\min}^{\mathrm{SSR}}\), normalized to the no-loss SSR baseline.
Panels (c,d) use local \(U(1)\)-covariant encoders with
\(n=5\), \(k_A=k_B=2\), 30 complete encoder trajectories, and all
100 two-node erasure-position pairs.  Increasing depth suppresses the
high-leakage tail and improves lower-tail and worst-position recovery over
the shallow cases, but the means are nonmonotone and the improvement
saturates within the sampled depths.  Whiskers in all panels are pointwise
95\% percentile cluster-bootstrap intervals from 5000 resamples.  Panels
(a,b) resample complete encoders jointly across all plotted \(k\) values
within each family and \(n\); panels (c,d) resample complete five-depth
encoder trajectories jointly across depth.  Erasure positions within an
encoder are not treated as independent replicates.  The \(n=7\) intervals
are conditional on the frozen 120-position samples, and none of the
intervals is a simultaneous confidence band.}
\label{fig:finite-size-mechanism}
\end{figure*}

Having identified the leakage--recovery tradeoff, we next close the
operational gap between retained QFI and an implementable GJC receiver.  For
operational recovery we use a local, pattern-dependent,
$U(1)$-covariant CPTP map optimized for the two logical coherence tangent
operators.  On an independent validation set at $n=5$, $L=10$, and
$k_A=k_B=2$, the tangent-weighted recovery improves the mean QFI and GJC-CFI
retentions by $5.50\%$ and $5.15\%$, respectively, relative to the
maximum-entanglement-fidelity recovery [Fig.~\ref{fig:operational-recovery}(a)].
Its worst-position QFI retention improves by $14.98\%$
[Fig.~\ref{fig:operational-recovery}(b)].  Encoder-paired absolute
Tangent-minus-Max-fid.\ differences are \(0.01861\)
\([0.01643,0.02071]\) for mean QFI, \(0.01843\)
\([0.01617,0.02064]\) for mean GJC-CFI, \(0.01637\)
\([0.01259,0.02033]\) for worst-position QFI, and \(0.01747\)
\([0.01342,0.02169]\) for worst-position GJC-CFI, where brackets are
pointwise 95\% percentile-bootstrap intervals.

A fixed uniform four-phase GJC receiver is used throughout the final
conditional-retention calculation.  As shown in
Fig.~\ref{fig:operational-recovery}(c), it changes
the mean measurement efficiency only from $0.4754$ to $0.4761$ relative to
the conventional two-phase receiver, but raises the mean within-encoder fifth
percentile from $0.3813$ to $0.4295$ and the mean within-encoder sampled
minimum from $0.3222$ to $0.4100$.  Its benefit
is therefore reduced worst-direction measurement mismatch relative to the
QFI retained after recovery, rather than a substantial mean-information
gain.  The paired four-phase-minus-two-phase changes are \(0.000632\)
\([0.000559,0.000700]\) for the mean, \(0.04815\)
\([0.04015,0.05579]\) for q05, and \(0.08786\)
\([0.07975,0.09613]\) for the sampled minimum.  The mean within-encoder
fifth percentile of the
\emph{absolute} CFI retention changes from $0.05585$ to $0.05504$ and hence
does not improve; its paired change is \(-0.000813\)
\([-0.001432,-0.000147]\).  The four-phase claim is specifically a
lower-tail robustness gain in measurement efficiency, not a large mean gain
or a uniform gain in absolute CFI\@.  Detailed
numerical tables, compatibility checks, and all-pattern comparisons are
provided in Appendix~\ref{sec:numerics}.

The recovery ablation uses 10 held-out encoders, whereas the receiver-design
comparison uses a separate eight-encoder set.  In both cases, summaries are
formed within each encoder before ensemble averaging.  Whiskers in
Fig.~\ref{fig:operational-recovery} are pointwise 95\% percentile-bootstrap
intervals for the encoder-level means.  Method-difference intervals use
joint resampling of the same encoders within each comparison.  Each analysis uses 5000 draws from two fixed pseudorandom-number-generator
(PRNG) streams; their roles and integer initialization values are recorded in
the reproducibility archive.  Erasure positions
and position--visibility cells are not treated as independent replicates.
The intervals are pointwise and carry neither simultaneous coverage nor a
multiplicity correction.

\begin{figure*}[t]
\centering
\includegraphics[width=\textwidth]{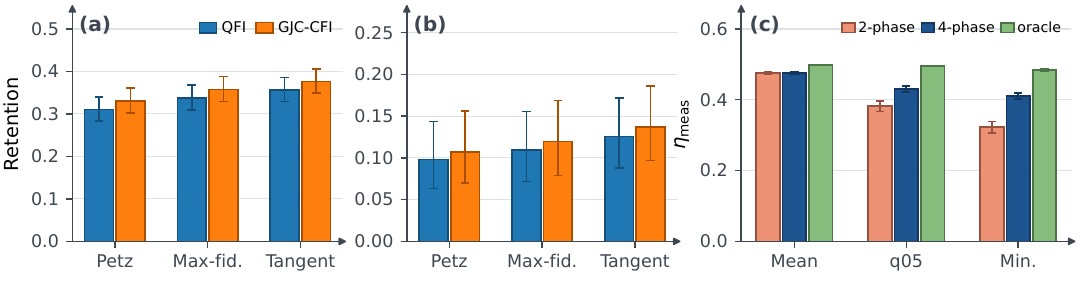}
\caption{Operational validation of local recovery and fixed-phase GJC
readout.  (a) Mean SSR-QFI and GJC-CFI retention for the transpose/Petz,
SSR-covariant maximum-entanglement-fidelity (Max-fid.), and
tangent-weighted (Tangent) recovery maps.
(b) Corresponding worst-position retention, obtained by minimizing over
erasure-position pairs within each encoder and then averaging the 10
independent encoder summaries.  Panels (a,b) use
\(n=5\), \(L=10\), \(k_A=k_B=2\), and all 100 position pairs; QFI and
CFI are normalized separately to their lossless baselines.
(c) Relative measurement efficiency of fixed two-phase, uniform four-phase,
and 16-point phase-grid-oracle receivers with the tangent-weighted decoder
frozen.  The mean, q05, and minimum are formed within each of eight
independent encoders over 100 position pairs and 25 visibilities before
averaging.  The oracle is a discrete
phase-informed benchmark, not an implementable fixed receiver or a
continuous true-phase oracle.  In every panel, whiskers denote pointwise
95\% percentile-bootstrap intervals for encoder-level means.  The paired
Tangent-minus-Max-fid.\ and four-phase-minus-two-phase intervals jointly
resample the same encoders within each comparison.  Each analysis uses 5000
draws; erasure positions and visibility-grid cells are not resampled as
independent observations.  The intervals are not simultaneous and have no
multiplicity correction.  All quantitative axes include zero.}
\label{fig:operational-recovery}
\end{figure*}

\subsection{Common-gate-set encoder online-gate costs and compiled recovery}

Panels~\ref{fig:operational-recovery}(a,b) are a controlled \(L=10\),
fixed-\(k\) decoder ablation; they do not select the final encoder.  The
all-pattern candidate comparison instead retains every flagged event and
uses equal statistical budgets: 30 paired \(L=2/L=10\) trajectories and 30
charge-Haar target pairs.  A seed audit found that, within each sample,
the trajectory-node-\(B\) and charge-Haar-node-\(A\) generators had been
assigned the same integer seed.  Their transformations consume random
numbers differently, but independence is not established.  We therefore
withdraw the earlier independent-ensemble Haar interval and use the
sample-index-paired follow-up below.  The screen remains a descriptive
candidate-selection calculation (Table~\ref{tab:all-pattern-rate}); equal
statistical budget alone does not imply equal state-preparation or hardware
cost.

The frozen encoder-online-gate follow-up uses one common abstract
nearest-neighbor \(U(1)\) gate alphabet.  Each node contains five storage
qubits with two fixed ancillary excitations and no additional synthesis work
ancilla.  Starting from
\(V_{\rm in}=(|01100\rangle,|11100\rangle)\), the compiler targets only the
two codeword columns needed for the logical isometry, not a complete
charge-sector Haar unitary.  All 60 frozen single-node charge-Haar targets
and one symmetric-Dicke isometry reused at the two nodes pass the depth-\(7\)
acceptance test.  The maximum coherent-isometry infidelity is
\(8.19\times10^{-12}\), the maximum phase-aligned Frobenius error is
\(2.86\times10^{-6}\), and the maximum absolute target-versus-circuit
worst-work-point \(R_C\) difference is \(1.78\times10^{-6}\).

For the newly evaluated controls, the 30 depth-\(7\) brickwork pairs, 30
compiled charge-Haar pairs, and one compiled Dicke pair comprise 122 local
codewords.  Recomputing all 31 nonempty local recovery branches therefore
requires \(122\times31=3782\) local semidefinite programs (SDPs).  The frozen depth-\(2\) recovery records are reused from the accepted
candidate-comparison calculation and are not
included in this count.  Each encoder pair retains 1024 two-node flagged
events, five visibility points, and the complete erasure-probability grid.
At \(p=0.30\), the mean retentions are \(0.561392\), \(0.509252\),
\(0.506859\), and \(0.418808\) for \(L=2\), \(L=7\), compiled charge-Haar,
and compiled Dicke, respectively.  A 10,000-draw sample-paired cluster
bootstrap with a max-\(t\) correction over the prespecified five-contrast
family gives \(L2-L7=0.052140\)
\([0.040052,0.064228]\), \(L2-\mathrm{cH7}=0.054533\)
\([0.038824,0.070243]\), and \(L2-\mathrm{D7}=0.142584\)
\([0.134519,0.150649]\), where brackets are simultaneous 95\% intervals for
the prespecified contrast family; the Dicke comparator is held fixed.
Thus \(L=2\) is a sampled encoder-online-gate Pareto candidate at the primary
work point for these realized circuits: it uses depth 2 and four two-qubit
gates per node, versus depth 7 and 14 gates for each control.  The depth-7
costs are verified feasible upper bounds, not minimal-synthesis
certificates.  Recovery, reference, latency, noise, flags, feed-forward,
reset, and detector resources remain unmatched.  Offline compiler search and
continuous-control precision are also excluded; detailed definitions and the
complete ledger appear in Appendix~\ref{sec:abstract-encoder-resource}.

The compiled-decoder instance uses a previously frozen depth-$2$
implementation batch generated from PRNG initialization values distinct from
those used for either candidate ensemble in the equal-budget and
common-gate-set comparisons.  Distinct initialization values alone are not
treated as proof of statistical independence.  The batches are analyzed
separately and are not pooled.  Their separately estimated Monte Carlo means at
\(p=0.30\) are \(0.561392\) in the equal-budget candidate comparison and
\(0.558045\) in the implementation batch used below; their difference is not
a compilation loss.
The implementation batch contains
number-conserving encoders on two $n=5$ memory blocks.  For each of 30
separately generated encoder pairs, the 31 nonempty local erasure sets at
each node are
recovered separately, giving \(30\times2\times31=1860\) nontrivial local
branches.  Every branch is represented by a support-pruned Stinespring
isometry, namely an isometric dilation of the recovery channel, and compiled
into a linear nearest-neighbor circuit of two-qubit \(U(1)\)-conserving
gates.  These gates are ``native'' only to the chosen
abstract gate alphabet; they are not claimed to be native to a calibrated
experimental platform.  All branches pass numerical checks of the Choi
matrix, complete positivity, and trace preservation.  The Choi matrix is the
standard operator representation of a quantum channel.  Eighteen branches
have a false optimizer termination flag but pass a stricter posterior residual
criterion with normalized Choi-matrix Frobenius residual below $10^{-7}$; the largest such
residual is $4.26\times10^{-8}$.  The largest residual over the complete
ensemble is $9.90\times10^{-5}$ and belongs to a normally terminated branch.

Each encoder pair retains the complete Cartesian product of the two local
flagged-erasure sets in Eq.~\eqref{eq:all-pattern-cfi}.  At zero imposed
effective attenuation, the minimum over the ten tested erasure probabilities
of the
30-pair mean within-pair worst-visibility CFI ratio is
$0.999999987$ for gate-set-compiled versus compressed recovery,
$1.000000004$ for compressed versus exact recovery, and
$0.999999995$ for gate-set-compiled versus exact recovery.  Deviations of
order $10^{-8}$ around unity are numerical solver and reconstruction
precision, not evidence of a physical information gain.  The
event-probability-weighted expected parallel decode depth is largest near
intermediate erasure probabilities; at $p=0.3$ its mean over 30 encoder pairs
is $5.0976$.  Full depth, gate-count, and verification statistics appear in
Appendix~\ref{sec:native-hardware}.

\section{Sampled conditional retention and partial-timing comparison}
\label{sec:conditional-window}

For the separately frozen \(L=2\) implementation batch---not for all designs
in the encoder-online-gate comparison---we use a transparent synthetic
output-coherence attenuation envelope.  After ideal compiled recovery, the
off-diagonal elements of a local branch containing \(m_E\) abstract compiled
two-qubit gates are multiplied by \((1-2e)^{m_E}\), where \(e\) indexes the
output-coherence attenuation.  No physical phase-flip or other microscopic
channel is inserted after each gate, and \(e\) is not a calibrated per-gate
error probability.
The two nodes are assumed to operate in parallel.  If one compiled gate layer
is assigned a fraction \(r=\tau_2/T_{\mathrm{bare}}\) of an otherwise
unspecified bare protocol cycle, the partial timing discount is
\begin{equation}
\begin{aligned}
\gamma(p,r)
&=
\left\{
1+r\!\left[
D_{\mathrm{enc}}
+\mathbb E_p\max(D_A,D_B)
\right]
\right\}^{-1},
\\[-2pt]
D_{\mathrm{enc}}&=2.
\end{aligned}
\label{eq:main-rate-envelope}
\end{equation}
Here \(D_A,D_B\) are the realized local compiled-decoder depths and
\(\mathbb{E}_p\) averages independent Bernoulli erasures at the two nodes.
The envelope uses one common synthetic attenuation index and a multiplicative
output factor; the resulting coherence factors are branch dependent through
\(m_E\).  It excludes amplitude damping, leakage,
crosstalk, correlated errors, and false erasure flags.  Its serial time
consists only of the encoder depth and the larger of the two decoder depths;
it excludes ancillary
reference preparation, reset, flag acquisition, classical feed-forward,
retrieval, and detection latency.
The ratios below are normalized to the ideal post-capture GJC CFI, so the
common factor \(\epsilon\eta_c\eta_a/2\) cancels; ``per captured mode'' is
conditional on reaching the encoded interface and is neither per incident
mode nor a fixed-total-resource comparison.  This cancellation assigns the
same one-reference-photon factor to encoded and bare trials; it does not make
the reference free, and the cancellation is invalid for a wall-clock
comparison if reference or detector supply is a bottleneck.
The conditional normalized excess retention, with this partial timing
discount, is
\begin{equation}
G(p,e,r)
=
\gamma(p,r)R_C(p,e)-(1-p)^2,
\label{eq:main-net-gain}
\end{equation}
where $R_C$ is the generalized worst-direction GJC-CFI ratio at the worst of
five visibility work points.

For every encoder pair in the frozen implementation batch, \(R_C\) retains
every two-node flagged-erasure event with its i.i.d.\ model probability.  The accepted simulator calculation is a model-only computational
reproduction using the fixed archived streams on the same 30 complete
encoder-pair/compilation realizations as the frozen parent calculation.
Those complete realizations are treated as exchangeable clusters.  The 1024
erasure events, five visibility work points, and all \((p,e,r)\) evaluations
within a realization are deterministic repeated evaluations, not independent
samples.  For the data-informed, post hoc conditional-gain reanalysis,
20,000 realization-cluster bootstrap draws use a fixed pseudorandom stream
recorded in the reproducibility archive,
NumPy's \texttt{linear} quantile rule, and the same resampled
realization-index vector across all 360 cells in each draw.  A two-sided
studentized max-\(t\) statistic supplies the common critical value
\(c_{.95}=2.616920\).

Across the evaluated family, 191\ cells have positive
simultaneous lower endpoints and 130\ cells have
negative simultaneous upper endpoints.  The resulting intervals are nominal
finite-sample approximate simultaneous 95\% intervals for that finite
360-cell family; finite-sample coverage is not guaranteed.  They provide no
coverage between evaluated grid points and do not certify a continuous
region in \(p\), \(e\), or \(r\).
The displayed $r_{\rm crit}$ is a plug-in sensitivity summary: the
30-pair mean $R_C$ and mean expected parallel decode depth are inserted into
Eqs.~\eqref{eq:main-rate-envelope}--\eqref{eq:main-net-gain}, and
$r_{\rm crit}$ is their zero crossing.  Thus the plug-in envelope is positive
for $0\leq r<r_{\rm crit}$, not at the crossing itself.  Because the rate
factor is nonlinear in depth, this plug-in crossing need not equal the exact
zero of the pairwise-averaged gain used for the bootstrap intervals.

The max-\(t\) correction applies to the pairwise mean \(G\) values on the
sampled grid.  It does not turn the plug-in \(r_{\rm crit}\) curves into
confidence limits or continuous boundaries.

Figure~\ref{fig:conditional-window} separates four nested diagnostics.
Panel (a) shows the all-pattern information retained after compiled recovery
and the effective output attenuation, before the partial timing factor
\(\gamma\) is applied.  Panel (b) resolves the zero-noise
gate-set-compiled-to-compressed compilation
loss on an explicitly magnified scale; the prespecified ratio floor \(0.99\)
would appear at \(10^6\) on that scale and is therefore outside the plotted
range.  Panel (c) gives the supremum of
nonnegative gate-layer times with positive plug-in gain, evaluated through
its zero crossing; a dash at zero means
that the plug-in gain is already nonpositive at $r=0$.  Panel (d) fixes
\(p=0.3\) and resolves the sampled dependence on both sensitivity-envelope
parameters.  Cell colors encode pairwise mean \(G\); marker classes are
determined by the simultaneous lower endpoints from the full 360-cell
conditional-gain family.  Thus information retention, numerical compilation accuracy,
partial timing feasibility, and the sampled two-parameter comparison remain
visually distinct.

\begin{figure*}[t]
\centering
\includegraphics[width=\textwidth]{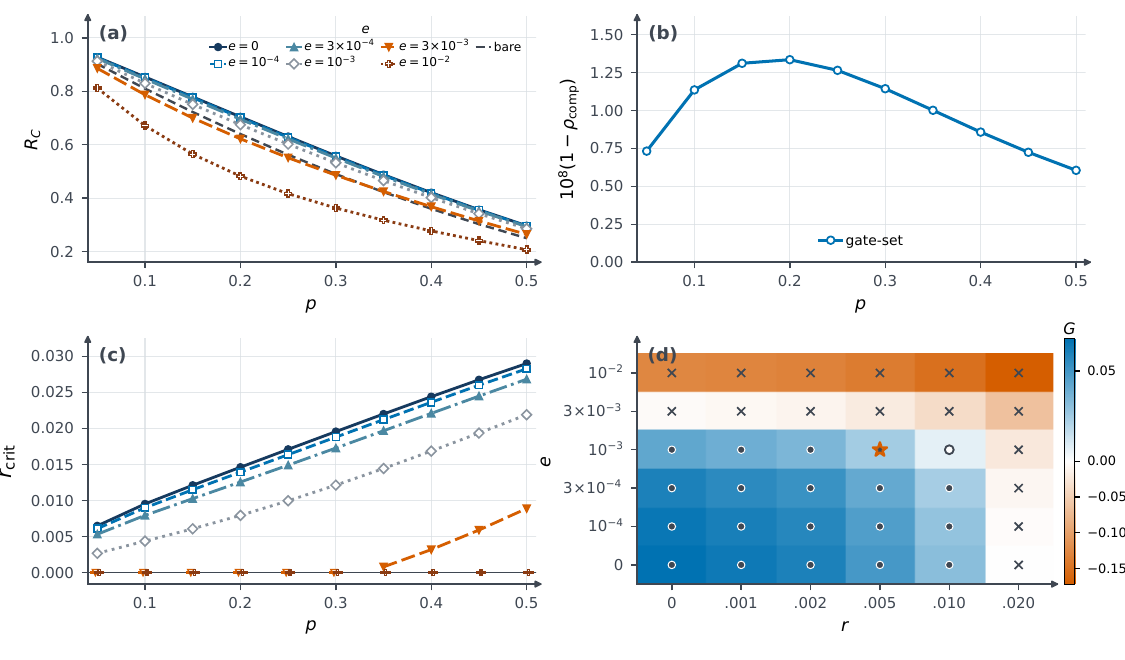}
\caption{Gate-set compilation and conditional-retention diagnostics
for the frozen implementation batch of 30 separately generated \(n=5\),
\(L=2\) encoder pairs.  This batch uses PRNG initialization values distinct from those of the
equal-statistical-budget candidate and common-gate-set comparisons in
Table~\ref{tab:all-pattern-rate}; the batches are analyzed separately and
not pooled, without inferring statistical independence from that distinction
alone.  Every encoded point includes
every flagged-erasure event with its i.i.d.-model probability weight and uses
the worst of five visibility work points.
(a) All-pattern GJC-CFI ratio \(R_C(p,e)\) before the processing-rate factor.
Here \(e\) indexes a synthetic output-coherence attenuation: for each branch,
the off-diagonal elements of the recovered compiled logical image are
multiplied by \((1-2e)^{m_E}\), where \(m_E\) is the abstract compiled
recovery-gate count.  This is not a phase-flip channel inserted after each
gate, and \(e\) is not a calibrated platform error probability.  The charcoal
dashed curve is bare storage \((1-p)^2\).
(b) Magnified zero-noise gate-set-compiled-to-compressed compilation loss
\(10^{8}(1-\rho_{\rm comp})\).  The in-panel label `gate-set' denotes
compilation in the chosen abstract gate alphabet, not a calibrated
platform-native implementation.  Deviations on this scale are
numerical compilation residuals, not physical information gain.
(c) Plug-in zero crossing \(r_{\rm crit}\) formed from the 30-pair mean
\(R_C\) and mean expected decode depth; it is not the zero of the
pairwise-averaged bootstrap estimand.  A crossed marker at \(r=0\) is a
sentinel for an envelope already nonpositive at zero.  Sentinels are not
connected to later positive roots, so the resulting break is intentional.
(d) Pairwise mean net gain \(G\) on the sampled \((e,r)\) grid at \(p=0.3\).
Filled circles mark the 19\ cells with positive
simultaneous 95\% lower endpoints; open circles mark the
1\ additional cell with a positive mean but
nonpositive simultaneous lower endpoint; crosses mark the
16\ cells with nonpositive means.  The orange
outlined star highlights \((e,r)=(10^{-3},0.005)\); its inferential class is
the simultaneous-positive (filled-circle) class.  The intervals are nominal finite-sample
approximate simultaneous 95\% intervals obtained with 20,000 draws from the
two-sided studentized realization-cluster bootstrap max-\(t\) procedure over
all 360 conditional-gain cells frozen for the post hoc reanalysis.
Finite-sample coverage is not guaranteed.  The marker classes concern
evaluated cells only; neither symbols nor connecting curves define an
interpolated or continuous confidence boundary.  The comparison is
conditional per captured mode, uses only the partial timing factor in
Eq.~\eqref{eq:main-rate-envelope}, and is not normalized by incident modes,
wall-clock time, or fixed total resources.  It is therefore not a calibrated
hardware threshold or platform advantage.}
\label{fig:conditional-window}
\end{figure*}

At \(e=10^{-3}\) and \(r=0.005\), the evaluated erasure-probability set with
positive simultaneous 95\% lower endpoints is
\(\{0.20,\,0.25,\,0.30,\,0.35,\,0.40,\,0.45,\,0.50\}\), whereas the set without a positive
simultaneous lower endpoint is
\(\{0.05,\,0.10,\,0.15\}\).  At the target
\((p,e,r)=(0.30,10^{-3},0.005)\), the pairwise mean gain is
\(G=+0.023976\), with simultaneous 95\% interval
\([0.013095,\,0.034857]\) and inferential class
the simultaneous-positive (filled-circle) class.  At \(e=3\times10^{-3}\) and \(r=0.005\), the
corresponding simultaneous-positive and not-simultaneous-positive evaluated
\(p\)-sets are \(\varnothing\) and
\(\{0.05,\,0.10,\,0.15,\,0.20,\,0.25,\,0.30,\,0.35,\,0.40,\,0.45,\,0.50\}\), respectively.  These
statements concern only the ten frozen values
\(p\in\{0.05,0.10,\ldots,0.50\}\); they do not define a continuous or
platform-calibrated threshold.

Table~\ref{tab:main-conditional-window} gives the numerical $p=0.3$ slice at
$r=0.005$.  Here
$\gamma_{\mathrm{crit}}=(1-p)^2/R_C$ is the minimum processing-rate ratio
required for positive plug-in gain.

\begin{table*}[t]
\caption{Simultaneous conditional-retention comparison on the \(p=0.3\),
\(r=0.005\) slice of the frozen synthetic grid, using the partial timing
proxy.  \(e\) is the output-coherence attenuation index defined in the text,
not a calibrated per-gate channel probability.  \(R_C\) is the mean over the
same 30 complete encoder-pair/compilation realizations of the
worst-work-point CFI ratio, and
\(\gamma_{\mathrm{crit}}=(1-p)^2/R_C\) is the plug-in rate-ratio requirement.
The final two columns give the pairwise mean gain in
Eq.~\eqref{eq:main-net-gain} and its nominal finite-sample approximate
simultaneous 95\% interval.  The two-sided studentized cluster-bootstrap
max-\(t\) critical value is computed over all 360 conditional-gain cells
frozen for the post hoc reanalysis, not only the four rows shown.}
\label{tab:main-conditional-window}
\begin{ruledtabular}
\begin{tabular}{ccccc}
synthetic \(e\) & \(R_C\) & \(\gamma_{\mathrm{crit}}\) & \(G\) &
simultaneous 95\% interval\\
\hline
\(0\) & 0.558045 & 0.878065 & +0.048930 & $[0.037936,\,0.059923]$\\
\(10^{-3}\) & 0.532205 & 0.920697 & +0.023976 & $[0.013095,\,0.034857]$\\
\(3\times10^{-3}\) & 0.485664 & 1.008927 & -0.020969 & $[-0.031607,\,-0.010331]$\\
\(10^{-2}\) & 0.362879 & 1.350312 & -0.139546 & $[-0.149110,\,-0.129982]$\\
\end{tabular}
\end{ruledtabular}
\end{table*}

Equation~\eqref{eq:main-net-gain} compares one encoded logical mode with one
bare logical mode and assumes that the four additional memories would not
otherwise store independent astronomical modes.  We therefore also evaluate
the conservative memory-only \(1/n=1/5\) throughput estimand for the same
frozen \(L=2\) batch and the same 30 complete
encoder-pair/compilation-realization clusters.  It is a separate 360-cell data-informed, post hoc family analyzed with
20,000 draws from a fixed pseudorandom stream, NumPy's \texttt{linear}
quantile rule, and its own two-sided studentized cluster-bootstrap max-\(t\)
critical value.  The reproducibility archive records the stream role and integer PRNG
initializer; the integer is an implementation identifier rather than a
calendar date or preregistration label.  It has
0\ cells with positive simultaneous lower endpoints
and 360\ cells with negative simultaneous upper
endpoints.  The two 360-cell families are bootstrapped separately, and no
joint 720-cell coverage is claimed.  Their nominal finite-sample approximate
simultaneous intervals are not pooled.  Thus this frozen implementation batch
does not beat five parallel bare memories over the evaluated
\(p\in\{0.05,0.10,\ldots,0.50\}\) grid when the bare memories are granted
sufficient reference and detector throughput.  This remains a memory-only
upper benchmark, not a fixed-total-resource result; a system comparison would
require a justified capture bottleneck or spare memory, or a higher-rate
architecture consistent with covariance--accuracy tradeoffs.

\section{Discussion: hierarchy of evidence and limitations}
\label{sec:main-limitations}

The claim set has three distinct evidence levels.  At the analytical level,
exact recovery is used only as a benchmark, strict covariant exact correction
is excluded by Lemma~\ref{lem:covariant-no-go}, and
Theorem~\ref{thm:fi-continuity} proves a receiver-level consequence of any
uniform logical diamond-error bound.  The appendices additionally give the
unrestricted fixed-pattern Haar estimate, the exact one-reference-photon SSR
factor, and a causal-depth necessary condition.  None of those results proves
that the specified finite-depth covariant ensemble achieves uniform
decoupling over all erasure patterns.

At the numerical level, the scaling study reaches only $n\leq7$, and the
compiled-pipeline calculation uses $n=5$.  The near-half crossover and closeness to
Haar or charge-Haar samples are finite-size trends, not convergence,
capacity, or threshold results.  The $L=10$, fixed-$k$ recovery comparison is
a decoder ablation, whereas retaining the final \(L=2\) candidate is supported
by an equal-statistical-budget all-pattern screen followed by the common-gate
encoder-online-gate calculation.  The latter compiles two-column charge-Haar
and Dicke isometries and recomputes the all-pattern recovery, but it does not
compile the recoveries of every comparison design.  Its depth-\(7\) values
are upper bounds from successful synthesis, not lower bounds or proofs of
minimal circuit depth.  The seed audit also precludes treating the frozen
trajectory and charge-Haar samples as independent; all corresponding resource-accounted contrasts resample complete
sample-index clusters jointly.  Both calculations use PRNG initialization values distinct from those of
the frozen \(L=2\) implementation batch used for recovery compilation,
Fig.~\ref{fig:conditional-window}, and the synthetic-sensitivity tables; the
batches are analyzed separately and not pooled.
The simulations and the present numerical
certification evaluate the GJC-CFI directly; they neither estimate nor report
a uniform diamond-norm error for the recovered logical channel.
Consequently, Corollary~\ref{cor:gjc-fi-continuity}
establishes the rigorous conversion needed once such a channel estimate is
available; it is not invoked as an a posteriori certification of the plotted
numerical CFI\@.

At the conditional engineering level, erasure locations are assumed to be
flagged.  Unflagged amplitude damping, dephasing, and false erasure flags
require a different recovery problem.  The sensitivity map uses a synthetic
output-level attenuation and normalized abstract gate-layer time; it omits
encoder and idle noise, calibrated leakage, crosstalk, flag error and latency,
reference preparation and consumption, feed-forward, retrieval, detection,
and reset.  The clean simulator rerun reproduces the
matched \(n=5,L=2\) calculation under frozen synthetic inputs and performs a
data-informed, post hoc full-grid simultaneous reanalysis on the same 30
complete realizations.  It is a fixed-stream computational reproduction, not an independent
statistical replication or a prospectively specified
confirmatory analysis.

A platform-total comparison is deliberately outside the present claim set.
The networked SiV experiment demonstrates signal-storage and nonlocal
heralding capabilities~\cite{Stas2026}, whereas the cold-atom experiment
stores and retrieves the ancillary Fock-state reference used for GJC
interference~\cite{Wang2026}.  These demonstrations provide complementary
architectural context, but they instantiate different memory roles and are
not treated as calibration snapshots for the \(n=5\) encoded signal-memory
architecture studied here.  Appendix~\ref{sec:platform-contract} therefore
states the platform inputs and provenance needed before
\(\mathcal J_\Pi\) can be evaluated; no hardware result is inferred from the
literature comparison.

Approximate random \(U(1)\)-covariant codes and asymptotic erasure-error
scaling in fixed-logical-size regimes are already
known~\cite{Faist2020,Kong2022}.  What remains unresolved is a finite-depth
local logical-channel bound with an efficient decoder and a
platform-complete, reference-accounted comparison.  The present construction
has rate \(1/n\); its negative memory-only result is not promoted to a
fixed-total-resource comparison.

The appendices make this hierarchy auditable by collecting the proofs,
finite-depth questions, full finite-size protocols, decoder and measurement
ablations, all-pattern statistics, gate-set compilation diagnostics, and the
detailed QFI and GJC derivations.

\section{Conclusion}

Quantum-assisted long-baseline interferometry depends on preserving the
nonlocal complex visibility after weak astronomical light has been coherently
captured.  Flagged memory loss is especially damaging because the erased
subsystem can reveal which-node information, while strict finite-dimensional
\(U(1)\) covariance prevents uniform exact correction of every single-memory
erasure.  This work asked whether local number-conserving scrambling, followed
by recovery conditioned only on the heralded erasure pattern, can protect the
visibility that is ultimately accessible to a fixed GJC receiver.

The analytical results establish the relevant hierarchy of guarantees.  Exact
local recovery restores the full two-parameter visibility family and therefore
all subsequent receiver statistics, whereas the covariance obstruction rules
out such uniform exact protection for the finite encoding considered here.
For approximate recovery, the channel-to-Fisher-information stability theorem
converts a uniform logical-channel error bound into a uniform bound on the
classical Fisher information delivered by the fixed receiver over compact
visibility regions.  It therefore links a coding-level statement to the
operational quantity used for astronomical estimation.

The finite-size calculations and circuit construction complete the intended
operational chain within the stated model.  Number-conserving scrambling
redistributes erasure-location risk and suppresses the most damaging leakage
events; task-weighted recovery returns part of the surviving information to a
fixed local readout; and an equal-budget comparison selects a shallow encoder
among the tested finite designs.  Every nontrivial recovery branch in a
separate implementation batch is also resolved into abstract
nearest-neighbor number-conserving gates.  These results demonstrate a
coherent finite-size mechanism rather than an asymptotic coding theorem.

The present low-rate construction does not improve throughput when the total
number of memories is fixed.  This study is not a platform-calibrated
comparison: existing SiV signal-storage and cold-atom ancillary-reference
demonstrations realize complementary components rather than the same encoded
architecture~\cite{Stas2026,Wang2026}.  The main significance is therefore
architectural: spare local memory capacity can be converted into protection
of astronomical coherence without abandoning locality, photon-number
symmetry, or an explicit receiver.  The next step is to develop higher-rate
finite-depth covariant codes with efficient recovery and to validate them in
a fully calibrated end-to-end platform comparison.

\section*{Data and code availability}

The numerical data, source code, executed notebooks, and computational-
environment specifications supporting this study are available in the
associated reproducibility archive.  No unpublished experimental data were
used.

\appendix

The appendices are organized as follows.  Appendices A--D give the covariant-QEC, approximate-recovery, finite-depth,
and phase-covariance arguments.  Appendix E contains the complete numerical
protocols and ablations.  Appendices F--G state the detailed limitations and
open problems, while Appendices H--J collect the QFI, GJC-receiver, and
operator-space derivations used in the main text.

\section{Covariant-QEC context and the architecture-specific open regime}
\label{sec:decoupling}

Volume-law entanglement and large Schmidt rank are useful diagnostics of scrambling~\cite{Page1993,Hayden2007}, but neither is by itself sufficient to establish correctability.  The operational criterion is decoupling of the erased subsystem from a reference that purifies the logical input.

Let $R$ be a reference qubit and prepare $\left\lvert \Phi_{RL}\right\rangle{}=(\left\lvert 00\right\rangle{}+\left\lvert 11\right\rangle{})/\sqrt2$.  After a local encoder $V$ and a bipartition into erased subsystem $E$ and retained subsystem $Q$, define
\begin{equation}
\rho_{REQ}
=
(\mathbb{I}_R\otimes V)
\left\lvert \Phi_{RL}\right\rangle{}\!\left\langle \Phi_{RL}\right\rvert{}
(\mathbb{I}_R\otimes V^\dagger).
\label{eq:encoded-Choi-state}
\end{equation}
The decoupling error is
\begin{equation}
D_E(V)
=
\frac12
\left\lVert \rho_{RE}-\rho_R\otimes\rho_E\right\rVert{}_1.
\label{eq:decoupling-error}
\end{equation}
Exact erasure correction is equivalent to $D_E(V)=0$.  Approximate decoupling implies the existence of an approximate recovery channel through the standard decoupling and information--disturbance framework~\cite{Knill2000,Hayden2007}.

Before imposing covariance, the fixed-pattern Haar benchmark can be made
explicit.

\begin{lemma}[Unrestricted fixed-pattern Haar decoupling]
\label{lem:haar-fixed-pattern}
Let
$V:\mathbb{C}^{d_L}\rightarrow
\mathbb{C}^{d_Q}\otimes\mathbb{C}^{d_E}$
be a Haar-random isometry obtained from the first $d_L$ columns of a Haar
unitary, where $d_L\leq d_Qd_E$, and define $D_E(V)$ by
Eq.~\eqref{eq:decoupling-error} using a maximally entangled state of Schmidt
rank $d_L$.  For a fixed erased subsystem,
\begin{equation}
\mathbb{E}_V D_E(V)
\leq
\frac12
\sqrt{
\frac{
d_Ed_Q(d_E^2-1)(d_L^2-1)
}{
d_L[(d_Ed_Q)^2-1]
}
}.
\label{eq:haar-fixed-pattern-bound}
\end{equation}
For one logical qubit encoded without symmetry constraints into $n$ physical
qubits, with a fixed set of $k$ erased qubits,
\begin{equation}
\mathbb{E}_V D_E(V)
\leq
\sqrt{\frac38}\,
2^{-(n-2k)/2}
\sqrt{
\frac{1-2^{-2k}}{1-2^{-2n}}
}.
\label{eq:haar-half-scaling}
\end{equation}
If $k/n\rightarrow\alpha\in(0,1/2)$, the last factor tends to one and
the bound decays as $2^{-(1-2\alpha)n/2}$.
\end{lemma}

\begin{proof}
Put $X=\rho_{RE}-\rho_R\otimes\rho_E$.  Since
$\rho_R=\mathbb{I}_R/d_L$,
\begin{equation}
\operatorname{Tr}(X^2)
=
\operatorname{Tr}(\rho_{RE}^2)
-
\frac{1}{d_L}\operatorname{Tr}(\rho_E^2).
\label{eq:haar-purity-difference}
\end{equation}
The two-copy Haar-isometry identity gives
\begin{align}
\mathbb{E}_V\operatorname{Tr}(\rho_{RE}^2)
&=
\frac{
d_E(d_Q^2-1)+d_Ld_Q(d_E^2-1)
}{
d_L[(d_Ed_Q)^2-1]
},
\label{eq:haar-re-purity}\\
\mathbb{E}_V\operatorname{Tr}(\rho_E^2)
&=
\frac{
d_Ld_E(d_Q^2-1)+d_Q(d_E^2-1)
}{
d_L[(d_Ed_Q)^2-1]
}.
\label{eq:haar-e-purity}
\end{align}
Substitution into Eq.~\eqref{eq:haar-purity-difference} yields
\begin{equation}
\mathbb{E}_V\operatorname{Tr}(X^2)
=
\frac{
d_Q(d_E^2-1)(d_L^2-1)
}{
d_L^2[(d_Ed_Q)^2-1]
}.
\label{eq:haar-hilbert-schmidt}
\end{equation}
Because $X$ acts on a space of dimension $d_Ld_E$,
$\lVert X\rVert_1\leq\sqrt{d_Ld_E}\lVert X\rVert_2$.
Jensen's inequality then proves
Eq.~\eqref{eq:haar-fixed-pattern-bound}.  Setting
$d_L=2$, $d_E=2^k$, and $d_Q=2^{n-k}$ gives
Eq.~\eqref{eq:haar-half-scaling}.
\end{proof}

Lemma~\ref{lem:haar-fixed-pattern} establishes exponential fixed-pattern
decoupling below one half and is consistent with the observed crossover of
the unrestricted Haar benchmark.  It does
\emph{not} establish that one sampled encoder corrects all
$\binom{n}{k}$ sets simultaneously:
\begin{equation}
\sup_E\mathbb{E}_V D_E(V)
\leq
\mathbb{E}_V\max_E D_E(V)
\label{eq:average-versus-uniform-erasure}
\end{equation}
Equality need not hold, so a fixed-pattern expectation cannot be substituted
for a high-probability worst-pattern guarantee.
Nor can the unrestricted result be transferred unchanged to the astronomical
logical qubit under a $U(1)$ constraint.  The
logical states $\left\lvert0\right\rangle{}$ and
$\left\lvert1\right\rangle{}$ transform with different charges, so a covariant
encoder must preserve their relative group action.  An isometry whose entire
range lies in one fixed-total-charge sector instead represents a
charge-neutral logical subsystem and does not by itself establish protection
of the vacuum--one-excitation coherence.  We therefore do not claim an
asymptotic capacity theorem from fixed-sector dimension counting.

\paragraph*{Known covariant-code results and the remaining question.}
Continuous-symmetry versions of the Eastin--Knill obstruction already give
quantitative lower bounds on the erasure-correction accuracy of
finite-dimensional covariant codes~\cite{Faist2020}.  The charge-Haar
construction relevant to the present numerical model is also not new in
principle: random $U(1)$-covariant codes generated by Haar-random
number-conserving unitaries have been analyzed in
Ref.~\cite{Kong2022}.  In the fixed-logical-size, fixed-erasure-size regimes
treated there, both average- and worst-case purified-distance errors vanish as
$O(n^{-1})$.  Asymptotically vanishing approximate error at fixed logical
size is therefore a known covariant-code result, not a conjecture of this
work.

The present family encodes one logical occupancy qubit into $n$ physical
memories and hence has rate $1/n$; we make no nonzero-rate claim.  The
architecture-specific open question is whether a specified geometrically
local, finite-depth, number-conserving circuit and an explicit local covariant
decoder can obtain comparable task-specific recovery guarantees under
independent flagged loss, uniformly over a stated visibility region, with the
logical-channel error and circuit and reference-frame costs included.
Corollary~\ref{cor:gjc-fi-continuity} then propagates any resulting uniform
diamond-error estimate to the operational GJC CFI\@.  Any higher-rate extension
must define the number and charge range of the logical degrees of freedom and
respect the known covariance--accuracy tradeoffs, rather than posit vanishing
full-state recovery error at a fixed positive rate.

Independently of random coding, an unconstrained exact $[[n,1,d]]$ quantum
code corrects $t=d-1$ flagged erasures.  The quantum Singleton bound gives
\begin{equation}
n-1\geq2(d-1)=2t,
\label{eq:singleton}
\end{equation}
and hence
\begin{equation}
t\leq\left\lfloor\frac{n-1}{2}\right\rfloor.
\label{eq:exact-half-limit}
\end{equation}
The near-one-half value is therefore a no-cloning-compatible upper limit for
exact local erasure correction of one logical qubit~\cite{Knill2000}.  It is
an unrestricted, noncovariant exact-code benchmark: it is neither an
achievability statement for the strictly $U(1)$-covariant logical-occupancy
code, which is subject to Lemma~\ref{lem:covariant-no-go}, nor a threshold for
approximate or task-specific covariant recovery.

\section{Approximate recovery and visibility transfer}
\label{sec:approximate}

For a concrete local encoder, erasure pattern, and decoder at node $j$, define
the effective logical channel
\begin{equation}
\Lambda_{j,E_j}
=
\mathcal{R}_j^{E_j}\circ\operatorname{Tr}_{E_j}\circ\mathcal{V}_j,
\qquad
\mathcal{V}_j(X)=V_jXV_j^\dagger.
\label{eq:effective-logical-channel}
\end{equation}
The most relevant process element for interferometry is
\begin{equation}
\Lambda_{j,E_j}(\left\lvert 0\right\rangle{}\!\left\langle 1\right\rvert{})
=
\lambda_{j,E_j}\left\lvert 0\right\rangle{}\!\left\langle 1\right\rvert{}
+\Xi_{j,E_j},
\label{eq:coherence-transfer}
\end{equation}
where \(\lambda_{j,E_j}\) is the Hilbert--Schmidt projection onto
\(\lvert0\rangle\!\langle1\rvert\), so that
\(\operatorname{Tr}[\lvert0\rangle\!\langle1\rvert^\dagger
\Xi_{j,E_j}]=0\); the remainder contains leakage, bias, or
quadrature-mixing terms.  If
the effective channel is phase covariant, trace preserving on the accepted
subensemble, and unbiased, then $\Xi_{j,E_j}$ has no component that adds an
unknown offset to the logical coherence.  Writing
$\lambda_j=\lambda_{j,E_j}$ for the selected local pattern,
\begin{equation}
g_{\mathrm{out}}\simeq\lambda_A^*\lambda_B g.
\label{eq:g-transfer}
\end{equation}
Here \(\simeq\) denotes the scalar-transfer approximation obtained by
discarding the explicitly separated remainders \(\Xi_{A,E_A}\) and
\(\Xi_{B,E_B}\), not an additional weak-source approximation.
The recovered GJC CFI is then
\begin{equation}
F_{\delta,E}^{\mathrm{GJC}}
=
\frac{\epsilon\eta_c\eta_a}{2}
\frac{
\nu_a^2\left\lvert \lambda_A\lambda_B\right\rvert{}^2
\boldsymbol{u}_{\delta'}\boldsymbol{u}_{\delta'}^{\mathsf T}
}{
1-\nu_a^2\left\lvert \lambda_A\lambda_B\right\rvert{}^2
\bigl(\boldsymbol{u}_{\delta'}^{\mathsf T}\boldsymbol{g}\bigr)^2
},
\label{eq:approx-gjc-fi}
\end{equation}
For \(\lambda_A\lambda_B\neq0\), the effective receiver phase is explicitly
\begin{equation}
\delta'
=
\delta-\arg(\lambda_A^*\lambda_B)
=
\delta+\arg(\lambda_A\lambda_B^*).
\label{eq:effective-receiver-phase}
\end{equation}
Because \(E\) is flagged and each \(\lambda_{j,E_j}\) is fixed by the
encoder--recovery branch, this correction is deterministic, pattern
conditioned, and independent of \(g\).  Thus the known phase of
\(\lambda_A\lambda_B^*\) is absorbed into \(\delta'\) with the same
conjugation convention as Eq.~\eqref{eq:g-transfer}; it is not an adaptive
phase estimate extracted from the astronomical data.  Any flag-processing,
feed-forward, or phase-setting latency remains outside the rate envelope.
If either transfer amplitude vanishes, the scalar visibility-transfer term
and its CFI vanish, so \(\delta'\) is immaterial.
Equation~\eqref{eq:approx-gjc-fi} is a useful special case, but if
$\Xi_E\neq0$ the full probability model must be differentiated directly; a
scalar ``QFI retention factor'' is generally insufficient.
Theorem~\ref{thm:fi-continuity} controls this general case.

\subsection{Proof of the operational stability theorem}

\begin{proof}[Proof of Theorem~\ref{thm:fi-continuity}]
Write $\Delta=\Lambda-\Gamma$.  The map $\Delta$ is
Hermiticity preserving and trace annihilating.  If $X$ is Hermitian and
traceless and $0\leq M\leq\mathbb{I}$, the Jordan decomposition
$X=X_+-X_-$ gives
\begin{equation}
\left\lvert\operatorname{Tr}(MX)\right\rvert
\leq
\frac12\left\lVert X\right\rVert{}_1.
\label{eq:effect-traceless-bound}
\end{equation}
Consequently,
\begin{equation}
\left\lvert p_y-q_y\right\rvert
\leq
\frac12
\left\lVert\Delta(\rho)\right\rVert{}_1
\leq
\varepsilon_\diamond.
\label{eq:fi-probability-difference}
\end{equation}
Because both channels are independent of $\boldsymbol{\theta}$,
differentiation commutes with them.  Since
$\partial_\mu\rho$ is also Hermitian and traceless,
\begin{equation}
\left\lvert
\partial_\mu p_y-\partial_\mu q_y
\right\rvert
\leq
\varepsilon_\diamond
\left\lVert\partial_\mu\rho\right\rVert{}_1.
\label{eq:fi-tangent-difference}
\end{equation}
Define
$\boldsymbol{a}_y=\boldsymbol{\nabla}p_y$ and
$\boldsymbol{b}_y=\boldsymbol{\nabla}q_y$.
Equation~\eqref{eq:fi-tangent-difference} implies
\begin{equation}
\left\lVert\boldsymbol{a}_y-\boldsymbol{b}_y\right\rVert{}_2
\leq
\varepsilon_\diamond S_\rho.
\label{eq:fi-gradient-difference}
\end{equation}
Trace-norm contraction of a CPTP channel on Hermitian operators, together
with Eq.~\eqref{eq:effect-traceless-bound}, also gives
\begin{equation}
\left\lVert\boldsymbol{a}_y\right\rVert{}_2,
\left\lVert\boldsymbol{b}_y\right\rVert{}_2
\leq\frac{S_\rho}{2}.
\label{eq:fi-gradient-size}
\end{equation}

First suppose $\varepsilon_\diamond\leq q_0/2$.  Then
Eq.~\eqref{eq:fi-probability-difference} gives $p_y\geq q_0/2$.
Using
\begin{equation}
F_\Lambda-F_\Gamma
=
\sum_y
\left(
\frac{\boldsymbol{a}_y\boldsymbol{a}_y^{\mathsf T}}{p_y}
-
\frac{\boldsymbol{b}_y\boldsymbol{b}_y^{\mathsf T}}{q_y}
\right)
\label{eq:fi-rank-one-sum}
\end{equation}
and
\begin{equation}
\left\lVert
\boldsymbol{a}\boldsymbol{a}^{\mathsf T}
-
\boldsymbol{b}\boldsymbol{b}^{\mathsf T}
\right\rVert{}_{\mathrm{op}}
\leq
\left(
\left\lVert\boldsymbol{a}\right\rVert{}_2
+
\left\lVert\boldsymbol{b}\right\rVert{}_2
\right)
\left\lVert\boldsymbol{a}-\boldsymbol{b}\right\rVert{}_2,
\label{eq:rank-one-difference}
\end{equation}
the contribution of one outcome obeys
\begin{align}
&
\left\lVert
\frac{\boldsymbol{a}_y\boldsymbol{a}_y^{\mathsf T}}{p_y}
-
\frac{\boldsymbol{b}_y\boldsymbol{b}_y^{\mathsf T}}{q_y}
\right\rVert{}_{\mathrm{op}}
\nonumber\\
&\quad\leq
\frac{
\left(
\left\lVert\boldsymbol{a}_y\right\rVert{}_2
+
\left\lVert\boldsymbol{b}_y\right\rVert{}_2
\right)
\left\lVert\boldsymbol{a}_y-\boldsymbol{b}_y\right\rVert{}_2
}{q_0/2}
\nonumber\\
&\qquad+
\left\lVert\boldsymbol{b}_y\right\rVert{}_2^2
\frac{\left\lvert p_y-q_y\right\rvert}{(q_0/2)^2}
\nonumber\\
&\quad\leq
S_\rho^2
\left(
\frac{2}{q_0}+\frac{1}{q_0^2}
\right)\varepsilon_\diamond.
\label{eq:fi-one-outcome-bound}
\end{align}
Summing over the $K$ outcomes gives the first constant in
Eq.~\eqref{eq:fi-general-constant}.

For $\varepsilon_\diamond>q_0/2$, the CFI of a fixed measurement after a
parameter-independent channel is bounded in positive-semidefinite order by
the input SLD QFI:
\begin{equation}
0\preceq F_\Lambda,F_\Gamma\preceq H[\rho].
\label{eq:cfi-qfi-order}
\end{equation}
It follows that
$-\!H[\rho]\preceq F_\Lambda-F_\Gamma\preceq H[\rho]$ and hence
\begin{equation}
\left\lVert F_\Lambda-F_\Gamma\right\rVert{}_{\mathrm{op}}
\leq H_{\max}
<
\frac{2H_{\max}}{q_0}\varepsilon_\diamond.
\label{eq:fi-large-error-bound}
\end{equation}
The two error regimes together prove
Eqs.~\eqref{eq:fi-general-bound} and~\eqref{eq:fi-general-constant}.
\end{proof}

\paragraph*{Evaluation for the GJC parity receiver.}
For Eq.~\eqref{eq:bloch-state},
\begin{equation}
\begin{aligned}
\partial_{g_R}\rho_s^{(1)}&=\frac{\sigma_x}{2},
&
\partial_{g_I}\rho_s^{(1)}&=-\frac{\sigma_y}{2},
\\
\left\lVert\sigma_x\right\rVert{}_1
&=
\left\lVert\sigma_y\right\rVert{}_1
=2,
&
S_\rho^2&=2.
\end{aligned}
\label{eq:gjc-derivative-norms}
\end{equation}
The two parity probabilities satisfy
\begin{equation}
q_z(g)
\geq
\frac{1-\nu_a g_{\max}}{2}
\geq
\frac{1-g_{\max}}{2}.
\label{eq:gjc-probability-floor}
\end{equation}
Moreover, the largest eigenvalue of Eq.~\eqref{eq:qfi-single} is
\begin{equation}
\left\lVert H^{(1)}(g)\right\rVert{}_{\mathrm{op}}
=
\frac{1}{1-\lvert g\rvert^2}
\leq
\frac{1}{1-g_{\max}^2}.
\label{eq:gjc-qfi-ceiling}
\end{equation}
Substitution of $K=2$ and
$q_0=(1-g_{\max})/2$ into
Eq.~\eqref{eq:fi-general-constant} yields
Eq.~\eqref{eq:gjc-uniform-constant}; this entire small-error entry in the
maximum also dominates the large-error entry.  This proves
Corollary~\ref{cor:gjc-fi-continuity}.

\paragraph*{Error metric and postselection.}
The linear statement applies directly to the half-diamond distance.  It also
applies to the worst-case, reference-assisted channel purified distance
\begin{equation}
P_\diamond(\Lambda,\Gamma)
=
\sup_{\rho_{RA}}
P\!\left[
(\mathbb{I}_R\otimes\Lambda)(\rho_{RA}),
(\mathbb{I}_R\otimes\Gamma)(\rho_{RA})
\right],
\label{eq:channel-purified-distance}
\end{equation}
where it suffices to take
\(\dim\mathcal{H}_R=\dim\mathcal{H}_A\), because
$\frac12\lVert\Lambda-\Gamma\rVert_\diamond\leq
P_\diamond(\Lambda,\Gamma)$.
For purified distance between normalized Choi states, a conversion to diamond
distance introduces an input-dimension factor.  An entanglement infidelity
\(1-F_e\), with \(F_e\) evaluated for a specified input--reference state
against the ideal channel, generally gives only a square-root bound and cannot replace
$\varepsilon_\diamond$ linearly.

To see the last point, consider
\begin{equation}
\rho_t
=
\frac12
\left[
\mathbb{I}
+
\frac{t}{\sqrt2}(\sigma_x+\sigma_y)
\right]
\label{eq:infidelity-counterexample-state}
\end{equation}
with a fixed $\sigma_x$ measurement, and let the error channel be the
$z$ rotation
$U_\vartheta=e^{-i\vartheta\sigma_z/2}$.
At $t=0$ the ideal and rotated output states are both
$\mathbb{I}/2$, but their Fisher informations are
\begin{equation}
F_0=\frac12,
\qquad
F_\vartheta
=
\frac{(\cos\vartheta-\sin\vartheta)^2}{2}
=
\frac12-\vartheta+O(\vartheta^2).
\label{eq:infidelity-counterexample-fi}
\end{equation}
By contrast,
$1-F_e=\sin^2(\vartheta/2)=O(\vartheta^2)$.
Thus neither single-work-point state distance nor entanglement infidelity
supports the claimed linear CFI bound without additional derivative control.

Finally, a normalized accepted map is generally nonlinear when its success
probability depends on the input.  The conditional GJC corollary therefore
requires a common $g$-independent success weight.  Otherwise the fixed
finite-outcome theorem should be applied to the complete instrument containing
success, failure, rejection, and erasure-flag outcomes.

\section{Finite-depth local scramblers}
\label{sec:finite-depth}

Haar-random encoders are analytical idealizations.  A physically structured
theory must specify a local circuit ensemble.  The unrestricted benchmark may
use generic local gates, whereas the architecture-specific encoder must use
number-conserving gates compatible with Eq.~\eqref{eq:covariant-encoding}.  In
either case write
\begin{equation}
V_j(L)=U_j^{(L)}U_j^{(L-1)}\cdots U_j^{(1)}V_j^{(0)},
\label{eq:brickwork-encoder}
\end{equation}
or a chaotic local Hamiltonian evolution
\begin{equation}
V_j(t)=e^{-itH_j}V_j^{(0)}.
\label{eq:chaotic-encoder}
\end{equation}

\begin{lemma}[Causal-depth obstruction to uniform exact erasure recovery]
\label{lem:causal-depth}
Suppose the logical qubit in $V_j^{(0)}$ is localized at one site of a
one-dimensional chain, all other inputs are fixed ancillas, and
$U_j^{(1)},\ldots,U_j^{(L)}$ form a depth-$L$ nearest-neighbor circuit.  If
one encoder is required to exactly correct every flagged erasure set of size
$k$,
then necessarily
\begin{equation}
2L+1>k.
\label{eq:causal-depth-bound}
\end{equation}
Thus protection against an extensive number $k=\alpha n$ of arbitrary
erasures requires $L=\Omega(n)$ in one dimension.  In a $d$-dimensional
bounded-degree lattice with bounded-range gates and causal-ball volume at
most $CL^d$, the corresponding causal-volume argument gives the necessary
scaling $L=\Omega(k^{1/d})$ for uniform exact recovery.
\end{lemma}

\begin{proof}
After depth $L$, the forward causal cone of the localized logical input
contains at most $2L+1$ sites.  If $k\geq2L+1$, choose an erasure set that
contains this entire cone and add arbitrary sites if necessary.  The retained
state is then independent of the logical input, so no recovery channel on the
retained sites can reconstruct an unknown logical qubit.  The
$d$-dimensional statement follows by replacing the interval by a causal
region of volume $O(L^d)$.
\end{proof}

\paragraph*{Open finite-depth questions.}
Lemma~\ref{lem:causal-depth} is a necessary condition, not an achievability
result.  For a specified number-conserving brickwork or Hamiltonian ensemble,
two inequivalent targets are the pattern-averaged statement
\begin{equation}
\mathbb{E}_{V\sim\mathcal{E}_{n,L}}
\left[
\sum_{\lvert E\rvert=k}\pi(E)D_E(V)
\right]
\leq\varepsilon
\label{eq:finite-depth-average-question}
\end{equation}
and the stronger high-probability worst-pattern statement
\begin{equation}
\Pr_{V\sim\mathcal{E}_{n,L}}
\left[
\max_{\lvert E\rvert=k}D_E(V)>\varepsilon
\right]
\leq\delta.
\label{eq:finite-depth-uniform-question}
\end{equation}
Here $\mathcal{E}_{n,L}$ is the specified size- and depth-dependent circuit
ensemble, while $\pi(E)\geq0$ and
$\sum_{\lvert E\rvert=k}\pi(E)=1$ define the physical pattern distribution.
Equations~\eqref{eq:finite-depth-average-question}
and~\eqref{eq:finite-depth-uniform-question} are written for circuits.  A
Hamiltonian analogue must specify finite-range, bounded-strength interactions
and retain the appropriate Lieb--Robinson tail rather than assume a strictly
compact causal cone.
Determining the depth or evolution time needed for either bound, including
its dependence on $n$, $k$, charge sector, geometry, and error probabilities,
remains open.  The standard information--disturbance framework converts a
complementary-channel decoupling estimate into existence of a recovery
channel~\cite{Knill2000,Hayden2007}; for fixed logical dimension, a
Choi-state bound $D_E$ leads generically to a logical diamond error of order
$O(\sqrt{d_LD_E})$.  The architecture-specific missing step is therefore the
finite-depth covariant bound itself, with explicit constants and a compatible
decoder.  Once its logical diamond error is known,
Corollary~\ref{cor:gjc-fi-continuity} supplies the corresponding
fixed-receiver CFI estimate.

Large bipartite entropy, large operator support, or approximate unitary-design
diagnostics may support this result, but they do not replace the decoupling
bound required for recovery.  A symmetry-preserving unitary 2-design acting
within the fixed-charge sectors is a compatible intermediate target:
second-moment design properties may suffice for standard decoupling estimates
while retaining a tractable description.  An unrestricted Clifford ensemble
would not respect the required \(U(1)\) covariance.

At the three sampled sizes, the choice $L=3n$ is numerically close to the
corresponding Haar-isometry retention curve at the sampled points.  This observation
does not establish linear-depth scaling and is not a substitute for the
decoupling statements in
Eqs.~\eqref{eq:finite-depth-average-question}
and~\eqref{eq:finite-depth-uniform-question}: only
$n\leq7$ is tested, and the numerical observable is QFI retention rather than
the trace-norm decoupling error $D_E$ itself.

\section{Phase covariance and superselection constraints}
\label{sec:ssr}

The cross-aperture coherence $\left\lvert 10\right\rangle{}\!\left\langle 01\right\rvert{}$ is defined relative to a phase convention.  Under independent local phase shifts,
\begin{equation}
U_A(\alpha_A)\otimes U_B(\alpha_B)
=e^{-i\alpha_A\hat n_A}\otimes e^{-i\alpha_B\hat n_B},
\end{equation}
the visibility transforms as
\begin{equation}
g\longmapsto ge^{-i(\alpha_A-\alpha_B)}.
\label{eq:g-phase-transform}
\end{equation}

\begin{lemma}[Loss of phase information under unreferenced local twirling]
\label{lem:twirl}
If the relative local phase is completely unknown and no asymmetry resource is supplied, local $U(1)$ twirling removes the astronomical coherence:
Define $U_A(\alpha)=e^{-i\alpha\hat n_A}\otimes\mathbb{I}_B$.  Then
\begin{equation}
\begin{aligned}
&\int_0^{2\pi}\frac{d\alpha}{2\pi}
U_A(\alpha)\rho_s^{(1)}(g)U_A(\alpha)^\dagger
\\
&\qquad=
\frac12\left(\left\lvert 10\right\rangle{}\!\left\langle 10\right\rvert{}+\left\lvert 01\right\rangle{}\!\left\langle 01\right\rvert{}\right).
\end{aligned}
\label{eq:twirled-state}
\end{equation}
\end{lemma}

\begin{proof}
The diagonal terms are invariant.  The off-diagonal operator transforms as
\begin{equation}
\left\lvert 10\right\rangle{}\!\left\langle 01\right\rvert{}\longmapsto e^{-i\alpha}\left\lvert 10\right\rangle{}\!\left\langle 01\right\rvert{}.
\end{equation}
Its uniform phase average vanishes, and similarly for its adjoint.
\end{proof}

\begin{lemma}[One shared reference photon under local twirling]
\label{lem:one-reference-qfi}
Let the astronomical single-photon state be
$\rho_s^{(1)}(g)$ from Eq.~\eqref{eq:bloch-state}, whose two local
intensities are equal, and supply the ideal balanced path-entangled
single-photon reference
\begin{equation}
\left\lvert\psi_\delta\right\rangle_a
=
\frac{
\left\lvert10\right\rangle_a
+e^{i\delta}\left\lvert01\right\rangle_a
}{\sqrt2}.
\label{eq:one-photon-reference-state}
\end{equation}
where $\delta$ is known, fixed, and independent of $g$.  Apply independent
local $U(1)$ twirling to the joint signal--reference state, with each local
action generated by the total signal-plus-reference photon number at that
node.  Then the SLD QFI matrix of the twirled joint state, which is the
corresponding SSR-restricted QFI benchmark, is
\begin{equation}
H_{\mathrm{ref}}^{\mathrm{SSR}}(g)
=
\frac12 H^{(1)}(g).
\label{eq:one-reference-qfi}
\end{equation}
For the complete weak-source state of Eq.~\eqref{eq:source-state}, the
corresponding leading-order result is
\begin{equation}
H_{\mathrm{ref}}^{\mathrm{SSR}}[\rho_s(g)]
=
\frac{\epsilon}{2}H^{(1)}(g)+O(\epsilon^2).
\label{eq:one-reference-weak-qfi}
\end{equation}
These statements assume a lossless, undephased reference photon, a fresh
reference for each measurement, and no additional reference uncertainty.
They do not include the cost of preparing, distributing, stabilizing, or
consuming that reference.
\end{lemma}

\begin{proof}
For $\boldsymbol{\alpha}=(\alpha_A,\alpha_B)$, let
\begin{align}
U_{\boldsymbol{\alpha}}
&=
e^{-i(\alpha_A \hat N_A+\alpha_B \hat N_B)},
\nonumber\\
\mathcal{G}_{AB}(X)
&=
\int_{[0,2\pi)^2}
\frac{d^2\boldsymbol{\alpha}}{(2\pi)^2}
U_{\boldsymbol{\alpha}}XU_{\boldsymbol{\alpha}}^\dagger
\label{eq:local-joint-twirl}
\end{align}
denote the independent local twirl, where $\hat N_A$ and $\hat N_B$ are the respective
total signal-plus-reference photon numbers.  The joint state
has two photons in total and decomposes into orthogonal sectors labelled by
the local total photon numbers $(N_A,N_B)=(2,0),(1,1),(0,2)$.  Direct
expansion gives
\begin{equation}
\mathcal{G}_{AB}\!\left[
\rho_s^{(1)}(g)\otimes
\left\lvert\psi_\delta\right\rangle{}\!
\left\langle\psi_\delta\right\rvert_a
\right]
=
\frac14\rho_{20}
\oplus
\frac12\rho_{11}(g)
\oplus
\frac14\rho_{02},
\label{eq:reference-twirled-blocks}
\end{equation}
where $\rho_{20}$ and $\rho_{02}$ are pure, independent of $g$, and the
normalized informative block is
\begin{equation}
\begin{aligned}
\rho_{11}(g)
=\frac12\bigl(
&\left\lvert X\right\rangle{}\!\left\langle X\right\rvert{}
+\left\lvert Y\right\rangle{}\!\left\langle Y\right\rvert{}
\\
&+ge^{i\delta}
\left\lvert X\right\rangle{}\!\left\langle Y\right\rvert{}
+g^*e^{-i\delta}
\left\lvert Y\right\rangle{}\!\left\langle X\right\rvert{}
\bigr),
\end{aligned}
\label{eq:reference-informative-block}
\end{equation}
with
\begin{align}
\left\lvert X\right\rangle
&=
\left\lvert1_s0_a\right\rangle_A
\left\lvert0_s1_a\right\rangle_B,
&
\left\lvert Y\right\rangle
&=
\left\lvert0_s1_a\right\rangle_A
\left\lvert1_s0_a\right\rangle_B.
\end{align}
The $g$-independent unitary
$\left\lvert10\right\rangle\mapsto\left\lvert X\right\rangle$ and
$\left\lvert01\right\rangle\mapsto
e^{-i\delta}\left\lvert Y\right\rangle$
makes $\rho_{11}(g)$ unitarily equivalent to $\rho_s^{(1)}(g)$.
For a block-diagonal state whose block weights are parameter independent,
the QFI is the weighted sum of the block QFIs.  Only the $(1,1)$ block
depends on $g$, and its weight is $1/2$, proving
Eq.~\eqref{eq:one-reference-qfi}.  The vacuum sector of
Eq.~\eqref{eq:source-state} is independent of $g$, while the single-photon
sector has parameter-independent weight $\epsilon$, which proves
Eq.~\eqref{eq:one-reference-weak-qfi} to the retained order.
\end{proof}

\paragraph*{Common resource convention for the SSR and GJC benchmarks.}
The ideal reference in Lemma~\ref{lem:one-reference-qfi} is exactly the ideal
limit of the ancillary state in Eqs.~\eqref{eq:ancilla-total}--%
\eqref{eq:ideal-reference-state}; it is not an additional free reference.
Accordingly, every ideal-reference numerical normalization uses one fresh
balanced reference photon per attempted GJC measurement, with
\(\eta_a=\nu_a=1\).  The SSR quantity is the measurement-independent QFI
remaining after the joint local twirl with that resource, whereas the GJC
quantity is the CFI extracted by the specified local interference and
detection measurement.  The preparation, distribution, stabilization, and
consumption costs are omitted from both conditional normalizations.  A
mixed, lossy reference is included in the GJC formula through
\((\eta_a,q_A,q_B,\nu_a)\), but an analogous nonideal-reference SSR-QFI was
not evaluated numerically.  In particular, \(\eta_a\) is a probability
factor rather than a reference-capacity or wall-clock ledger; failed
generation attempts and bottlenecked reference lanes do not cancel from a
platform comparison.  These quantities are required by
Appendix~\ref{sec:platform-contract}.

The GJC ancillary resource acts as a nonlocal phase-reference resource subject to the relevant photon-number superselection rule~\cite{Zhang2025}.  The encoder and decoder must therefore obey one of two consistent models.

\paragraph*{Covariant encoding.}
There exist additive physical generators $\hat N_j$ and fixed background
charges $Q_{0,j}$ such that
\begin{equation}
e^{-i\alpha\hat N_j}V_j
=e^{-i\alpha Q_{0,j}}V_j e^{-i\alpha\hat n_{L_j}},
\label{eq:covariant-encoding}
\end{equation}
and the recovery satisfies the matching covariance relation stated explicitly
in Eq.~\eqref{eq:recovery-covariance}.  The background
phase is essential for the numerical construction, whose logical codewords
occupy the physical charge sectors $Q_{0,j}$ and $Q_{0,j}+1$.  In this case,
an unknown phase is not converted into an uncontrolled logical error by the
coding procedure.

\paragraph*{Explicit reference-frame consumption.}
If generic random gates do not satisfy Eq.~\eqref{eq:covariant-encoding}, the required local oscillators, phase-locked controls, or additional reference states must be included in the resource model.  Their relative phase uncertainty contributes to the effective coherence-transfer factors in Eq.~\eqref{eq:g-transfer}.

The number-conserving calculation below implements one finite-size version of
this requirement by placing the two logical codewords in neighboring charge
sectors and restricting both the encoder and recovery to covariant operations.
It verifies covariance and reference-assisted information retention for
$n=5$.  By Lemma~\ref{lem:covariant-no-go}, its correction of arbitrary local
erasures must be approximate.  Its charge-Haar benchmark is a finite-$n$
instance of the known random $U(1)$-covariant-code framework~\cite{Kong2022}; the architecture-specific additions are the local
finite-depth ensemble, the explicit recovery, and the operational GJC-CFI
evaluation.  The calculation does not establish a nonzero asymptotic coding
rate or a uniform half-erasure threshold.

\section{Numerical methods and extended finite-size results}
\label{sec:numerics}

We now document the finite-system protocols underlying the main-text summary.
The numerical record separates the baseline calculation, equal-budget
candidate comparison, common-gate-set resource analysis, and all-pattern
model reanalysis.  These calculations distinguish numerically demonstrated
effects from the
asymptotic statements that remain open.  Unless stated otherwise, ratios are
conditioned on successful astronomical capture and availability of the
ancillary reference, so the modeled factors \(\epsilon,\eta_c,\eta_a\) can be
restored multiplicatively.  Appendix~\ref{sec:all-pattern-rate} retains every
ideal flagged-erasure event, and Appendix~\ref{sec:native-hardware} adds only
a partial gate-depth timing envelope.  The memory-only \(1/n\) throughput
bound is reported separately rather than absorbed into a conditional
efficiency.  None of these quantities evaluates the platform endpoint
\(\mathcal J_\Pi\) in Eq.~\eqref{eq:platform-information-rate}.

\begin{table*}[t]
\caption{Denominators used in the manuscript.  ``Included'' lists only the
factors represented in the corresponding formula; it is not a claim of
platform completeness.}
\label{tab:denominator-map}
\begin{ruledtabular}
\begin{tabular}{p{0.12\textwidth}p{0.18\textwidth}p{0.28\textwidth}p{0.32\textwidth}}
quantity & denominator & included & excluded from the quantity\\
\hline
\(F_\delta^{\rm GJC}\) & incident temporal mode &
\(\epsilon,\eta_c,\eta_a,\nu_a\) in the stated model &
unmodeled reference/detector failures; platform timing and capacity\\
\(R_C\) & ideal post-capture GJC CFI &
all ideal flagged-erasure patterns &
incident failures and total resources\\
\(G\) & conditional mode with partial timing &
\(R_C\), abstract gate depth, \(e\) &
incident failures; full critical path, calibrated noise, and capacity\\
\(\mathcal J_\Pi\) & wall-clock second at fixed \(\mathcal C\) &
complete classical record &
none by definition; endpoint not evaluated here\\
\end{tabular}
\end{ruledtabular}
\end{table*}

\subsection{Work point and worst-direction metric}

The nominal visibility is
\begin{equation}
g_0=0.3+0.4i.
\label{eq:numerical-g0}
\end{equation}
For two positive-definite information matrices $A$ and $B$, we use the generalized worst-direction ratio
\begin{equation}
\mathcal{R}_{\min}(A\mid B)
=
\lambda_{\min}\!\left(B^{-1/2}AB^{-1/2}\right).
\label{eq:numerical-rmin}
\end{equation}
Equation~\eqref{eq:numerical-rmin} equals the smallest retained fraction over all local directions in the $(g_R,g_I)$ plane.  It prevents a gain in one quadrature from hiding the loss of the other.  Each ratio is evaluated for an individual encoder and erasure pattern before ensemble averaging.

For the recovery comparison, let \(H_{\rm pre}\) be the SSR-restricted QFI
after erasure and before recovery, \(H_{\rm dec}\) the SSR-restricted QFI
after the recovered two-node state, and \(F_{\rm GJC}\) the CFI of the fixed
GJC receiver applied to that state.  The three reported samplewise
efficiencies are
\begin{align}
\eta_{\rm dec}
&=\mathcal R_{\min}(H_{\rm dec}\mid H_{\rm pre}),
\nonumber\\
\eta_{\rm meas}
&=\mathcal R_{\min}(F_{\rm GJC}\mid H_{\rm dec}),
\nonumber\\
\eta_{\rm ext}
&=\mathcal R_{\min}(F_{\rm GJC}\mid H_{\rm pre}).
\label{eq:three-efficiencies}
\end{align}
Because generalized worst-direction ratios need not multiply,
\(\eta_{\rm ext}\) is evaluated directly rather than defined as
\(\eta_{\rm dec}\eta_{\rm meas}\).  It measures post-erasure information
extraction from the pre-recovery state through decoder and receiver; it is
not a system-level end-to-end efficiency.  A legacy archive field for this quantity is retained only for backward
compatibility and is not used as a system-level efficiency.

At Eq.~\eqref{eq:numerical-g0}, the analytic and numerical single-photon QFI matrices agree to machine precision and are
\begin{equation}
H_0^{(1)}=
\begin{pmatrix}
1.12 & 0.16\\
0.16 & 1.213333
\end{pmatrix}.
\label{eq:numerical-qfi0}
\end{equation}
In the ideal reference-assisted SSR model of
Lemma~\ref{lem:one-reference-qfi}, one shared entangled reference photon gives
\begin{equation}
H_{0,\mathrm{ref}}^{\mathrm{SSR}}
=\frac12H_0^{(1)}
=
\begin{pmatrix}
0.56 & 0.08\\
0.08 & 0.606667
\end{pmatrix},
\label{eq:numerical-ssr-qfi0}
\end{equation}
whereas complete independent local $U(1)$ twirling without any asymmetry
resource gives the zero matrix, in agreement with
Lemma~\ref{lem:twirl}.  Equations~\eqref{eq:numerical-qfi0}
and~\eqref{eq:numerical-ssr-qfi0} are the two no-loss normalization matrices
used below.

\subsection{Finite depth, finite size, and the local half-erasure crossover}

We first use independent local brickwork encoders without imposing charge conservation.  The codewords at the two nodes are generated independently, no gate connects the nodes, and erasures are flagged.  For each $n\in\{3,5,7\}$ we sample 30 brickwork circuits and 60 local Haar isometries.  Erasure-position pairs are enumerated exactly when feasible and otherwise sampled without replacement, with at most 120 pairs.  Table~\ref{tab:finite-size} compares a depth $L=3n$ with the Haar benchmark at
\begin{equation}
k_*=\left\lfloor\frac{n-1}{2}\right\rfloor,
\end{equation}
the largest integer loss on the protected side of one half.

\begin{table*}[t]
\caption{Finite-size worst-direction QFI retention for symmetric flagged erasures $k_A=k_B=k_*$ at depth $L=3n$.  The fifth percentile is taken over erasure positions for each encoder and then averaged over encoders.}
\label{tab:finite-size}
\begin{ruledtabular}
\begin{tabular}{cccccccc}
$n$ & $k_*$ & $L$ & brickwork mean & Haar mean & mean/Haar & brickwork q05 & Haar q05\\
\hline
3 & 1 & 9  & 0.5699 & 0.6004 & 0.949 & 0.4147 & 0.4691\\
5 & 2 & 15 & 0.5530 & 0.5655 & 0.978 & 0.4666 & 0.4727\\
7 & 3 & 21 & 0.5495 & 0.5543 & 0.991 & 0.4973 & 0.5051
\end{tabular}
\end{ruledtabular}
\end{table*}

Across these three sampled sizes, the point estimates are numerically close
to the corresponding Haar values and the retention curves display a
finite-size drop around $k/n=1/2$.  At $n=7$, for example, the Haar means are
$0.9762$, $0.8835$, $0.5543$, and $0.0747$ for $k=1,2,3,4$, respectively.  The
data therefore exhibit a finite-size drop near one half but do not establish
convergence or a threshold.  The Singleton
bound in Eq.~\eqref{eq:exact-half-limit} supplies a no-cloning-compatible
upper limit for exact abstract codes, but the present data do not establish a
covariant asymptotic threshold.  As sanity checks, the repetition/Greenberger--Horne--Zeilinger (GHZ) encoder
loses the worst-direction QFI after one erased memory, whereas the exact
$[[5,1,3]]$ code retains unity for every flagged pattern with $k\leq2$ and
fails outside its exact correction radius.

We next impose number conservation.  Each node contains $n=5$ physical memories and two ancillary excitations.  The logical zero and one codewords then occupy total-charge sectors 2 and 3, respectively, each of dimension $\binom52=\binom53=10$.  Number-conserving two-memory gates generate a brickwork ensemble, and independent Haar unitaries within the two charge sectors define the charge-Haar benchmark, a finite-size instance of the random $U(1)$-covariant construction analyzed in Ref.~\cite{Kong2022}.  At depth $L=15$, averaging 30 circuits and all position pairs gives the SSR-assisted retentions
\begin{equation}
\begin{array}{c|ccc}
k & 1 & 2 & 3\\ \hline
\text{brickwork} & 0.8264 & 0.4793 & 0.0720\\
\text{charge Haar} & 0.8344 & 0.4780 & 0.0691
\end{array}.
\label{eq:charge-haar-comparison}
\end{equation}
The comparison in Eq.~\eqref{eq:charge-haar-comparison} gives finite-depth to
charge-Haar mean ratios of $99.0\%$ and $100.3\%$ at $k=1$ and $k=2$.
The value slightly above $100\%$ is a finite-ensemble fluctuation and is not
evidence that the circuit outperforms the Haar ensemble.  The SSR-assisted
and unrestricted normalized retentions agree within $0.18$ percentage points
in this calculation.  This finite-size agreement is consistent with the ideal
reference factor in Lemma~\ref{lem:one-reference-qfi}; it does not show that
the reference is dispensable, because the unreferenced locally twirled QFI is
identically zero by Lemma~\ref{lem:twirl}.

\subsection{Scrambling, tangent-leakage tails, and recovery risk}
\label{sec:leakage-mechanism}

Let \(U_L\) be the complete \(n\)-memory encoder unitary after \(L\)
brickwork layers, \(d=2^n\), and
\(Z_\ell=I-2\hat n_\ell\).  The remote density out-of-time-order correlator (OTOC), a diagnostic of
operator spreading, is
\begin{equation}
C_{\rm dens}(L)
=
\frac{1}{n-1}
\sum_{\ell=2}^{n}
\frac{
\left\|
\left[U_LZ_1U_L^\dagger,Z_\ell\right]
\right\|_F^2
}{2d}.
\label{eq:density-otoc}
\end{equation}
It is evaluated on the full physical Hilbert space and averages the
normalized squared commutator over sites remote from the initially occupied
logical memory.  This is a unitary-spreading diagnostic only: it is not an
information metric and does not certify decoupling, recovery, or metrological
advantage.

To diagnose why scrambling changes erasure robustness, write the two encoded
logical branches across the retained--erased cut as amplitude matrices \(U\)
and \(V\), whose rows index the retained computational basis and whose
columns index the erased basis.  For the two real visibility directions, the
erased-environment derivatives are
\begin{equation}
\begin{aligned}
\partial_R\rho_E&=\frac12(C_E+C_E^\dagger),
&
\partial_I\rho_E&=\frac12(-iC_E+iC_E^\dagger),
\\[-2pt]
C_E&=U^{\mathsf T}V^*.
\end{aligned}
\end{equation}
We define the task-directed tangent leakage
\begin{equation}
\epsilon_{\rm tan}(E)
=
\left[
\frac{\|\partial_R\rho_E\|_1^2+\|\partial_I\rho_E\|_1^2}{2}
\right]^{1/2}.
\label{eq:tangent-leakage}
\end{equation}
It quantifies local distinguishability of visibility changes in the subsystem that is discarded.  It is not assumed to add to the retained QFI\@.

For $n=5$, two erasures per node, 30 complete number-conserving brickwork trajectories, and all 100 two-node position pairs, Table~\ref{tab:leakage-mechanism} shows a redistribution rather than a monotone removal of leakage.  From $L=0$ to $L=15$, the position-averaged leakage changes from $0.1600$ to $0.2061$, while its within-encoder position standard deviation falls from $0.3685$ to $0.0647$ and its position 95th percentile falls from $1$ to $0.3234$.  Within the 100 enumerated position pairs, the zero sample minima at $L=0$ are replaced by positive sample minima at the tested nonzero depths; this is a finite-size observation, not a uniform lower-bound theorem.

\begin{table*}[t]
\caption{Scrambling diagnostics, erased-environment tangent-leakage tail, and recovered SSR-QFI retention for $n=5$ and $k_A=k_B=2$.  The OTOC is defined in Eq.~\eqref{eq:density-otoc}.  QFI columns use the generalized worst-direction ratio of Eq.~\eqref{eq:numerical-rmin}; the position means, q05/q95 summaries, and worst-position values are computed within each encoder and then averaged over 30 encoders.}
\label{tab:leakage-mechanism}
\begin{ruledtabular}
\begin{tabular}{cccccc}
$L$ & $C_{\rm dens}$ & q95 $\epsilon_{\rm tan}$ & mean recovered QFI & q05 recovered QFI & worst recovered QFI\\
\hline
0  & 0      & 1.0000 & 0.3600 & 0      & 0\\
2  & 0.3153 & 0.7092 & 0.4494 & 0.0705 & 0.0253\\
5  & 0.5170 & 0.4952 & 0.3889 & 0.1485 & 0.1001\\
10 & 0.6758 & 0.3455 & 0.3486 & 0.1896 & 0.1518\\
15 & 0.7031 & 0.3234 & 0.3432 & 0.1879 & 0.1515
\end{tabular}
\end{ruledtabular}
\end{table*}

The depth-controlled statistic is the Pearson correlation of within-depth
percentile-rank residuals, with \(L=0\) excluded.  Its percentile intervals
use 5000 cluster-bootstrap draws that resample the 30 complete brickwork
trajectories with replacement using a fixed pseudorandom stream recorded in
the reproducibility archive.  This gives
\(-0.351\) between OTOC and high-tail leakage (95\% interval
\([-0.498,-0.177]\)), \(-0.268\) between high-tail leakage and lower-tail
recovered QFI (\([-0.445,-0.088]\)), and \(+0.345\) between OTOC and
lower-tail QFI (\([0.163,0.496]\)).  These associations support the proposed
mechanism but are not a causal theorem.

The same data reveal an average--risk tradeoff.  Here the position
conditional value at risk, \(\mathrm{CVaR}_{5\%}\), is the mean over the worst
5\% of erasure positions.  The pairs (mean QFI, position
\(\mathrm{CVaR}_{5\%}\)) are $(0.4494,0.0485)$, $(0.3889,0.1225)$,
$(0.3486,0.1689)$, and $(0.3432,0.1695)$ for $L=2,5,10,15$, respectively;
charge-Haar gives $(0.3455,0.1784)$.  A symmetric Dicke encoder gives the
location-independent pair $(0.2777,0.2777)$.  This diagnostic-stage calculation originally omitted its preparation cost; Appendix~\ref{sec:abstract-encoder-resource} supplies a common-gate-set depth-\(7\) realization without changing the diagnostic interpretation.  Thus ``more scrambling'' is not a scalar optimization principle: shallow circuits maximize the mean, whereas deeper or structured encoders reduce erasure-position risk.

\subsection{SSR-covariant recovery optimized for metrology}

Retained QFI is not yet an operational receiver.  We therefore compare three
local, pattern-dependent CPTP recovery maps for $n=5$, $L=10$, and
$k_A=k_B=2$: the transpose, or Petz, map (the canonical recovery built
from the noise channel and a reference code state), an SSR-covariant
semidefinite program maximizing entanglement fidelity, and an SSR-covariant
program that weights
the two logical coherence tangent operators more strongly than the population
operators.

\paragraph*{Erasure map, Choi convention, and covariance.}
The optimization is performed independently for each site and each known
erased set $E$.  Let $S=E^c$ denote the retained subsystem, let
$m=|S|=n-|E|$, and set $D=2^m$.  For encoding isometry $V$, the
flagged-erasure channel from the logical qubit to the retained memories is
\begin{equation}
{\cal E}_E(X)
=
\operatorname{Tr}_E\!\left[VXV^\dagger\right].
\label{eq:erasure-noise-map}
\end{equation}
Writing the encoded logical codewords as
\begin{equation}
V|i\rangle
=
\sum_{a=0}^{D-1}\sum_\alpha
(M_i^E)_{a\alpha}|a\rangle_S|\alpha\rangle_E,
\qquad i=0,1,
\end{equation}
the four operator images supplied to the optimization are
\begin{equation}
A_{ij}^E
:=
{\cal E}_E(|i\rangle\langle j|)
=
M_i^E(M_j^E)^\dagger.
\label{eq:erasure-operator-images}
\end{equation}

For a recovery channel
${\cal R}_E:{\cal B}({\cal H}_S)\rightarrow{\cal B}(\mathbb C^2)$, we
represent the channel by its Choi matrix and use the unnormalized, input-first
convention
\begin{equation}
J_E
=
\sum_{a,b=0}^{D-1}
|a\rangle\langle b|\otimes
{\cal R}_E(|a\rangle\langle b|).
\label{eq:recovery-choi-convention}
\end{equation}
Thus, with $r,s\in\{0,1\}$,
\begin{align}
\bigl[{\cal R}_E(X)\bigr]_{rs}
&=
\sum_{a,b=0}^{D-1}X_{ab}(J_E)_{ar,bs}
\nonumber\\
&=
\left\{
\operatorname{Tr}_S\!\left[
(X^{\mathsf T}\otimes I_2)J_E
\right]
\right\}_{rs}.
\label{eq:recovery-choi-action}
\end{align}
Complete positivity and trace preservation are therefore
\begin{equation}
J_E\succeq0,
\qquad
\operatorname{Tr}_L J_E=I_D,
\qquad
\sum_{r=0}^1(J_E)_{ar,br}=\delta_{ab}.
\label{eq:recovery-cptp}
\end{equation}
In particular, $\operatorname{Tr}J_E=D$; $J_E$ is not a normalized Choi
state.

Let
\begin{equation}
U_S(\theta)=e^{-i\theta N_S},
\qquad
U_L(\theta)=e^{-i\theta n_L},
\qquad
n_L=|1\rangle\langle1|.
\end{equation}
The covariance condition implemented numerically is
\begin{equation}
{\cal R}_E\!\left(
U_S(\theta)XU_S(\theta)^\dagger
\right)
=
U_L(\theta){\cal R}_E(X)U_L(\theta)^\dagger.
\label{eq:recovery-covariance}
\end{equation}
Equivalently,
\begin{equation}
\left[
J_E,\,
-N_S^{\mathsf T}\otimes I_2+I_D\otimes n_L
\right]=0.
\label{eq:recovery-covariance-choi}
\end{equation}
If $q_a$ is the Hamming weight of retained computational-basis state
$|a\rangle$, Eq.~\eqref{eq:recovery-covariance-choi} gives the explicit
selection rule
\begin{equation}
(J_E)_{ar,bs}=0
\quad\text{unless}\quad
r-q_a=s-q_b.
\label{eq:recovery-charge-selection}
\end{equation}
The code therefore constructs
\begin{equation}
J_E=\bigoplus_{\Delta=-m}^{1}J_E^{(\Delta)},
\qquad
J_E^{(\Delta)}\succeq0,
\qquad
\Delta=r-q_a.
\label{eq:recovery-charge-blocks}
\end{equation}
The dimension of the $\Delta$ block is
\begin{equation}
d_\Delta
=
\sum_{r=0}^1\binom{m}{r-\Delta},
\label{eq:recovery-block-dimension}
\end{equation}
where an out-of-range binomial coefficient is zero.  For the reported
$n=5$, $|E|=2$ calculation, $m=3$, and the block dimensions for
$\Delta=-3,-2,-1,0,1$ are $1,4,6,4,1$, respectively.

\paragraph*{Maximum-fidelity SDP.}
Let $\Lambda_E={\cal R}_E\circ{\cal E}_E$ and
$|\Phi_2\rangle=(|00\rangle+|11\rangle)/\sqrt2$.  The entanglement fidelity
for the maximally mixed logical input is
\begin{align}
F_e(\Lambda_E)
&=
\langle\Phi_2|
(\operatorname{id}\otimes\Lambda_E)
(|\Phi_2\rangle\langle\Phi_2|)
|\Phi_2\rangle
\nonumber\\
&=
\frac14\sum_{i,j=0}^1
\langle i|\Lambda_E(|i\rangle\langle j|)|j\rangle
\nonumber\\
&=
\frac14\operatorname{Re}
\sum_{i,j=0}^1\sum_{a,b=0}^{D-1}
(A_{ij}^E)_{ab}(J_E)_{ai,bj}.
\label{eq:recovery-entanglement-fidelity}
\end{align}
Consequently, the maximum-fidelity recovery is the SDP
\begin{widetext}
\begin{equation}
\begin{aligned}
\underset{J_E}{\operatorname{maximize}}\quad&
\frac14\operatorname{Re}
\sum_{i,j=0}^1\sum_{a,b=0}^{D-1}
(A_{ij}^E)_{ab}(J_E)_{ai,bj}
\\
\operatorname{subject\ to}\quad&
J_E^{(\Delta)}\succeq0
\quad(\Delta=-m,\ldots,1),
\\
&
\sum_{r=0}^1(J_E)_{ar,br}=\delta_{ab}
\quad(a,b=0,\ldots,D-1),
\\
&
(J_E)_{ar,bs}=0
\quad\text{if}\quad r-q_a\neq s-q_b.
\end{aligned}
\label{eq:fidelity-recovery-sdp}
\end{equation}
\end{widetext}
The factor $1/4=d_L^{-2}$ makes the objective equal to one for the identity
logical channel.  The transpose occurs in
Eq.~\eqref{eq:recovery-choi-action}; with the input-first Choi indices used
here, the coefficient in
Eq.~\eqref{eq:recovery-entanglement-fidelity} is $(A_{ij}^E)_{ab}$, without
an additional transpose or complex conjugation.

\paragraph*{Tangent-weighted SDP.}
The code uses the Hilbert--Schmidt orthonormal logical basis
\begin{equation}
B_0=\frac{I}{\sqrt2},
\quad
B_z=\frac{Z}{\sqrt2},
\quad
B_x=\frac{X}{\sqrt2},
\quad
B_y=\frac{Y}{\sqrt2},
\label{eq:tangent-logical-basis}
\end{equation}
with $\operatorname{Tr}(B_\mu^\dagger B_\nu)=\delta_{\mu\nu}$.
Here $B_0,B_z$ span the population sector, whereas $B_x,B_y$ are the two
Hermitian coherence directions carrying the real and imaginary visibility
tangents.  Define the affine recovery errors
\begin{equation}
D_{\mu,E}(J_E)
:=
\Lambda_E(B_\mu)-B_\mu,
\label{eq:tangent-error-definition}
\end{equation}
whose matrix elements are
\begin{align}
[D_{\mu,E}(J_E)]_{rs}
={}&
\sum_{i,j=0}^1\sum_{a,b=0}^{D-1}
(B_\mu)_{ij}(A_{ij}^E)_{ab}(J_E)_{ar,bs}
\nonumber\\
&-(B_\mu)_{rs}.
\label{eq:tangent-error-elements}
\end{align}
For $(w_0,w_z,w_x,w_y)=(1,1,\lambda,\lambda)$, the implemented convex
objective is
\begin{equation}
{\cal L}_\lambda(J_E)
=
\sum_{\mu\in\{0,z,x,y\}}
w_\mu\left\|D_{\mu,E}(J_E)\right\|_F^2.
\label{eq:tangent-weighted-loss}
\end{equation}
Equivalently, if $P_i=|i\rangle\langle i|$, then
\begin{align}
{\cal L}_\lambda
={}&
\sum_{i=0}^1
\left\|\Lambda_E(P_i)-P_i\right\|_F^2
\nonumber\\
&+
\lambda\sum_{i\neq j}
\left\|
\Lambda_E(|i\rangle\langle j|)
-|i\rangle\langle j|
\right\|_F^2,
\label{eq:tangent-population-coherence-form}
\end{align}
which makes the relative population/coherence weighting explicit.

The implementation minimizes Eq.~\eqref{eq:tangent-weighted-loss} directly,
subject to Eqs.~\eqref{eq:recovery-cptp}
and~\eqref{eq:recovery-charge-selection}.  To exhibit the same problem as a
standard linear-objective SDP, introduce real variables $t_\mu$ and
\begin{equation}
z_{\mu,E}(J_E)
=
\sqrt{w_\mu}\,
\operatorname{vec}\!\left[D_{\mu,E}(J_E)\right]\in\mathbb C^4.
\end{equation}
The exactly equivalent Schur-complement formulation is
\begin{widetext}
\begin{equation}
\begin{aligned}
\underset{J_E,\{t_\mu\}}{\operatorname{minimize}}\quad&
\sum_{\mu\in\{0,z,x,y\}}t_\mu
\\
\operatorname{subject\ to}\quad&
J_E^{(\Delta)}\succeq0
\quad(\Delta=-m,\ldots,1),
\\
&
\sum_{r=0}^1(J_E)_{ar,br}=\delta_{ab}
\quad(a,b=0,\ldots,D-1),
\\
&
(J_E)_{ar,bs}=0
\quad\text{if}\quad r-q_a\neq s-q_b,
\\
&
\begin{pmatrix}
t_\mu & z_{\mu,E}(J_E)^\dagger\\
z_{\mu,E}(J_E) & I_4
\end{pmatrix}
\succeq0
\quad(\mu=0,z,x,y).
\end{aligned}
\label{eq:tangent-recovery-sdp}
\end{equation}
\end{widetext}
The basis operators in Eq.~\eqref{eq:tangent-logical-basis} have unit
Hilbert--Schmidt norm.  No further division of ${\cal L}_\lambda$ by
$2(1+\lambda)$ is made in the implementation.  Thus the code solves a convex
quadratic objective with semidefinite constraints; the lifted formulation
in Eq.~\eqref{eq:tangent-recovery-sdp} establishes its exact SDP form.

The shared hyperparameter is selected from
$\lambda\in\{1,2,4,8,16,32,64,128\}$.  The training set contains 20
independently generated two-node encoder pairs.  For each pair, 25 of the 100
two-node erasure-position pairs are sampled without replacement, and the
selection statistic is
\begin{equation}
S(\lambda)
=
\frac12\,\overline{\eta_{\rm ext}}
+
\frac12\,Q_{0.05}(\eta_{\rm ext}),
\label{eq:tangent-lambda-score}
\end{equation}
where the mean and pooled fifth percentile are taken over the resulting 500
training instances.  This selects $\lambda=128$, which is then frozen.
Held-out validation uses 10 new two-node encoder pairs and all
$\binom52^2=100$ erasure-position pairs for each pair.  The proxy loss in
Eq.~\eqref{eq:tangent-weighted-loss} is used only to construct each local
recovery; model selection and validation are evaluated using the physical
SSR-QFI and GJC classical Fisher information.  Before recovery the
validation-set mean QFI retention is $0.4866$.  The frozen paper-mode run
uses the Splitting Conic Solver (SCS) with numerical tolerance
\(10^{-7}\), at most \(10^5\) iterations, and solver normalization enabled.

\begin{table*}[t]
\caption{Held-out-validation performance at $n=5$, $L=10$, and $k_A=k_B=2$.  $\overline{R}_H$ and $\overline{R}_F$ are the mean QFI and GJC-CFI retentions relative to their respective no-loss baselines; ``worst'' denotes the erasure-position minimum averaged over encoders.  The decoder, measurement, and post-erasure extraction efficiencies in Eq.~\eqref{eq:three-efficiencies} are computed samplewise as generalized worst-direction ratios and then averaged.}
\label{tab:recovery-comparison}
\begin{ruledtabular}
\begin{tabular}{lccccccc}
recovery & $\overline{R}_H$ & $R_H^{\mathrm{worst}}$ & $\overline{R}_F$ & $R_F^{\mathrm{worst}}$ & $\eta_{\mathrm{dec}}$ & $\eta_{\mathrm{meas}}$ & $\eta_{\mathrm{ext}}$\\
\hline
Petz & 0.3105 & 0.0981 & 0.3298 & 0.1075 & 0.6089 & 0.4712 & 0.2861\\
maximum fidelity & 0.3381 & 0.1093 & 0.3578 & 0.1195 & 0.6688 & 0.4688 & 0.3128\\
tangent weighted & 0.3567 & 0.1256 & 0.3762 & 0.1370 & 0.7127 & 0.4663 & 0.3318
\end{tabular}
\end{ruledtabular}
\end{table*}

Relative to the maximum-fidelity recovery, the tangent-weighted map improves the mean QFI retention by $5.50\%$, the GJC CFI by $5.15\%$, the decoder efficiency by $6.56\%$, and the post-erasure extraction efficiency by $6.09\%$.  Its worst-position QFI improves by $14.98\%$.  Paired encoder bootstrap gives absolute Tangent-minus-maximum-fidelity changes of \(0.01861\) \([0.01643,0.02071]\), \(0.01843\) \([0.01617,0.02064]\), \(0.01637\) \([0.01259,0.02033]\), and \(0.01747\) \([0.01342,0.02169]\) for mean QFI, mean GJC CFI, worst-position QFI, and worst-position GJC CFI, respectively.  These are pointwise 95\% percentile intervals from 5000 joint resamples of the 10 encoders, without multiplicity correction.  The selected value $\lambda=128$ lies at the edge of the scan, and the final tested doubling changes the training selection score by only $0.206\%$.  The scan is locally flat at its tested boundary, but it does not exclude a better value outside the grid.  Complete positivity, trace preservation, charge covariance, QFI data processing, and the Choi-matrix convention were checked numerically for every reported channel.

\subsection{GJC phase design and compatibility}

For ideal reference visibility and equal allocation between the conventional phases $0$ and $\pi/2$, Eq.~\eqref{eq:gjc-fisher} gives at $g_0$
\begin{equation}
F_{0,\pi/2}^{\mathrm{GJC}}
=
\begin{pmatrix}
0.274725 & 0\\
0 & 0.297619
\end{pmatrix}.
\label{eq:numerical-gjc-f0}
\end{equation}
The matrix in Eq.~\eqref{eq:numerical-gjc-f0} has worst-direction efficiency
$0.4314$ relative to Eq.~\eqref{eq:numerical-ssr-qfi0}.  Equal allocation
among $0,\pi/4,\pi/2,3\pi/4$ raises the no-loss efficiency to $0.4642$.
The discrete oracle selects, from 16 uniformly spaced candidate phases on
$[0,\pi)$, the phase nearest to $\arg g$ modulo $\pi$ and combines it equally
with the orthogonal grid setting; this 16-point phase-grid oracle reaches
$0.4922$.  It is not a continuous true-phase-alignment oracle.  A robust
fixed-design semidefinite program trained over visibility phase selects
exactly the uniform four-phase allocation.

For any fixed phase schedule \(\{(\delta_m,w_m)\}\), all quoted Fisher
matrices are normalized per total attempted input mode:
\begin{equation}
F_{\mathrm{mix}}
=
\sum_m w_m F_{\delta_m},
\qquad
w_m\geq0,
\qquad
\sum_m w_m=1.
\label{eq:phase-budget-normalization}
\end{equation}
The uniform four-phase receiver therefore already contains \(w_m=1/4\).
Moving from two to four settings repartitions one fixed trial budget; it does
not multiply the trial count or introduce another factor of \(1/4\) in the
rate envelope.  Platform-specific switching and phase-stabilization latency
is not included.

With the tangent-weighted decoder frozen, validation on eight independent encoders, all $k=2$ erasure positions, and 25 complex visibilities with radii $0.2$, $0.5$, and $0.8$ gives the results in Table~\ref{tab:phase-design}.  The four-phase design changes the mean asymptotic measurement efficiency only from $0.4754$ to $0.4761$, but raises the mean within-encoder fifth percentile from $0.3813$ to $0.4295$ and the mean within-encoder minimum over the tested position--visibility grid from $0.3222$ to $0.4100$.  Paired encoder bootstrap gives four-phase-minus-two-phase changes of \(0.000632\) \([0.000559,0.000700]\), \(0.04815\) \([0.04015,0.05579]\), and \(0.08786\) \([0.07975,0.09613]\) for the mean, q05, and sampled minimum, respectively.  These are pointwise 95\% percentile intervals from 5000 joint resamples of the eight encoders.  This is a robustness improvement rather than a large mean-information gain: the mean absolute CFI retention rises by only $0.16\%$, while its mean within-encoder fifth percentile changes from $0.05585$ to $0.05504$ with paired difference \(-0.000813\) \([-0.001432,-0.000147]\).  Because the schedule is fixed before the visibility orientation and erasure pattern are known, the q05 and sampled-minimum efficiencies diagnose phase-agnostic receiver mismatch on the validation grid.  The result does not establish an improvement in absolute low-tail CFI, finite-sample estimator risk, or robustness to atmospheric phase fluctuations.

\begin{table}[t]
\caption{Local GJC phase-design efficiency on the finite validation grid.
The $g_0$ column is the
no-loss value at Eq.~\eqref{eq:numerical-g0}; the remaining columns summarize
the validation grid.  The mean, q05, and minimum are computed within each
encoder's position--visibility grid and then averaged over the eight
encoders.  The 16-point grid-oracle row selects the nearest
candidate phase to the true visibility phase modulo $\pi$ and its orthogonal
grid setting.  It is a discrete phase-informed upper benchmark rather than a
continuous oracle or an implementable fixed receiver.}
\label{tab:phase-design}
\setlength{\tabcolsep}{3pt}
\begin{ruledtabular}
\begin{tabular}{lcccc}
design & $g_0$ & mean & q05 & min.\\
\hline
fixed $0,\pi/2$ & 0.4314 & 0.4754 & 0.3813 & 0.3222\\
uniform four & 0.4642 & 0.4761 & 0.4295 & 0.4100\\
16-point grid oracle & 0.4922 & 0.4989 & 0.4950 & 0.4849
\end{tabular}
\end{ruledtabular}
\end{table}

Multiparameter compatibility also changes under approximate recovery.  We quantify the weak SLD commutator by
\begin{equation}
\mathcal{C}_{RI}
=
\frac{\left\lvert
\operatorname{Tr}\!\left[\rho(L_RL_I-L_IL_R)\right]/(2i)
\right\rvert}
{\sqrt{H_{RR}H_{II}}}.
\label{eq:weak-commutator}
\end{equation}
The diagnostic in Eq.~\eqref{eq:weak-commutator} is numerically zero
($<4\times10^{-17}$) at the no-loss work point.  Over 800 decoded validation
instances it has mean $0.0710$, median $0.0564$, 95th percentile $0.1933$,
and maximum $0.2962$.  Thus QFI retention after a generic approximate decoder
need not be jointly attainable by one two-parameter measurement, reinforcing
the need to report the actual GJC CFI\@.

Finite-sample adaptive optimization is not included in the certified claim
set.  Estimator bias, physical-boundary effects, pilot cost, feedback latency,
and realistic capture and detection rates have not yet been assessed under a
single frozen protocol.  The fixed uniform four-phase receiver is therefore
used as a nonadaptive numerical default; phase switching, allocation, and
control costs are not modeled.

\subsection{All-pattern normalized retention and the locality bottleneck}
\label{sec:all-pattern-rate}

We next retain every flagged event and evaluate Eq.~\eqref{eq:all-pattern-cfi},
rather than fixing $k$ or conditioning on successful recovery.  Each node has
32 erasure sets, giving 1024 two-node events.  All encoded designs are
evaluated under one common protocol: the same five prescribed visibility
points, a uniform four-phase receiver with readout-mixing probability
$0.005$, the same procedure for constructing an exact (uncompressed and
uncompiled) pattern-dependent local recovery channel, and the same i.i.d.
Bernoulli erasure weights with perfect flags.  These are model branch
probabilities, not a complete physical event distribution.  Here ``exact'' distinguishes the
abstract numerical recovery map from its later circuit approximation; it
does not assert perfect correction of every erasure.  The local decoder and
receiver are fixed independently of $g$.  Table~\ref{tab:all-pattern-rate}
reports the CFI generalized worst-direction ratio at the worst of the five
visibility points, averaged at the encoder-pair level.  The equal-budget candidate comparison contains 30 trajectory clusters.  Within each cluster, the
depth-$2$ encoder at each node is the exact two-layer prefix of the
corresponding depth-$10$ trajectory, so the \(L2-L10\) contrast is paired.
The 30 charge-Haar target pairs were intended to use a separate fixed
stream, but a later audit found that the trajectory-node-\(B\) and
charge-Haar-node-\(A\) integer seeds coincide within each sample.  Different
generators do not establish statistical independence, so the nominal
independent-ensemble Haar intervals are withdrawn.  The corrected
common-gate-set analysis resamples the full \(L=2/L=7/\)charge-Haar
sample-index cluster jointly.  The statistical unit is one complete two-node encoder pair: the
1024 erasure events, five visibility points, and sampled \(p\) values within
a pair are not independent replicates.  The symmetric Dicke construction is
one deterministic benchmark and receives no ensemble confidence interval.

The equal-budget candidate-comparison random streams are also disjoint
from those of the previously frozen 30-pair \(L=2\) implementation batch used
for gate-set compilation, Fig.~\ref{fig:conditional-window}, and
Table~\ref{tab:hardware-window}.  The equal-budget comparison uses exact,
uncompressed, uncompiled pattern-dependent recovery.  It supports retaining \(L=2\), but
does not retrospectively replace the implementation ensemble or match the
physical preparation costs of charge-Haar and Dicke.

\begin{table*}[t]
\caption{Equal-statistical-budget random-ensemble all-pattern screen, with a
deterministic Dicke benchmark, of the normalized four-phase GJC-CFI
retention at the worst of
five prescribed visibility points.
Every flagged event, including zero-information events, is weighted by its
model probability under i.i.d.\ erasure and perfect flags; the readout-mixing
probability is $0.005$, and encoded
rows use the same exact-recovery construction without compression or
compilation.  The bare row is the analytic value $(1-p)^2$.
The \(L=2\) and \(L=10\) rows are paired prefixes and full trajectories from
30 clusters; charge-Haar uses 30 frozen target pairs.  The table reports
means only.  Because the seed audit does not support independence of the
trajectory and charge-Haar samples, corrected inference is reported in
Table~\ref{tab:abstract-encoder-resource}.  Symmetric Dicke is one
deterministic benchmark, not an ensemble estimate.}
\label{tab:all-pattern-rate}
\begin{ruledtabular}
\begin{tabular}{lcccccc}
design & $N_{\rm pair}$ & $p=0.1$ & $p=0.2$ & $p=0.3$ & $p=0.4$ & $p=0.5$\\
\hline
bare storage & --- & 0.8100 & 0.6400 & 0.4900 & 0.3600 & 0.2500\\
brickwork $L=2$ & 30 & 0.8552 & 0.7070 & 0.5614 & 0.4240 & 0.3001\\
brickwork $L=10$ & 30 & 0.8561 & 0.6809 & 0.4996 & 0.3344 & 0.2004\\
charge-Haar & 30 & 0.8624 & 0.6894 & 0.5069 & 0.3387 & 0.2021\\
symmetric Dicke & 1 & 0.7302 & 0.5441 & 0.4188 & 0.3337 & 0.2717\\
\end{tabular}
\end{ruledtabular}
\end{table*}

The primary work point \(p=0.30\) was fixed before the equal-budget screen.  Its
\(L2-L10\) point estimate is \(0.061747\), and its \(L2-\)charge-Haar target
point estimate is \(0.054533\).  The previously reported simultaneous
intervals treated the latter contrast as independent and are not used here
after the seed audit.  The corrected, common-gate-set inference is
given next.  Descriptively, the table's point-estimate winner is charge-Haar
at \(p=0.05,0.10\) and \(L=2\) at
\(p=0.15,0.20,\ldots,0.50\); this does not establish uniform optimality on
the grid or over a continuous interval.  For the \(L=2\) means, the
break-even ratios \(\gamma_{\rm crit}=(1-p)^2/R_C\) are
\(0.947,0.905,0.873,0.849,\) and \(0.833\) at the five table columns.  These
conditional ratios omit the memory-only \(1/n\) throughput penalty and do not
resolve recovery or platform resources.

\subsection{Common abstract gate-set encoder online-gate comparison}
\label{sec:abstract-encoder-resource}

The follow-up compares encoder preparation in one common abstract gate model.
The online-gate ledger includes circuit depth, two-qubit gate count,
storage qubits, and synthesis work ancillas, but excludes offline compiler
search and the physical precision or calibration cost of continuous gate
parameters.  Every node is a five-site linear chain, begins with
\[
V_{\rm in}=(|01100\rangle,|11100\rangle),
\]
and uses no synthesis work ancilla beyond the two fixed ancillary
excitations already present in the storage block.  A circuit layer contains
the two disjoint nearest-neighbor gates on bonds \((0,1),(2,3)\) or
\((1,2),(3,4)\), alternating between layers.  Each gate has the five-parameter
number-conserving form
\begin{equation}
G(\boldsymbol\theta)
=
1\oplus U_1(\theta,\chi,\zeta,\eta)
\oplus e^{i\alpha_2},
\qquad U_1\in U(2).
\label{eq:encoder-native-u1-gate}
\end{equation}
The numerical compiler preserves the relative phase of the two logical
columns; only one common global phase may be removed.  For target and circuit
isometries \(V_{\rm tar},V_{\rm circ}\in\mathbb C^{32\times2}\), acceptance
requires
\begin{equation}
\begin{aligned}
\epsilon_{\rm iso}
&=
1-\frac14
\left|
\operatorname{Tr}\!\left(V_{\rm tar}^{\dagger}V_{\rm circ}\right)
\right|^2
\leq10^{-10},
\\
\delta_V
&=
\min_{\phi}
\frac{\|V_{\rm circ}-e^{i\phi}V_{\rm tar}\|_F}{\sqrt2}
\leq10^{-5}.
\end{aligned}
\label{eq:encoder-isometry-acceptance}
\end{equation}
Only the two columns relevant to the logical input are synthesized.
Particle-number-preserving state-preparation constructions and
linear-nearest-neighbor Dicke-state circuits are known~\cite{Arrazola2022,Bartschi2019}, but their controlled-excitation or
standard-gate primitives differ from Eq.~\eqref{eq:encoder-native-u1-gate}.
Those references therefore do not determine the costs reported here.

\begin{table*}[t]
\caption{Encoder-only online resources in the common abstract nearest-neighbor
\(U(1)\) gate set.  All designs use five storage qubits per node, two fixed
ancillary excitations, and no additional synthesis work ancilla.
\(R_C(0.30)\) is the mean worst-work-point all-pattern GJC-CFI retention.
The compiled charge-Haar and Dicke targets satisfy the stated numerical
tolerances within depth \(7\); these synthesis costs are verified upper
bounds, not proven minima.  Recovery, offline compiler search,
continuous-control precision, and calibrated-platform resources are
excluded.}
\label{tab:abstract-encoder-resource}
\begin{ruledtabular}
\begin{tabular}{lcccccc}
design & \(N_{\rm pair}\) & \(D/\)node & gates/node &
gates/two nodes & within \(D\leq7\) & \(R_C(0.30)\)\\
\hline
brickwork \(L=2\) & 30 & 2 & 4 & 8 & yes & 0.561392\\
brickwork \(L=7\) & 30 & 7 & 14 & 28 & yes & 0.509252\\
charge-Haar compiled \(D=7\) & 30 & 7 & 14 & 28 & yes & 0.506859\\
Dicke compiled \(D=7\) & 1 & 7 & 14 & 28 & yes & 0.418808\\
brickwork \(L=10\) & 30 & 10 & 20 & 40 & no & 0.499645\\
\end{tabular}
\end{ruledtabular}
\end{table*}

The charge-Haar rows compile the 60 single-node targets from the 30 frozen
pairs; the single Dicke circuit is reused at both nodes.  All 61 synthesis
targets pass Eq.~\eqref{eq:encoder-isometry-acceptance}.  The maximum
\(\epsilon_{\rm iso}\) is \(8.19\times10^{-12}\), and the maximum
\(\delta_V\) is \(2.86\times10^{-6}\).  Recomputing every actual codeword
for the newly evaluated controls means 30 depth-\(7\) brickwork pairs, 30
compiled charge-Haar pairs, and one compiled Dicke pair, or 122 local
codewords in total.  Their 31 nonempty branches require
\(122\times31=3782\) local recovery SDPs.  The frozen depth-\(2\)
recovery records are reused and excluded from this count.  The maximum
absolute difference between the frozen-target and compiled-circuit
worst-work-point \(R_C\) values over the complete sample--\(p\) grid is
\(1.78\times10^{-6}\), below the prespecified \(10^{-5}\) tolerance.

The primary family contains the five contrasts \(L2-L7\),
\(L2-\mathrm{cH7}\), \(L2-\mathrm{D7}\), \(L7-\mathrm{cH7}\), and
\(L7-\mathrm{D7}\) at \(p=0.30\).  A 10,000-draw max-\(t\) cluster bootstrap
uses one joint sample-index resample for \(L=2,L=7,\) and charge-Haar because
of the disclosed seed overlap; Dicke is fixed.  The first three estimates
and simultaneous 95\% intervals, conditional on the fixed Dicke comparator,
are
\begin{align}
\Delta_{2,7}
&=0.052140\ [0.040052,0.064228],\nonumber\\
\Delta_{2,\mathrm{cH7}}
&=0.054533\ [0.038824,0.070243],\nonumber\\
\Delta_{2,\mathrm{D7}}
&=0.142584\ [0.134519,0.150649].
\label{eq:encoder-resource-primary-contrasts}
\end{align}
For the Dicke contrast, only the stochastic \(L=2\) ensemble is resampled;
the single deterministic benchmark is held fixed.
The \(L7-\mathrm{cH7}\) interval includes zero, so the calculation does not
rank those two depth-\(7\) ensembles.  The positive \(L=2\) intervals and
lower encoder cost establish only a sampled Pareto statement at the primary
work point for the realized circuits.  They do not prove global Pareto
optimality, minimal synthesis depth, a uniform interval in \(p\), or a
platform-complete advantage.

The gap between information present before recovery and information extracted by the product decoder is not primarily an SSR penalty in the finite cases tested.  In 12 prespecified cases spanning depths $2,5,10,15$ and median, 95th-percentile, and maximum decoder gaps, allowing a joint two-node recovery while retaining local $U(1)$ covariance closes $72.6\%$ of the product-decoder gap on average and at least $49.8\%$.  On the four maximum-gap cases, relaxing further to global $U(1)$ gives an additional normalized gain $0.0365$, while removing SSR entirely adds only $0.0023$.  This decomposition identifies product locality as the dominant observed decoder constraint, but it is not a uniform theorem and the main protocol continues to use local product recovery.

\subsection{Abstract recovery circuits and a conditional sensitivity envelope}
\label{sec:native-hardware}

For each of 30 separately generated $n=5$, $L=2$ encoder pairs in the frozen
implementation batch, every local recovery Choi matrix is first restricted to
the subspace reachable after the
corresponding flagged erasure.  For circuit construction, the reachable
channel is truncated at fixed retained spectral mass $0.999$ and then
renormalized to a CPTP map.  The numerical output retains three distinct
layers---the uncompressed exact SDP recovery, this fixed compressed target,
and the reconstructed gate-set circuit---so that compression and compilation effects
are not conflated.  A legacy sweep over alternative compression settings also
produced percentage rank reductions, Givens-count proxies, and a CFI
perturbation estimate; those sweep summaries are outside the present
reproducibility certification and are not used below.  The fixed compressed
support-pruned Stinespring isometry is
compiled into a linear nearest-neighbor circuit of number-conserving gates
\begin{equation}
U^{(2)}=1\oplus U_1\oplus e^{i\alpha_2},\qquad U_1\in U(2),
\label{eq:native-u1-gate}
\end{equation}
with virtual single-qubit $Z$ phases when the reachable circuit has only one
qubit.  Let $J_{\rm target}$ be the fixed compressed target and
\(J_{\rm circ}\) the reconstructed circuit channel.  We certify the
normalized Choi Frobenius residual
\begin{equation}
\delta_J=
\frac{\lVert J_{\rm circ}-J_{\rm target}\rVert_F}
{\max\{1,\lVert J_{\rm target}\rVert_F\}}.
\label{eq:normalized-choi-residual}
\end{equation}
An analytic-gradient optimizer alternates with charge-blocked Procrustes
alignment of the Stinespring environment gauge.  A branch is accepted if its
reconstructed circuit Choi matrix is CPTP and either (i) the optimizer
terminates normally with $\delta_J\leq10^{-4}$, or (ii) a false termination
flag is accompanied by the stricter $\delta_J\leq10^{-7}$.  The second route
is recorded separately and never used to accept a branch near the main
threshold.

\begin{table*}[t]
\caption{Compilation into the abstract number-conserving gate set over 30
separately generated \(L=2\) encoder pairs.  Branches include both nodes.
Gate counts
refer to two-qubit gates per local branch; virtual \(Z\) phases are not
counted and are not assigned time or error.  Environment initialization and
discard, flag acquisition, classical control, work-ancilla occupancy, and
reset are outside this gate-count ledger.  ``Residual'' counts the stricter
posterior checks associated with a false optimizer termination flag.}
\label{tab:native-compilation}
\setlength{\tabcolsep}{3pt}
\begin{ruledtabular}
\begin{tabular}{cccccccc}
$k$ & branches & pass & residual & mean $D$ & max $D$ & mean gates & max gates\\
\hline
1 & 300 & 100\% & 0  & 5.727 & 8 & 8.590 & 12\\
2 & 600 & 100\% & 0  & 4.253 & 8 & 4.253 & 8\\
3 & 600 & 100\% & 3  & 2.360 & 6 & 1.417 & 6\\
4 & 300 & 100\% & 15 & 1.613 & 6 & 1.613 & 6\\
5 & 60  & 100\% & 0  & 2.000 & 2 & 1.000 & 1
\end{tabular}
\end{ruledtabular}
\end{table*}

All 1860 nonempty branches pass.  Eighteen use the stricter residual
criterion; their maximum $\delta_J$ is $4.26\times10^{-8}$, their maximum
trace-preservation Frobenius residual is $4.90\times10^{-16}$, and the most
negative Choi eigenvalue is $-1.52\times10^{-16}$.  The maximum $\delta_J$
over the full ensemble is $9.90\times10^{-5}$ and belongs to a normally
terminated branch.
The Fisher results report separately the gate-set-compiled-to-compressed
retention, the compressed-to-exact effect, and the
gate-set-compiled-to-exact total retention; agreement with the compressed circuit
target is not used as a substitute for agreement with the uncompressed
recovery.  Taking the minimum over the ten tested erasure probabilities of
the 30-pair mean within-pair worst-visibility ratio gives, respectively,
$0.999999987$, $1.000000004$, and $0.999999995$.  The corresponding minima
before averaging over encoder pairs and visibility work points are
$0.999999780$, $1.000000000$, and $0.999999790$.  The small values above
unity are numerical residuals at solver and channel-reconstruction
precision, not physical amplification of Fisher information.

To expose the remaining model assumptions, we use a transparent
\(U(1)\)-covariant synthetic output-coherence attenuation envelope.  After
ideal compiled recovery, the off-diagonal elements of a local branch
containing \(m_E\) gates in the abstract compiled set are multiplied by
\begin{equation}
\eta_E=(1-2e)^{m_E}.
\label{eq:hardware-dephasing-envelope}
\end{equation}
Here \(e\) is an output-coherence attenuation index.  No physical phase-flip
or other microscopic channel is inserted after each gate, and \(e\) is not a
calibrated per-gate error probability.
The same synthetic \(0.5\%\) uniform outcome-mixing rule
\(p_o\mapsto0.995p_o+0.005/K\), for \(K\) modeled outcomes, is applied to
encoded and bare protocols.  It is not a calibrated detector confusion
matrix and does not equalize detector number, dead time, or bandwidth.  The
two nodes are assumed to operate in parallel, so assigning a gate-layer time
\(\tau_2=rT_{\rm bare}\) gives
\begin{equation}
\begin{aligned}
\gamma(p,r)
&=
\left\{
1+r\!\left[
D_{\rm enc}
+\mathbb E_p\max(D_A,D_B)
\right]
\right\}^{-1},
\\[-2pt]
D_{\rm enc}&=2,
\end{aligned}
\label{eq:hardware-rate-envelope}
\end{equation}
Here \(D_A\) and \(D_B\) are the compiled decoder-circuit depths for the
realized local erasure branches, and \(\mathbb{E}_p\) averages independent
Bernoulli erasures at the two nodes.
This serial-depth model omits reference preparation and distribution, reset,
flag acquisition, feed-forward, retrieval, and detection latency.  The model
uses one common synthetic attenuation index and a multiplicative output
factor; the resulting coherence factors depend on the branch through
\(m_E\), which counts only abstract recovery gates.  Encoder
and idle noise, amplitude damping, leakage, crosstalk, correlated faults, and
false erasure flags are not included.  The omitted critical-path and noise
fields are enumerated in Appendix~\ref{sec:platform-contract}.
The conditional normalized excess retention with this partial timing factor is
\begin{equation}
G(p,e,r)=\gamma(p,r)R_C(p,e)-(1-p)^2.
\label{eq:hardware-net-gain}
\end{equation}
This is a sensitivity envelope, not a calibrated platform model or a rigorous worst-case noise bound.

\begin{table*}[t]
\caption{Simultaneous conditional-retention comparison on the \(p=0.3\)
slice of the frozen synthetic grid, using the partial timing proxy and the
worst of five visibility points for each of the same 30 complete
encoder-pair/compilation realizations.  \(e\) is the output-coherence
attenuation index defined in the text, not a calibrated per-gate channel
probability.  \(r_{\rm crit}\) is the zero crossing of the plug-in envelope
formed from mean inputs; it is not a confidence limit, and a dash means that
the plug-in gain is already nonpositive at \(r=0\).  \(G_{0.005}\) is the
pairwise mean gain at \(r=0.005\), and the final column gives its nominal
finite-sample approximate simultaneous 95\% interval from the two-sided
studentized cluster-bootstrap max-\(t\) family over all 360 conditional-gain
cells frozen for the post hoc reanalysis.}
\label{tab:hardware-window}
\begin{ruledtabular}
\begin{tabular}{cccccc}
synthetic \(e\) & \(R_C\) & \(\gamma_{\rm crit}\) & \(r_{\rm crit}\) &
\(G_{0.005}\) & simultaneous 95\% interval\\
\hline
\(0\) & 0.558045 & 0.878065 & 0.01957 & +0.048930 & $[0.037936,\,0.059923]$\\
\(10^{-3}\) & 0.532205 & 0.920697 & 0.01214 & +0.023976 & $[0.013095,\,0.034857]$\\
\(3\times10^{-3}\) & 0.485664 & 1.008927 & --- & -0.020969 & $[-0.031607,\,-0.010331]$\\
\(10^{-2}\) & 0.362879 & 1.350312 & --- & -0.139546 & $[-0.149110,\,-0.129982]$\\
\end{tabular}
\end{ruledtabular}
\end{table*}

Table~\ref{tab:hardware-window} and Fig.~\ref{fig:conditional-window} identify
191\ of the 360 conditional, per-captured-mode cells
with positive simultaneous 95\% lower endpoints.  On the
\(e=10^{-3},r=0.005\) slice, the generated simultaneous-positive and
not-simultaneous-positive evaluated \(p\)-sets are
\(\{0.20,\,0.25,\,0.30,\,0.35,\,0.40,\,0.45,\,0.50\}\) and
\(\{0.05,\,0.10,\,0.15\}\).  On the
\(e=3\times10^{-3},r=0.005\) slice, the corresponding sets are
\(\varnothing\) and
\(\{0.05,\,0.10,\,0.15,\,0.20,\,0.25,\,0.30,\,0.35,\,0.40,\,0.45,\,0.50\}\).  At \(p=0.3\),
19\ of the 36 sampled \((e,r)\) cells have
positive simultaneous lower endpoints.  The nominal finite-sample
approximate max-\(t\) intervals apply to the 360-cell post hoc family and
provide no coverage between evaluated points.
The comparison assumes that the four additional memories in an encoded block
would not otherwise store independent astronomical modes.  Under the more
conservative memory-only \(1/n\) throughput bound
\(R_{\rm enc}/R_0\leq1/n=0.2\), this frozen \(L=2\) implementation batch does
not beat five parallel bare memories at any evaluated point
\(p=0.05,0.10,\ldots,0.50\), even when the bare strategy is granted
sufficient reference and detector throughput.
This is not a fixed-total-resource comparison.  Consequently, a practical advantage
requires either spare parallel memory, another capture-rate bottleneck, or a
code with nonzero asymptotic rate.  The simplified channel in
Eq.~\eqref{eq:hardware-dephasing-envelope} therefore defines a statistically
evaluated sensitivity map, not a calibrated experimental prediction.

\section{Platform-total-resource requirements and present scope}
\label{sec:platform-contract}

This section specifies the information required before a platform-total
comparison may be attempted; it is not itself a performance model or a
ranking of existing experiments.  The operational
boundary begins with an incident spatiotemporal mode and ends with the complete
classical record.  It therefore includes capture and storage, encoding and
idle periods, erasure-flag acquisition, branch selection, feed-forward and
recovery, reference generation and distribution, retrieval and readout,
detector dead time, reset, failed attempts, and occupied quantum and classical
resources.

\paragraph*{Common endpoint and comparator.}
For each strategy, \(Y_\Pi(0,T)\) in
Eq.~\eqref{eq:platform-information-rate} must retain successful, failed,
rejected, flagged, no-click, and timeout outcomes.  Capture, reference, flag,
and recovery failures are charged to the incident-mode and wall-clock
denominators rather than removed by conditioning.  Encoded and bare strategies
must share the incident-mode definition, observation window, collection
boundary, failed-branch rule, and total capacity vector \(\mathcal C\); the
comparator is the best allowed bare strategy under that same contract.

\paragraph*{Required inputs and provenance.}
The required inputs fall into three dependency classes: external platform or
calibration records, local protocol definitions, and quantities derived only
after the platform snapshot and native compilation are frozen.  Every
external record requires a typed value and unit, calibration time,
uncertainty, and traceable source.  The protocol definition must be fixed
before evaluation, and each derived quantity must identify its parent records
and compilation.  A simulated value or a qualitative literature description
cannot substitute for a measured or certified platform input.

\paragraph*{Eligibility for platform evaluation.}
A software model authorizes reproduction only of the stated synthetic
calculation; it is not experimental evidence or a hardware-feasibility
result.  A hardware implementation becomes eligible for a platform-total
calculation only after the target encoded block, calibrated channel and timing
records, local protocol, branch schedules, and native compilation are specified
under one consistent snapshot.

\paragraph*{Current scope.}
The networked SiV experiment supplies relevant signal-storage, erasure, and
nonlocal-heralding ingredients~\cite{Stas2026}, while the cold-atom experiment
stores the replaceable ancillary Fock-state reference used for GJC
interference~\cite{Wang2026}.  Because these demonstrations implement
different memory roles and neither is the calibrated \(n=5\),
\(U(1)\)-encoded signal-memory architecture analyzed here, they are cited as
experimental context rather than entered into a platform-total comparison.
No hardware platform is selected or ranked.

\paragraph*{Schedule, channel, and uncertainty closure.}
For a future hardware record, every encoded and bare branch must compile to an
acyclic native schedule with calibrated durations, explicit idle intervals,
exclusive and capacity resources, and parallel-operation and crosstalk
contexts.  Schedule-ordered capture, reference, encoder, memory, recovery,
readout, and detector maps must form a physical instrument: conditional
branches are completely positive and trace nonincreasing, and their sum is
trace preserving.  Calibration uncertainty, experimental sampling error,
Monte Carlo error, and design-ensemble uncertainty must be propagated and
reported separately.

\paragraph*{Interpretation.}
The synthetic calculation reproduces only the stated model.  It does not
evaluate the wall-clock fixed-resource endpoint
\(\mathcal J_\Pi(g;\mathcal C)\) in
Eq.~\eqref{eq:platform-information-rate}, establish a platform prediction, or
show hardware advantage.  A future implementation would require a separate,
prespecified platform-total analysis after the inputs above are frozen;
meeting those requirements would make the analysis well defined but would not
by itself establish advantage.  In the channel proofs, \(R\) denotes a
mathematical purifying reference, not the physical GJC phase-reference
subsystem \(a_Aa_B\), and therefore carries no hardware-resource entry.

\section{Detailed limitations and no-go statements}
\label{sec:limitations}

\subsection{Pre-storage loss}

Let $\mathcal{L}_{\eta_c}$ describe signal loss before the state reaches the encoded memory.  Since the encoder acts only after $\mathcal{L}_{\eta_c}$,
\begin{equation}
H[\mathcal{V}\circ\mathcal{L}_{\eta_c}(\rho_s(g))]
=H[\mathcal{L}_{\eta_c}(\rho_s(g))]
\leq H[\rho_s(g)].
\end{equation}
No subsequent parameter-independent unitary can restore QFI carried away before capture without access to the lost environment.  The factor $\eta_c$ in Eq.~\eqref{eq:stored-qfi} is therefore irreducible within the present architecture.

\subsection{Loss of an entire telescope node}

Local codes protect erasures inside each node.  If an entire node and all of its encoded block are lost, then one side of the cross-aperture coherence is traced out and the visibility cannot be reconstructed from the other node alone.  Protecting node failure would require nonlocal redundancy across additional telescope nodes, which is a different network-coding problem.

\subsection{Unflagged amplitude damping and dephasing}

The unrestricted exact-code bound in Eq.~\eqref{eq:exact-half-limit} and the
finite-size crossovers reported above concern known erasure locations.  For unflagged
amplitude damping, the relevant error set includes both no-jump and jump
operators; for dephasing it includes phase errors that directly suppress $g$.
Exact protection then requires a code correcting the complete channel, not
merely an erasure decoder.  Scrambling alone is not a recovery operation and
does not establish such correction.

\subsection{QFI without a decoder}

Scrambling may place parameter information in high-order correlations of the retained block.  The standard GJC receiver probes a very restricted family of local optical observables.  Without a decoder or an explicitly constructed logical GJC POVM, high retained QFI does not imply high observed CFI\@.  The exact-recovery benchmark in Corollary~\ref{cor:cfi} therefore depends crucially on recovery before readout.

\subsection{Asymptotic versus finite-sample readout}

The Fisher matrices in Eqs.~\eqref{eq:gjc-fisher} and~\eqref{eq:approx-gjc-fi} are local asymptotic quantities.  A phase design that improves their generalized eigenvalues need not improve a biased, constrained estimator at finite sample size.  Any future claim of adaptive advantage must therefore include estimator bias, physical-boundary effects, feedback latency, the cost of the pilot stage, and a training--validation protocol fixed before the final Monte Carlo run.

\section{Open problems beyond the finite-size flagged-erasure setting}
\label{sec:proof-programme}

Beyond the finite-size flagged-erasure model, the following
architecture-specific questions remain open.  They concern the additional
requirements of local finite-depth dynamics, explicit decoding, operational
GJC readout, and resource-accounted implementation; they are not a claim
that covariant approximate QEC itself lacks an existing theory.  Each item
requires either a theorem with explicit scaling and error bounds or a
platform-calibrated validation under a stated resource model.

\begin{enumerate}
  \item[\textbf{O1.}] \textbf{Physical signal-storage channel and imperfect
  erasure information.}  Construct a platform-specific CPTP instrument for
  capture, coherent storage, erasure heralding, and retrieval of the weak
  astronomical state.  The model must determine $\eta_c$ and $\lambda_c$ and
  either calibrate the flux and intensity imbalance
  \((\epsilon,\kappa)\) or estimate them jointly with \(g\) using the full
  outcome distribution or an additional photometric measurement.  It must
  include amplitude damping, phase noise, leakage, crosstalk, multiphoton
  corrections, and false-positive and false-negative erasure flags.  It must
  also separate irrecoverable pre-storage loss from errors acting on the
  encoded memory block.

  \item[\textbf{O2.}] \textbf{Scalable local recovery and platform
  implementation.}  For a scalable family of covariant encoders and flagged
  erasure sets, construct parameter-independent local decoders with explicit
  bounds on ancilla size, gate count, circuit depth, and classical
  feed-forward cost.  The deviation between the intended and compiled logical
  channels must then be bounded under calibrated gate, leakage, and readout
  errors.

  \item[\textbf{O3.}] \textbf{Finite-depth local covariant coding beyond the
  charge-Haar benchmark.}  Known symmetry-preserving Haar constructions
  already give asymptotically vanishing erasure error in specified
  fixed-logical-size regimes~\cite{Kong2022}.  For the particular local
  circuit and independent-loss ensemble used here, prove either the
  physically weighted pattern-average target in
  Eq.~\eqref{eq:finite-depth-average-question} or the high-probability
  worst-pattern target in Eq.~\eqref{eq:finite-depth-uniform-question}, with
  explicit finite-size constants and dependence on circuit depth, charge
  sector, geometry, and erased fraction.  The standard
  information--disturbance conversion must then be made quantitative for the
  resulting complementary-channel error and an explicit efficient decoder.
  Corollary~\ref{cor:gjc-fi-continuity} converts the final logical diamond
  error into a GJC-CFI error on compact visibility regions.
  Lemma~\ref{lem:causal-depth} supplies only a necessary lower bound on depth,
  not achievability.
  Any higher-rate extension must define its logical charge range and compare
  its accuracy with the continuous-symmetry lower bounds~\cite{Faist2020}, rather than assume vanishing full-state error at fixed
  positive rate.

  \item[\textbf{O4.}] \textbf{Resource-accounted local synthesis for the
  interferometric architecture.}  Develop encoder and decoder families with
  an explicit logical-rate--accuracy tradeoff that obey the photon-number
  superselection rule and admit scalable number-conserving circuit synthesis.
  The resource analysis must include preparation, distribution, phase
  stability, loss, and possible consumption of the finite reference state
  rather than treating reference-assisted accessibility as cost free.

  \item[\textbf{O5.}] \textbf{Finite-sample, resource-normalized system
  advantage.}  Existing SiV signal-storage and cold-atom ancillary-reference
  demonstrations establish complementary ingredients but instantiate
  different architectures~\cite{Stas2026,Wang2026}.  The present model-only
  calculation therefore does not evaluate \(\mathcal J_\Pi\) for either
  platform.

  The open task is to populate and evaluate the contract for one named
  end-to-end platform and a frozen calibration snapshot.  The analysis must
  combine calibrated channel and implementation errors with estimator bias,
  physical-boundary effects, pilot-stage cost, feedback latency, and a
  training--validation protocol fixed before the final simulation or
  experiment.  A system-level advantage would require a simultaneously valid
  positive lower bound against the best bare strategy at fixed incident
  modes, wall-clock time, physical memories, reference supply, classical
  capacity, and detector resources.
\end{enumerate}

These problems have a natural dependency order.  Problem O1 fixes the
physical channel and rate model.  Problems O3--O4 ask whether known
covariant-code principles can meet this architecture's local-depth, decoder,
and SSR-reference requirements, while O2 connects logical recovery to
hardware.  The synthetic model can establish computational consistency only; it cannot
establish a hardware result.  A separate, prospectively specified evaluation
of \(\mathcal J_\Pi\) for O5 becomes meaningful only after one end-to-end
platform and calibration snapshot satisfy the requirements above.  Satisfying
those requirements is necessary for evaluation, not evidence of advantage.

\section{Detailed derivation of the source QFI matrix}
\label{app:qfi}

For
\begin{equation}
\rho=\frac12(\mathbb{I}+\boldsymbol{r}\cdot\boldsymbol{\sigma}),
\end{equation}
the SLD $L_\mu$ solving
\begin{equation}
\partial_\mu\rho
=\frac12(L_\mu\rho+\rho L_\mu)
\end{equation}
may be written as
\begin{equation}
L_\mu=a_\mu\mathbb{I}+\boldsymbol{b}_\mu\cdot\boldsymbol{\sigma}.
\end{equation}
Using
\begin{equation}
(\boldsymbol{a}\cdot\boldsymbol{\sigma})(\boldsymbol{b}\cdot\boldsymbol{\sigma})
=(\boldsymbol{a}\cdot\boldsymbol{b})\mathbb{I}
+i(\boldsymbol{a}\times\boldsymbol{b})\cdot\boldsymbol{\sigma},
\end{equation}
comparison of the identity and Pauli components yields
\begin{equation}
a_\mu=-\frac{\boldsymbol{r}\cdot\partial_\mu\boldsymbol{r}}{1-\left\lvert \boldsymbol{r}\right\rvert{}^2},
\qquad
\boldsymbol{b}_\mu
=
\partial_\mu\boldsymbol{r}
+\frac{\boldsymbol{r}\cdot\partial_\mu\boldsymbol{r}}{1-\left\lvert \boldsymbol{r}\right\rvert{}^2}\boldsymbol{r}.
\end{equation}
The QFI matrix
\begin{equation}
H_{\mu\nu}=\frac12\operatorname{Tr}[\rho(L_\mu L_\nu+L_\nu L_\mu)]
\end{equation}
then becomes Eq.~\eqref{eq:qubit-qfi-general}.

For $\boldsymbol{r}=(g_R,-g_I,0)$,
\begin{equation}
\partial_{g_R}\boldsymbol{r}=(1,0,0),
\qquad
\partial_{g_I}\boldsymbol{r}=(0,-1,0),
\end{equation}
and
\begin{equation}
\boldsymbol{r}\cdot\partial_{g_R}\boldsymbol{r}=g_R,
\qquad
\boldsymbol{r}\cdot\partial_{g_I}\boldsymbol{r}=g_I.
\end{equation}
Therefore
\begin{align}
H_{RR}&=1+\frac{g_R^2}{1-\left\lvert g\right\rvert{}^2}
=\frac{1-g_I^2}{1-\left\lvert g\right\rvert{}^2},\\
H_{II}&=1+\frac{g_I^2}{1-\left\lvert g\right\rvert{}^2}
=\frac{1-g_R^2}{1-\left\lvert g\right\rvert{}^2},\\
H_{RI}&=\frac{g_Rg_I}{1-\left\lvert g\right\rvert{}^2},
\end{align}
which proves Eq.~\eqref{eq:qfi-single}.

For the unequal-intensity state in Eq.~\eqref{eq:unbalanced-source}, let
\(s_\kappa=\sqrt{\kappa(1-\kappa)}\).  Its Bloch vector is
\begin{equation}
\begin{aligned}
\boldsymbol r_\kappa
&=
\left(2s_\kappa g_R,-2s_\kappa g_I,2\kappa-1\right)^{\mathsf T},
\\
1-\lVert\boldsymbol r_\kappa\rVert^2
&=4s_\kappa^2(1-|g|^2).
\end{aligned}
\label{eq:unbalanced-bloch-vector}
\end{equation}
Substitution into Eq.~\eqref{eq:qubit-qfi-general} gives
\begin{align}
H_{\kappa\kappa}^{(1)}
&=\frac{1}{\kappa(1-\kappa)},
\\
H_{\kappa g_R}^{(1)}
&=H_{\kappa g_I}^{(1)}=0,
\\
H_{gg}^{(1)}
&=4\kappa(1-\kappa)H^{(1)}(g).
\end{align}
which proves Eq.~\eqref{eq:unbalanced-qfi-block}.  The vanishing cross block
is a property of the normalized one-photon family, not an assumption that
the intensity is known.

For the block-diagonal state
\begin{equation}
\rho=(1-\epsilon)\rho_0\oplus\epsilon\rho_1(g),
\end{equation}
with parameter-independent weights and $g$-independent $\rho_0$, the SLD is block diagonal and only the second block contributes.  Hence
\begin{equation}
H[\rho]=\epsilon H[\rho_1],
\end{equation}
which proves Eq.~\eqref{eq:qfi-weak-source} to the retained order.
If \(\epsilon\) is itself estimated, the two parameter-independent block
weights contribute the binary-distribution information
\(H_{\epsilon\epsilon}=1/[\epsilon(1-\epsilon)]\), while
\(H_{\epsilon\alpha}=0\) for
\(\alpha\in\{\kappa,g_R,g_I\}\), proving
Eq.~\eqref{eq:flux-qfi-block}.

\section{Detailed derivation of the GJC probabilities}
\label{app:gjc}

Within the sector containing one astronomical photon and one ancillary photon, the only terms that give one local photon at each node are
\begin{equation}
\left\lvert X\right\rangle{}=\left\lvert s_Aa_B\right\rangle{},
\qquad
\left\lvert Y\right\rangle{}=\left\lvert a_As_B\right\rangle{}.
\end{equation}
For the unequal-intensity signal state of
Eq.~\eqref{eq:unbalanced-source}, with \(g\) replaced by the stored
visibility \(g_c\), the product state has matrix elements
\begin{equation}
\begin{aligned}
\rho_{XX}&=\kappa q_B,
&
\rho_{YY}&=(1-\kappa)q_A,
\\
\rho_{XY}&=\sqrt{\kappa(1-\kappa)}\,g_c c^*.
\end{aligned}
\label{eq:general-gjc-xy-elements}
\end{equation}
The joint local-click state $\left\lvert r\right\rangle{}_A\left\lvert s\right\rangle{}_B$ satisfies
\begin{equation}
\left\langle r,s\middle\vert X\right\rangle{}=\frac{s}{2},
\qquad
\left\langle r,s\middle\vert Y\right\rangle{}=\frac{r}{2}.
\end{equation}
Hence, conditioned on the one-photon sectors,
\begin{align}
q_{rs}^{(1)}
&:=\left\langle r,s\right\rvert{}
\rho_{s,\kappa}^{(1)}(g_c)\otimes\rho_a^{(1)}
\left\lvert r,s\right\rangle{}
\nonumber\\
&=
\frac14\left(\rho_{XX}+\rho_{YY}\right)
+\frac{rs}{2}\operatorname{Re}(\rho_{XY})
\nonumber\\
&=
\frac14\left(A_\kappa+rsB_{\kappa,\delta}\right),
\label{eq:general-gjc-conditional-probability}
\end{align}
where Eqs.~\eqref{eq:mu-resource}, \eqref{eq:storage-jacobian}, and
\eqref{eq:unbalanced-gjc-ab} were used in the last line.
Multiplication by the captured astronomical one-photon probability
$\epsilon\eta_c$ and the ancillary availability $\eta_a$ proves
Eqs.~\eqref{eq:unbalanced-gjc-probability}--%
\eqref{eq:unbalanced-gjc-success}.  In the balanced ideal-storage case
\(\kappa=q_A=q_B=1/2\) and \(g_c=g\), this becomes
\[
q_{rs}^{(1)}
=
\frac18\left[1+rs\operatorname{Re}(g\mu^*)\right],
\]
which proves Eq.~\eqref{eq:gjc-probability}.

For parity \(z=rs\), there are two detector pairs per value of \(z\).
In the balanced case,
\begin{equation}
p_z=\frac{\epsilon\eta_c\eta_a}{4}
\left[1+z\nu_a\boldsymbol{u}_\delta^{\mathsf T}\boldsymbol{g}\right].
\end{equation}
The failure probability is parameter independent to this order.  Differentiating and summing
\begin{equation}
[F_\delta]_{\mu\nu}
=\sum_{z=\pm1}
\frac{(\partial_\mu p_z)(\partial_\nu p_z)}{p_z}
\end{equation}
gives Eq.~\eqref{eq:gjc-fisher}.  More generally,
\begin{equation}
\partial_{\boldsymbol g}B_{\kappa,\delta}
=\sqrt{\kappa(1-\kappa)}\,\nu_aJ_c^{\mathsf T}
\boldsymbol u_\delta.
\end{equation}
Summing the four successful outcomes gives
Eq.~\eqref{eq:unbalanced-gjc-fisher}.  Including the failure outcome before
forming the complete Fisher matrix yields the nuisance-parameter Schur
complement in Eq.~\eqref{eq:gjc-nuisance-schur}.

\section{Operator-basis proof of local recovery}
\label{app:operator-proof}

Equation~\eqref{eq:exact-local-recovery} need only be verified on the four matrix units
\begin{equation}
\left\lvert 0\right\rangle{}\!\left\langle 0\right\rvert{},\quad\left\lvert 0\right\rangle{}\!\left\langle 1\right\rvert{},\quad\left\lvert 1\right\rangle{}\!\left\langle 0\right\rvert{},\quad\left\lvert 1\right\rangle{}\!\left\langle 1\right\rvert{}.
\end{equation}
The diagonal terms guarantee recovery of local occupation probabilities, whereas the off-diagonal terms guarantee recovery of phase coherence with an arbitrary external reference.  In particular, preservation of populations alone is insufficient for astronomy: a code that maps both off-diagonal matrix units to zero would preserve the local intensity but destroy the complex visibility.

The cross-aperture operator factorizes as
\begin{equation}
\left\lvert 10\right\rangle{}\!\left\langle 01\right\rvert{}_{AB}
=\left\lvert 1\right\rangle{}\!\left\langle 0\right\rvert{}_A\otimes\left\lvert 0\right\rangle{}\!\left\langle 1\right\rvert{}_B.
\end{equation}
Local recovery of both off-diagonal matrix units therefore restores the global astronomical coherence without requiring a nonlocal decoding gate.

\end{document}